\documentclass{article}
\usepackage{graphicx} 
\usepackage{amsmath, amssymb, placeins}
\usepackage{amsthm}
\usepackage{bbm}
\usepackage{xcolor}
\usepackage{prodint}
\usepackage{cite}
\usepackage[a4paper, total={7in, 10in}]{geometry}

\title{Covariance Correction for Permutation Statistics in Multiple Testing Problems}
\newtheorem{example}{Example}
\newtheorem{ass}{Assumption}
\newtheorem{rem}{Remark}
\newtheorem{theorem}{Theorem}
\newtheorem{defi}{Definition}

\newcommand{\E}{\mathrm{E}}
\newcommand{\Cov}{\mathrm{Cov}}

\begin{document}
\author{Merle Munko$^{1}$ and Paavo Sattler$^{2,3}$}
\date{ }
\maketitle
\begin{center}\noindent${}^{1}$ {Department of Mathematics, Otto von Guericke University Magdeburg, Universitätsplatz 2, 39106 Magdeburg, Germany\\\mbox{ }\hspace{1 ex}email: merle.munko@ovgu.de} \\
 \noindent${}^{2}$ {Department of Statistics, TU Dortmund University, Joseph-von-Fraunhofer-Straße 2-4, 44227 Dortmund, Germany\\\mbox{ }\hspace{1 ex}email: paavo.sattler@tu-dortmund.de }\\
  \noindent${}^{3}$ {Institute of Statistics, RWTH Aachen University, 
Kreuzherrenstraße 2,
52062 Aachen, Germany}

\vspace*{8.25ex}
 \end{center}


\begin{abstract}
In qualitative statistics, permutation tests are very popular, mainly because of their finite-sample exactness under exchangeability. 
However, in non-exchangeable settings, the covariance structure of permuted statistics typically differs from that of the original statistic. A common solution is studentization, which restores asymptotic correctness for general hypotheses while preserving exactness under exchangeability.
In multiple testing settings, however, standard studentization fails to provide the correct joint limiting distribution.
Existing solutions such as prepivoting address this issue but are computationally expensive and therefore rarely used in practice.
We propose a general, computationally more efficient methodology that overcomes this fundamental limitation. 
By appropriately correcting the covariance matrix of multiple permutation statistics, our approach restores the correct joint asymptotic dependence structure, enabling asymptotically valid permutation tests in broad multiple testing frameworks.
The proposed method is highly flexible: it accommodates singular covariance structures and is not tied to specific parameters, test statistics, or permutation schemes. This generality makes it applicable across a wide range of problems. Extensive simulation studies demonstrate that our approach results in reliable inference and outperforms existing methods across diverse settings.
\end{abstract}

\section{Introduction}
Permutation tests are a central tool in nonparametric statistical inference and are widely applied across a broad range of settings and test statistics. 
A key advantage of permutation tests is that they avoid parametric distributional assumptions. Moreover, under exchangeability, they are finitely exact \cite{hemerik2018exact}. Consequently, a broad literature work has developed permutation tests under the exchangeability assumption; see, for example, the textbooks by \cite{manly2018randomization, pesarin2001multivariate, westfall1993resampling}.

In practice, however, exchangeability is often violated. In such settings, applying unstudentized permutation tests can lead to inflated type-I error rates \cite{Huang2006}. To address this issue, \cite{chung2013, janssen1997studentized, janssen2005resampling, neuhaus1993conditional} introduced studentized permutation tests that remain asymptotically valid even under non-exchangeable data.
Since then, studentization has become a key principle for permutation tests in non-exchangeable settings and the idea was taken up by several authors across diverse statistical problems; see, for example, \cite{ditzhaus2021inferring, ditzhaus2021qanova, ditzhaus2020more, ditzhaus2023casanova, ditzhaus2020bootstrap, ditzhausVC, ditzhaus2023studentized, dobler2024nonparametric, dobler2024two, dobler2018bootstrap, friedrich2015permuting, konietschke2012studentized, munko2024rmst, munko2025multiple, munko2025effect, neubert2007studentized, pauly2016permutation, pauly2015asymptotic, wu2021randomization}.

In multiple testing problems, however, studentization alone is typically insufficient to resolve discrepancies in joint limiting distributions. While it corrects the variance for marginal distributions, it does not account for the dependence structure between permutation statistics. As a consequence, the joint behavior of studentized permutation statistics can still be misspecified. This issue arises even in simple settings; for instance, Tukey-type contrasts \cite{tukey} for all pairwise comparisons of 
$k>2$ means generally induce incorrect dependence structures of the permutation statistics.
As a result, when an accurate approximation of the joint limit distribution is required, authors often resort to alternative resampling methods, such as different kinds of bootstrapping; see, e.g., \cite{ditzhausVC, munko2024rmst, munko2025multiple, Munko2025AStA, munko2025, romano2005exact}. However, unlike permutation tests, bootstrap methods do not possess the appealing finite-sample properties.
Therefore, 
the prepivoting idea, introduced by \cite{beran1988balanced,beran1988prepivoting}, was employed by \cite{chung2016multivariate} to overcome the problem of differing joint limit distributions.
However, it comes at a substantial computational cost: its implementation involves bootstrap distributions within each permutation sample, resulting in nested resampling loops. This computational burden makes the method impractical in many applications and likely explains its limited use in practice.

The contribution of this paper is the development of a novel permutation procedure that overcomes these limitations.
Our approach avoids nested resampling and is therefore computationally more efficient than the method in \cite{chung2016multivariate}, while restoring the correct joint limit distribution.

The proposed procedure is highly general and applies to a wide range of settings and parameters, including:
\begin{itemize}
    \item location parameters such as the mean \cite{Baumeister2025}, the Mann-Whitney effect \cite{thiel2024,rubarth2022}, and quantiles \cite{baumeister2024, baumeister2025early},
    \item parameters from survival analysis such as the restricted mean survival time \cite{munko2024rmst,munko2025effect} and the restricted mean time lost \cite{munko2025multiple},
    \item  dispersion parameters such as vectorized covariance matrices \cite{sattler2022} or correlation matrices \cite{sattler2023},
\end{itemize}
and many others.
Moreover, the method is not restricted to specific test statistics. 
Beyond univariate linear statistics, it also allows for quadratic form statistics and one-sided or asymmetric formulations.
We exemplify the flexibility of the proposed method by considering different multiple testing problems, including Dunnett-, grand-mean-, and Tukey-type contrasts of group means as well as contrasts of restricted mean survival times in a one-way layout using quadratic-form–based test statistics.

The remainder of this paper is organized as follows:
In Section~\ref{sec:methods}, we present a very general setup for multiple hypotheses, introduce some standard assumptions like asymptotic normality, and propose possible corrections for a potentially wrong covariance matrix of the permutation test statistics under different assumptions.
The general methodology is exemplified in several examples throughout Section~\ref{sec:methods}.
Furthermore, specific examples are presented in detail in Section~\ref{sec:examples}. This includes applications for uni- and multivariate mean comparisons as well as an application from survival analysis.
In addition, we present results of extensive simulation studies for the three different examples.
Finally, the methodology is discussed in Section~\ref{sec:discussion}.
The appendix includes a section on finite-sample properties of the proposed methods (Appendix~\ref{sec:finite}), the proof of the stated theorems (Appendix~\ref{sec:proofs}), and further simulation results (Appendix~\ref{sec:simu}).

\section{Methodology}\label{sec:methods}
In Section~\ref{ssec:notation}, we firstly introduce the notation used in this work.
This includes the multiple testing problem, the (permutation) quantities, and the assumptions.
Under the stated assumption, the naive permutation test statistics will generally be not consistent.
Thus, we continue by considering three different sets of assumptions under which we will derive consistent multiple permutation tests in Sections~\ref{ssec:fullrank}--\ref{ssec:convrate}.

\subsection{Notation}\label{ssec:notation}

\paragraph{Testing problem}
We start by formulating the multiple testing problem.
Let $\boldsymbol{\theta} = (\boldsymbol{\theta}_1',...,\boldsymbol{\theta}_L')'\in\mathbb R^r$ be the parameter of interest for $\boldsymbol{\theta}_\ell \in \mathbb R^{r_\ell}, \ell\in\{1,...,L\},$ where $L\in\mathbb N$ is the number of hypotheses, $r_1,...,r_L \in\mathbb N$, and $r := \sum_{\ell=1}^L r_\ell$.
Then, we consider multiple tests for the null hypotheses
\begin{align}
    \label{eq:testprob}
    \mathcal H_{0,\ell}: \boldsymbol{\theta}_\ell = \boldsymbol{\theta}_{0,\ell}, \quad \ell\in\{1,...,L\},
\end{align}
 and set $\boldsymbol{\theta}_0 = (\boldsymbol{\theta}_{0,1}',...,\boldsymbol{\theta}_{0,L}')' \in\mathbb R^r$ for some $L \in\mathbb N$.
Of course, this also covers global testing problems with $L=1$ hypothesis as a special case, but is even more general.

Suppose we have an estimator $\widehat{\boldsymbol{\theta}} = (\widehat{\boldsymbol{\theta}}_1',...,\widehat{\boldsymbol{\theta}}_L')'$ for $\boldsymbol{\theta}$ fulfilling a limit theorem under $\bigcap_{\ell\in\mathcal T}\mathcal{H}_{0,\ell}$, that is
\begin{align}\label{eq:asnormal}
    \sqrt{n}(\widehat{\boldsymbol{\theta}}_\ell  - {\boldsymbol{\theta}}_{0,\ell})_{\ell\in\mathcal T} \xrightarrow{d} (\mathbf Z_\ell)_{\ell\in\mathcal T}
\end{align}
as $n\to\infty$ for all $\mathcal T \subset \{1,...,L\}$, where $n$ is the number of data points and $\mathbf Z =(\mathbf Z_1',...,\mathbf Z_L')' \sim \mathcal N_r(\mathbf 0, \boldsymbol{\Sigma})$ is a centered $r$-variate normally distributed random variable.
Here and throughout, $\mathbf{0}=(0,\ldots,0)'$ denotes the zero vector.
Furthermore, let $\widehat{\boldsymbol{\Gamma}}$ be a consistent estimator for some constant ${\boldsymbol{\Gamma}}$ in a metric space $\mathbb G$.
Here, ${\boldsymbol{\Gamma}}$ can be interpreted as nuisance parameters that can be estimated consistently from the data as, e.g., (co-)variances etc.

We consider test statistics of the following (very general) form
\begin{align*}
   W_\ell = w_\ell(\sqrt{n}(\widehat{\boldsymbol{\theta}}_\ell  - {\boldsymbol{\theta}}_{0,\ell}), \widehat{\boldsymbol{\Gamma}}), 
\end{align*}
where $w_\ell: \mathbb R^{r_\ell} \times \mathbb G \to \mathbb R$ is a function that is continuous in $(\mathbf{x}, \boldsymbol{\Gamma})$ for all $\mathbf x \in \mathrm{Im}(\Cov(\mathbf{Z}_\ell)) \subset \mathbb R^{r_\ell}$ with $\mathrm{Im}(\Cov(\mathbf{Z}_\ell))$ denoting the image of $\Cov(\mathbf{Z}_\ell)$.
It should be noted that the continuity is required for several first arguments $\mathbf x$ whereas it only needs to hold in the constant $\boldsymbol{\Gamma}$
for the second argument. This asymmetry results from the different convergence statements: $\sqrt{n}(\widehat{\boldsymbol{\theta}}_\ell  - {\boldsymbol{\theta}}_{0,\ell})$ is asymptotically normal, i.e., it converges to a distribution living on the space $\mathrm{Im}(\Cov(\mathbf{Z}_\ell))$. In contrast,  $\widehat{\boldsymbol{\Gamma}}$ converges to the constant $\boldsymbol{\Gamma}$ in probability.
Quadratic form-based test statistics and multiple contrast tests (MCTs) are an important special case:

\begin{example}\label{QFexample}
In order to illustrate the rather abstract notation, let us consider $k$ samples of i.i.d.~$d$-variate data points $$\mathbf{X}_{i1},...,\mathbf{X}_{in_i}, \quad i\in\{1,...,k\},$$ with means $\boldsymbol{\mu}_i := \E(\mathbf{X}_{i1})$ and existing covariance matrices $\Cov(\mathbf{X}_{i1})$. 
Moreover, we assume that $n_i/n \to \kappa_i \in (0,1)$ as $n := \sum_{i=1}^k n_i\to\infty$.

Now, we might be interested in comparing the mean vectors $\boldsymbol{\mu}_i, i\in\{1,...,k\}$, of the different groups. This can be done with different parameters of interest:
If there is a reference group, e.g.\ group 1, we would choose the differences $\boldsymbol{\mu}_{\ell+1} - \boldsymbol{\mu}_1, \ell\in\{1,...,k-1\}, $ as parameters $\boldsymbol{\theta}_\ell$.
Otherwise, we can either consider the differences $\boldsymbol{\theta}_\ell = \boldsymbol{\mu}_\ell - \overline{\boldsymbol{\mu}}, \ell\in\{1,...,k\},$ to the average mean $\overline{\boldsymbol{\mu}} = \frac{1}{k}\sum_{i=1}^k \boldsymbol{\mu}_i$ or all pairwise differences $\boldsymbol{\mu}_{\ell_1} - {\boldsymbol{\mu}}_{\ell_2}, \ell_1,\ell_2\in\{1,...,k\}, \ell_1 > \ell_2$.

We can generalize this by writing $\boldsymbol \theta_\ell := \mathbf H_\ell \boldsymbol\mu$ for matrices $\mathbf H_\ell \in \mathbb R^{r_\ell\times kd}$ with $\boldsymbol{\mu} := (\boldsymbol{\mu}_1',...,\boldsymbol{\mu}_k')'$. 
Suitable estimators are $\widehat{\boldsymbol \theta}_\ell :=  \mathbf H_\ell \widehat{\boldsymbol\mu}$, where $\widehat{\boldsymbol\mu}:= (\widehat{\boldsymbol{\mu}}_1',...,\widehat{\boldsymbol{\mu}}_k')'$ is a vector containing all empirical means $$\widehat{\boldsymbol{\mu}}_i := \frac{1}{n_i}\sum_{j=1}^{n_i} \mathbf{X}_{ij}, \quad i\in\{1,...,k\}.$$ The central limit theorem ensures the asymptotic normality \eqref{eq:asnormal} with $\boldsymbol{\Sigma} := \mathbf H  \boldsymbol{\Gamma}   \mathbf H'$, where $\mathbf H := ( \mathbf{H}_1', ... \mathbf{H}_L')'$ and $\boldsymbol{{\Gamma}} := \oplus_{i=1}^k \kappa_i^{-1}  \Cov(\mathbf X_{i1}) \in \mathbb G$ and $\mathbb G$ denotes the set of all symmetric and positive semi-definite $kd \times kd$ matrices. 

Suitable test statistics for multivariate quantities are so-called \emph{quadratic form-based test statistics} \cite{sattler2025quadraticformbasedmultiple} $$W_\ell = n(\widehat{\boldsymbol{\theta}}_\ell - \boldsymbol{\theta}_{0,\ell})' \cdot  {\mathbf{M}}_\ell(\widehat{\boldsymbol{\Gamma}})  \cdot(\widehat{\boldsymbol{\theta}}_\ell - \boldsymbol{\theta}_{0,\ell}), \quad \ell\in\{1,...,L\},$$ for functions $$\mathbf{M}_\ell: \mathbb G \to \{\mathbf M\in \mathbb R^{r_\ell\times r_\ell} \mid \mathbf M \text{ is symmetric and positive semi-definite}\}$$ that are continuous in $\boldsymbol{\Gamma}$.
Hence, let $w_\ell(\mathbf{x}, \mathbf{G}) = \mathbf{x}' \mathbf{M}_\ell(\mathbf{G}) \mathbf{x}$.
Then, $W_\ell$ are classical quadratic form-based test statistics.
More concrete, one can choose 
\begin{itemize}
    \item $\mathbf{M}_\ell(\mathbf{G}) := (\mathbf H_\ell \mathbf{G} \mathbf H_\ell')^+$ for the \emph{Wald-type test statistic}, where here and throughout $\mathbf A^+$ denotes the Moore-Penrose inverse of a matrix $\mathbf A$. To ensure the continuity of $\mathbf{M}_\ell$ in $\boldsymbol{\Gamma}$, we further assume that $\Cov(\mathbf{X}_{11}),...,\Cov(\mathbf{X}_{k1})$ are of full rank.
    \item $\mathbf{M}_\ell(\mathbf{G}) := \mathrm{tr}(\mathbf H_\ell \mathbf{G} \mathbf H_\ell')^{-1} \mathbf I$ for the \emph{ANOVA-type test statistic}, where here and throughout $\mathrm{tr}(\mathbf{A})$ denotes the trace of a quadratic matrix $\mathbf{A}$ and $\mathbf{I}$ the identity matrix. To ensure the continuity of $\mathbf{M}_\ell$ in $\boldsymbol{\Gamma}$, we further assume that $\mathrm{tr}(\mathbf H_1 \boldsymbol{\Gamma} \mathbf H_1'),...,\mathrm{tr}(\mathbf H_L \boldsymbol{\Gamma} \mathbf H_L') > 0$.
\end{itemize}
In the case of univariate parameters of interest, i.e., $r_1 = ... = r_L =1$, classical MCTs are covered together with the test statistic
$ W_\ell =
    \sqrt{n} (\widehat{\boldsymbol{\theta}}_\ell - \boldsymbol{\theta}_{0,\ell})  \cdot {\mathbf{M}}_\ell(\widehat{\boldsymbol{\Gamma}}), \ell\in\{1,...,L\},$
    or $ W_\ell =
    \sqrt{n} |\widehat{\boldsymbol{\theta}}_\ell - \boldsymbol{\theta}_{0,\ell}| \cdot  {\mathbf{M}}_\ell(\widehat{\boldsymbol{\Gamma}}), \ell\in\{1,...,L\},$
where $$\mathbf{M}_\ell: \mathbb G \to [0,\infty),\quad \mathbf{M}_\ell(\mathbf{G}) := \begin{cases}
   1/ \sqrt{\mathbf H_\ell \mathbf{G} \mathbf H_\ell'} & \text{ if } \mathbf H_\ell \mathbf{G} \mathbf H_\ell' > 0\\
   0 & \text{else.}
\end{cases}$$
\end{example}

\paragraph{Permutation}
Instead of working with the limit distribution of $\mathbf W$ for testing the multiple hypotheses, using resampling procedures to approximate this distribution is often beneficial in terms of small sample performance.
One popular resampling procedure is the random permutation approach, which even provides finite sample guarantees under specific conditions.
However, for the permutation version of $\widehat{\boldsymbol{\theta}}$, we usually have the problem that the limit distribution differs from the distribution of $\mathbf Z$.
Let us assume that, for a sequence $\widehat{\boldsymbol{\theta}}^\pi$ of resampling test statistics, it holds
\begin{align}\label{eq:permconv}
    \sqrt{n}\, \widehat{\boldsymbol{\theta}}^\pi  \xrightarrow{d^*} \mathbf{Z}^\pi \sim \mathcal{N}_r(\mathbf{0}, \boldsymbol{\Sigma}^\pi),
\end{align}
where here and throughout $\xrightarrow{d^*}$ denotes conditional weak convergence in outer probability, see \cite{vaartWellner2023} for details.
We explicitly allow $\boldsymbol{\Sigma}^\pi \neq \boldsymbol{\Sigma}$, as this will result for many statistics by permutation limit theorems, see, e.g., \cite{ditzhausVC,munko2024rmst,munko2025multiple} for some applications. 
A very general way to prove \eqref{eq:permconv} in an empirical process setup is given in Corollary~1 of \cite{munko2024conditionaldeltamethodresamplingempirical}. Here, it is shown that uniform Hadamard differentiable functionals of permutation empirical processes converge weakly in outer probability to centered Gaussian limits, which generally have a different covariance structure than the limits of the functional of the original empirical processes, see \cite{munko2024conditionaldeltamethodresamplingempirical} for details. 
Similar observations can be made for other resampling approaches, e.g., the randomization approach by algebraic groups, see Remark~2 in \cite{dobler2023randomized}, and the pooled bootstrap approach, see Corollary~2 in \cite{munko2024conditionaldeltamethodresamplingempirical}.

Moreover, we assume that there is a consistent permutation estimator of the nuisance parameters, that is $\widehat{\boldsymbol{\Gamma}}^\pi \xrightarrow{P} \boldsymbol{\Gamma}$ as $n\to\infty$.
Of course, we can always choose $\widehat{\boldsymbol{\Gamma}}^\pi = \widehat{\boldsymbol{\Gamma}}$, but we do not restrict to this case.

\addtocounter{example}{-1}
\begin{example}[continued]
In Example~\ref{QFexample}, we define the permuted data points $$(\mathbf{X}_{11}^\pi,...,\mathbf{X}_{1n_1}^\pi,\mathbf{X}_{21}^\pi,..., \mathbf{X}_{kn_k}^\pi) := (\mathbf{Y}_{\pi_1},...,\mathbf{Y}_{\pi_n}),$$ where $(\mathbf{Y}_1,...,\mathbf{Y}_n) := (\mathbf{X}_{11},..., \mathbf{X}_{kn_k})$ denotes the pooled sample and $(\pi_1,...,\pi_n)$ denotes a random permutation of $\{1,...,n\}$ independent of the data.
Next, we define the permutation counterpart of the mean vector by $\widehat{\boldsymbol{\mu}}^\pi := (\widehat{\boldsymbol{\mu}}_1^{\pi\prime},...,\widehat{\boldsymbol{\mu}}^{\pi\prime}_k)^\prime $ with $\widehat{\boldsymbol{\mu}}^{\pi}_i := \frac{1}{n_i}\sum_{j=1}^{n_i} \mathbf{X}_{ij}^\pi $ and set $\widehat{\boldsymbol{\theta}}^\pi := \mathbf H \widehat{\boldsymbol{\mu}}^\pi$.
Furthermore, we assume that $\mathbf H$ fulfils the contrast property
$\mathbf H (\mathbf 1 \otimes \mathbf I) = \mathbf{0}_{r \times d}$, where here and throughout $\mathbf 1 = (1,...,1)'\in\mathbb R^k$ denotes the vector of ones, $\otimes$ denotes the Kronecker product and $\mathbf{0}_{r \times d} \in \mathbb R^{r\times d}$ the zero matrix.
Then,
Theorem~1 in \cite{munko2024conditionaldeltamethodresamplingempirical} ensures that  
$\sqrt{n}\, \widehat{\boldsymbol{\theta}}^\pi$ converges weakly in outer probability to a centered normal variable. The covariance matrix of the limit can be calculated as
$\boldsymbol{\Sigma}^\pi := \mathbf H  \boldsymbol{\Gamma}^\pi   \mathbf H' $ with $$\boldsymbol{\Gamma}^\pi  := \bigoplus_{i=1}^k \left(\kappa_i^{-1}\sum_{j=1}^k \kappa_j \Cov(\mathbf{X}_{j1})\right).$$
Thus, we observe that the covariance matrices are generally unequal $\boldsymbol{\Sigma}^\pi \neq \boldsymbol{\Sigma}$, as mentioned above.
\end{example}

Actually, we do not really need that $\widehat{\boldsymbol{\theta}}^\pi$ results from a permutation approach, as we will only use \eqref{eq:permconv} in our proofs.
Hence, the methodology of this section has a very broad range of applications and can also be applied to statistics resulting from other resampling procedures, such as, e.g., pooled bootstrap approaches.
However, other resampling procedures often yield straightforward covariance corrections.
For instance, in Example~\ref{QFexample}, the covariance matrix of the limit of 
$\sqrt{n}(\widehat{\boldsymbol{\mu}}^\pi - \widehat{\boldsymbol{\mu}}_\bullet)$ is always singular since $\left((n_1/n,...,n_k/n)\otimes \mathbf{e}'_j\right)\cdot \sqrt{n}(\widehat{\boldsymbol{\mu}}^\pi - \widehat{\boldsymbol{\mu}}_\bullet) = 0$ for all $j\in\{1,...,d\}$, where $$\widehat{\boldsymbol{\mu}}_\bullet = \mathbf{1}\otimes\left( \sum_{i=1}^k \frac{n_i}{n} \widehat{\boldsymbol{\mu}}_i \right)$$  denotes the pooled mean vector and $\mathbf{e}_j \in \mathbb{R}^d$ denotes the $j$th unit vector.
In contrast, the covariance matrix of the limit of
a standardized pooled bootstrap counterpart 
would be invertible under rather weak assumptions and, thus, can be corrected in a similar way to \cite{ditzhausVC}.
That is why we focus on permutation approaches in this work.
Additionally, permutation approaches often have the advantage of finite-sample properties, which we will study for our setting in more detail in Appendix~\ref{sec:finite}. 

\paragraph{Assumptions and consequences}
Another requirement that we need is that we have consistent estimators for both covariance matrices $\boldsymbol{\Sigma}$ and $\boldsymbol{\Sigma}^\pi$.
So far, we can summarize the stated assumptions as follows:
\begin{ass}\label{the_ass}
We assume 
\begin{enumerate}
    \item \textbf{Joint Asymptotic Normality of the Estimator}: As $n\to\infty$,
$ \sqrt{n}(\widehat{\boldsymbol{\theta}}_\ell  - {\boldsymbol{\theta}}_{0,\ell})_{\ell\in\mathcal T} \xrightarrow{d} (\mathbf Z_\ell)_{\ell\in\mathcal T} $
 under $\bigcap_{\ell\in\mathcal T}\mathcal{H}_{0,\ell}$ for all $\mathcal T \subset \{1,...,L\}$, where $\mathbf Z \sim \mathcal N_r(\mathbf 0, \boldsymbol{\Sigma})$.
\item \textbf{Joint Asymptotic Normality of the Permutation Estimator}: As $n\to\infty$,
$  \sqrt{n}\, \widehat{\boldsymbol{\theta}}^\pi  \xrightarrow{d^*} \mathbf{Z}^\pi \sim \mathcal{N}_r(\mathbf{0}, \boldsymbol{\Sigma}^\pi). $
\item \textbf{Consistent Estimators for the Covariance Matrices}: We have estimators $\widehat{\boldsymbol{\Sigma}}, \widehat{\boldsymbol{\Sigma}}^\pi$ taking values in the space of all positive semi-definite matrices and fulfilling $\widehat{\boldsymbol{\Sigma}} \xrightarrow{P} \boldsymbol{\Sigma}, \widehat{\boldsymbol{\Sigma}}^\pi \xrightarrow{P} {\boldsymbol{\Sigma}}^\pi$ as $n\to\infty$.
\item \textbf{Consistent Estimators for $\boldsymbol{\Gamma}$}: We have estimators $\widehat{\boldsymbol{\Gamma}}, \widehat{\boldsymbol{\Gamma}}^\pi$ taking values in a metric space $\mathbb G$ and fulfilling $\widehat{\boldsymbol{\Gamma}} \xrightarrow{P} \boldsymbol{\Gamma}, \widehat{\boldsymbol{\Gamma}}^\pi \xrightarrow{P} {\boldsymbol{\Gamma}}$ as $n\to\infty$ for some constant $\boldsymbol{\Gamma} \in\mathbb G$.
\item \textbf{General Form of the Test Statistics}: The test statistics can be written as $ W_\ell = w_\ell(\sqrt{n}(\widehat{\boldsymbol{\theta}}_\ell  - {\boldsymbol{\theta}}_{0,\ell}), \widehat{\boldsymbol{\Gamma}})$, $\ell\in\{1,...,L\}, $
where $w_\ell: \mathbb R^{r_\ell}\times \mathbb G \to \mathbb R$ is a function that is continuous in $(\mathbf{x}, \boldsymbol{\Gamma})$ for all $\mathbf x\in\mathbb {\emph{Im}}(\Cov(\mathbf{Z}_\ell))$. Moreover, we denote $w: \mathbb R^r \times \mathbb G \to \mathbb R^L, w((\mathbf x_\ell)_{\ell\in\{1,...,L\}}, \mathbf{G}) := \left( w_\ell(\mathbf x_\ell, \mathbf{G}) \right)_{\ell\in\{1,...,L\}}$ and $$\mathbf{W} := w(\sqrt{n}(\widehat{\boldsymbol{\theta}}  - {\boldsymbol{\theta}}_{0}), \widehat{\boldsymbol{\Gamma}}) = (W_\ell)_{\ell \in\{1,...,L\}} .$$
\end{enumerate}
\end{ass}

Under the stated assumptions and under the global null hypothesis $\bigcap_{\ell=1}^L \mathcal H_{0,\ell}$, it holds 
\begin{align}\label{eq:Wconv}
    \mathbf{W} \xrightarrow{d} w(\mathbf{Z}, {\boldsymbol{\Gamma}})
\end{align}
 as $n\to\infty$ by the continuous mapping theorem.
Of course, we can not guarantee that the naive permutation statistics $w(\sqrt{n}\, \widehat{\boldsymbol{\theta}}^\pi, \widehat{\boldsymbol{\Gamma}}^\pi)$ are converging jointly to the same limit $w(\mathbf{Z},\boldsymbol{\Gamma})$ due to potentially different covariance matrices $\boldsymbol{\Sigma}^\pi \neq \boldsymbol{\Sigma}$.
In fact, the naive permutation is only consistent whenever $w(\mathbf{Z},\boldsymbol{\Gamma}) \overset{d}{=} w(\mathbf{Z}^\pi,\boldsymbol{\Gamma})$.
Then, one can use, e.g., the methodology of Section~2.3 in \cite{munko2025inference} or of Section~4.1 in \cite{romano2005exact} for family-wise error rate (FWER) controlling testing procedures, simultaneous confidence regions, and adjusted p-values for the null hypotheses \eqref{eq:testprob}.
However, if $w(\mathbf{Z},\boldsymbol{\Gamma}) \overset{d}{=} w(\mathbf{Z}^\pi,\boldsymbol{\Gamma})$ does not hold as, e.g., in Example~\ref{QFexample}, inference procedures that rely on the consistency of the permutation approach are not applicable any more.
In the following, we aim to correct the possibly wrong limit distribution of the permutation statistics $w(\mathbf{Z}^\pi,\boldsymbol{\Gamma})$.

Due to the previous assumptions, the problem boils down to finding a transformation of the random variable $\sqrt{n}\, \widehat{\boldsymbol{\theta}}^\pi $ 
such that the conditional weak limit is $\mathcal N_r (\mathbf 0,\boldsymbol{\Sigma})$ instead of $\mathcal N_r (\mathbf 0,\boldsymbol{\Sigma}^\pi)$.
Since we assumed both limit distributions to be Gaussian, it remains to correct the covariance structure of the permutation test statistics, see e.g. \cite{ditzhausVC} for a similar approach for the pooled bootstrap. 
If $\boldsymbol{\Sigma}^{\pi}$ is of full rank, we can simply correct the different covariance structure of $\mathbf{Z}^\pi$ by multiplying the term $\boldsymbol{\Sigma}^{1/2} (\boldsymbol{\Sigma}^{\pi})^{-1/2}$ in front.
In practice, where we do not know the true covariance matrices, we would use $\widehat{\boldsymbol{\Sigma}}^{1/2} ((\widehat{\boldsymbol{\Sigma}}^{\pi})^{1/2})^+$ instead. This case is considered in more detail in Section~\ref{ssec:fullrank}.

However, in many applications, it turns out that $\boldsymbol{\Sigma}^{\pi}$ can not be of full rank due to the nature of dependencies arising from permutation approaches.
Then, the correction is not that straightforward.
We aim to use the eigenvalue decompositions of the covariance estimators to obtain a valid correction of the covariance structure.
Here, we need to be extra careful as the map from a symmetric, positive semi-definite matrix to the orthogonal matrix of its eigenvalue decomposition is generally not globally continuous.
To cover this issue, we consider two different approaches:
The first one in Section~\ref{ssec:distinct} requires the additional assumption that all non-zero eigenvalues are distinct. Under this assumption, the map to the orthogonal projectors onto the eigenspaces is continuous, and we can prove the asymptotic validity of our approach by the continuous mapping theorem if we interpose a diagonal matrix that randomly switches the signs of the eigenvectors.
The second approach in Section~\ref{ssec:convrate} requires the additional assumption that $\widehat{\boldsymbol{\Sigma}}$ converges with a known rate $r_n$ to $\boldsymbol{\Sigma}$. Then, all eigenvectors corresponding to eigenvalues that are 'too close' in terms of this convergence rate are additionally multiplied by a random orthogonal matrix to overcome possible confounding dependencies.

Before delving into the different approaches in more detail, we first want to remark that an alternative perspective on the hypotheses and parameters can sometimes suffice, such that at least one set of assumptions is satisfied.

\begin{rem}[Different hypotheses and parameters]\label{Extension}
 The hypotheses do not necessarily need to be chosen as
$\mathcal H_{0,\ell}: \boldsymbol{\theta}_\ell = \boldsymbol{\theta}_{0,\ell}, \ \ell\in\{1,\dots,L\}$.
In contrast, for a parameter $\boldsymbol{\eta}\in\mathbb R^q$, we can also consider
$\mathcal H_{0,\ell}: \boldsymbol{\eta} \in \mathbf H_{0,\ell}, \ \ell\in\{1,\dots,L\},$
for subsets $\mathbf H_{0,1},\dots,\mathbf H_{0,L} \subset \mathbb R^{q}$ and $q\in\mathbb N$.
We assume that there is an asymptotically normal estimator fulfilling
\[
\sqrt{n}(\widehat{\boldsymbol{\eta}}  - {\boldsymbol{\eta}})
\xrightarrow{d}
\mathcal{N}_q(\mathbf 0, \boldsymbol{\Xi})
\quad \text{as } n\to\infty
\] as well as an asymptotically normal permutation estimator
\[
\sqrt{n}\widehat{\boldsymbol{\eta}}^\pi
\xrightarrow{d^*}
\mathcal{N}_q(\mathbf 0, \boldsymbol{\Xi}^\pi)
\quad \text{as } n\to\infty.
\]
Under the previously used notation, we can consider test statistics that can be written as
\[
W_\ell = \widetilde{w}_\ell\!\left(\sqrt{n}(\widehat{\boldsymbol{\eta}}  - {\boldsymbol{\eta}}), \widehat{\boldsymbol{\Gamma}}\right)
\]
whenever $\mathcal H_{0,\ell}: \boldsymbol{\eta} \in \mathbf H_{0,\ell}$ holds, where,
$\widetilde{w}_\ell: \mathbb R^q \times \mathbb G \to \mathbb R$
is a function that is continuous in $(\mathbf x,\boldsymbol{\Gamma})$
for all $\mathbf x\in\mathrm{Im}(\boldsymbol{\Xi})$ such that
$\widetilde{w}_\ell(\sqrt{n}(\widehat{\boldsymbol{\eta}}  - {\boldsymbol{\eta}}), \widehat{\boldsymbol{\Gamma}})$
can be calculated under $\mathcal{H}_{0,\ell}$ even if $\boldsymbol{\eta}$ is unknown.
In some cases, where none of the assumption sets below (Assumptions~\ref{fullrank_ass}--\ref{conv_ass}) is plausible for
$\mathcal H_{0,\ell}: \boldsymbol{\theta}_\ell = \boldsymbol{\theta}_{0,\ell}, \ \ell\in\{1,\dots,L\}$, such a change of parameters can result in covariance matrices
$\boldsymbol{\Xi}$ and $\boldsymbol{\Xi}^\pi$ fulfilling at least one of the assumption sets.
\end{rem}

\subsection{Case 1: Full rank of $\boldsymbol{\Sigma}^{\pi}$}\label{ssec:fullrank}
Throughout this section, we suppose the following:
\begin{ass}\label{fullrank_ass}  \textbf{Full Rank of $\boldsymbol{\Sigma}^{\pi}$}: It holds $\mathrm{rank}(\boldsymbol{\Sigma}^{\pi}) = r$.
\end{ass}
To repair the incorrect covariance structure of the permutation, let us define the covariance-corrected permutation  test statistics by
$$ \mathbf W^\pi := w\left( \sqrt{n}  \widehat{\boldsymbol{\Sigma}}^{1/2}  ((\widehat{\boldsymbol{\Sigma}}^\pi)^{1/2})^+ \widehat{\boldsymbol{\theta}}^\pi ,\widehat{\boldsymbol{\Gamma}}^\pi\right) .$$
Under Assumptions~\ref{the_ass}.3 and \ref{fullrank_ass}, we have that $\widehat{\boldsymbol{\Sigma}}^{1/2} \xrightarrow{P} {\boldsymbol{\Sigma}}^{1/2}$ and $((\widehat{\boldsymbol{\Sigma}}^\pi)^{1/2})^+ \xrightarrow{P} ({\boldsymbol{\Sigma}}^\pi)^{-1/2}$ as $n\to\infty$.
Hence, the continuous mapping theorem implies the permutation consistency, i.e.,
$$ \mathbf W^\pi \xrightarrow{d^*} w\left(  {\boldsymbol{\Sigma}}^{1/2}  ({\boldsymbol{\Sigma}}^\pi)^{-1/2} \mathbf{Z}^\pi ,{\boldsymbol{\Gamma}}\right) \overset{d}{=} w\left(   \mathbf{Z} ,{\boldsymbol{\Gamma}}\right)$$ as $n\to\infty$. Note that the limit distribution corresponds to the limit distribution in \eqref{eq:Wconv}.

Further, it should be noted that this approach clearly uses the full rank assumption (Assumption~\ref{fullrank_ass}).
Firstly, it is needed to show that $((\widehat{\boldsymbol{\Sigma}}^\pi)^{1/2})^+ \xrightarrow{P} (({\boldsymbol{\Sigma}}^\pi)^{1/2})^+$ holds as $n\to\infty$ since the Moore-Penrose inverse is generally not continuous.
Secondly, we use the full rank structure for ${\boldsymbol{\Sigma}}^{1/2}  (({\boldsymbol{\Sigma}}^\pi)^{1/2})^+ \mathbf{Z}^\pi \overset{d}{=} \mathbf{Z}$. If ${\boldsymbol{\Sigma}}^\pi$ would not have full rank, the covariance matrix of ${\boldsymbol{\Sigma}}^{1/2}  (({\boldsymbol{\Sigma}}^\pi)^{1/2})^+ \mathbf{Z}^\pi$ would be ${\boldsymbol{\Sigma}}^{1/2}  (({\boldsymbol{\Sigma}}^\pi)^{1/2})^+ {\boldsymbol{\Sigma}}^\pi (({\boldsymbol{\Sigma}}^\pi)^{1/2})^+ {\boldsymbol{\Sigma}}^{1/2}$, which is generally not equal to $\boldsymbol{\Sigma}$.

In some applications, Assumption~\ref{fullrank_ass} turns out to be rather weak, as the following example shows.
\begin{example}[Many-to-one comparisons of means]
    Let us consider many-to-one comparisons in Example~\ref{QFexample}.
    They can be realized by choosing the contrast matrices 
    $\mathbf{H}_\ell = (\mathbf{e}_{\ell+1} - \mathbf{e}_1)\otimes \mathbf{I}$ for $\ell\in\{1,...,k-1\}, L = k-1$.
    If $\sum_{i=1}^k \kappa_i \Cov(\mathbf{X}_{i1})$ is of full rank, $\boldsymbol{\Gamma}^\pi$ is of full rank and, thus, 
    $\mathrm{rank}(\boldsymbol{\Sigma}^{\pi}) = \mathrm{rank}(\mathbf H \boldsymbol{\Gamma}^\pi \mathbf H') = \mathrm{rank}(\mathbf{H}) = (k-1)d $. Hence, if $\sum_{i=1}^k \kappa_i \Cov(\mathbf{X}_{i1})$ is of full rank, $\boldsymbol{\Sigma}^{\pi} \in \mathbb{R}^{(k-1)d \times (k-1)d}$ fulfils Assumption~\ref{fullrank_ass}.
\end{example}

However, in other applications, Assumption~\ref{fullrank_ass} is impossible as shown in the following example.
\begin{example}[All-pairs comparisons of means]\label{manytomany}
    Let us consider all-pairs comparisons in Example~\ref{QFexample}.
    They can be realized by choosing the contrast matrices 
    $\mathbf{H}_{\ell_1\ell_2} = (\mathbf{e}_{\ell_2} - \mathbf{e}_{\ell_1})\otimes \mathbf{I}$ for $\ell_1,\ell_2\in\{1,...,k\}, \ell_1 > \ell_2, L = k(k-1)/2.$
    Then,
    $\mathrm{rank}(\boldsymbol{\Sigma}^{\pi}) = \mathrm{rank}(\mathbf H \boldsymbol{\Gamma}^\pi \mathbf H') \leq \mathrm{rank}(\mathbf{H}) = (k-1)d $. Hence, $\boldsymbol{\Sigma}^{\pi} \in \mathbb{R}^{k(k-1)d/2 \times k(k-1)d/2}$ can never fulfil Assumption~\ref{fullrank_ass} if $k > 2$.
\end{example}

\subsection{Case 2: Distinct Eigenvalues of $\boldsymbol{\Sigma}^\pi$}\label{ssec:distinct}
To correct the permutation approach in situations where Assumption~\ref{fullrank_ass} does not hold, as, e.g., in Example~\ref{manytomany}, we use a different correction by using the eigenvalue decompositions of the covariance matrices.
For symmetric $\widehat{\boldsymbol{\Sigma}}$ and $\widehat{\boldsymbol{\Sigma}}^\pi$, let 
\begin{align*}
    \boldsymbol{\Sigma} &= \mathbf{U}\mathbf{D}\mathbf{U}', \qquad \boldsymbol{\Sigma}^\pi = \mathbf{U}^\pi\mathbf{D}^\pi\mathbf{U}^{\pi\prime}\\
    \widehat{\boldsymbol{\Sigma}} &= \widehat{\mathbf{U}}\widehat{\mathbf{D}}\widehat{\mathbf{U}}', \qquad \widehat{\boldsymbol{\Sigma}}^\pi = \widehat{\mathbf{U}}^\pi\widehat{\mathbf{D}}^\pi\widehat{\mathbf{U}}^{\pi\prime}
\end{align*}
denote the eigen value decompositions, i.e., $\mathbf{U}, \mathbf{U}^\pi$ are orthogonal matrices,
$\mathbf{D}, \mathbf{D}^\pi$ are diagonal matrices with decreasing diagonal elements, and
$\widehat{\mathbf{U}}, \widehat{\mathbf{U}}^\pi, \widehat{\mathbf{D}}, \widehat{\mathbf{D}}^\pi$ take values in the space of all orthogonal matrices and the space of all diagonal matrices with decreasing diagonal elements, respectively.

Furthermore, we suppose the following assumptions:
\begin{ass}\label{distinctass}We assume\phantom{.}
    \begin{enumerate}
    \item \textbf{Ranks of Covariance Matrices}: $\mathrm{rank}(\boldsymbol{\Sigma}) \leq \mathrm{rank}(\boldsymbol{\Sigma}^\pi) =: \widetilde r \leq r$.
    \item \textbf{Symmetry and an Assumption on the Ranks of the Covariance Matrix Estimators}: The covariance matrix estimators fulfil $\widehat{\boldsymbol{\Sigma}} = \widehat{\boldsymbol{\Sigma}}', \widehat{\boldsymbol{\Sigma}}^\pi = \widehat{\boldsymbol{\Sigma}}^{\pi\prime}$ and the inner probability that $\min\{\mathrm{rank}(\widehat{\boldsymbol{\Sigma}}), \mathrm{rank}(\widehat{\boldsymbol{\Sigma}}^\pi)\} \leq \widetilde r$ holds tends to 1 as $n\to\infty$.
        \item \textbf{Distinct Eigenvalues of $\boldsymbol{\Sigma}^\pi$}: All positive eigenvalues of $\boldsymbol{\Sigma}^\pi$ are distinct, i.e., $\mathbf{D}_{11}^\pi > ... > \mathbf{D}_{\widetilde{r}\widetilde{r}}^\pi$, where $\mathbf{D}_{ii}^\pi$ is the $i$th diagonal element of $\mathbf{D}^\pi$.
        \item \textbf{Measurability of $\widehat{\mathbf{U}}, \widehat{\mathbf{U}}^\pi$}: $\widehat{\mathbf{U}}, \widehat{\mathbf{U}}^\pi$ 
        are asymptotically measurable, see Definition~1.3.7 in \cite{vaartWellner2023} for details.
    \end{enumerate}
\end{ass}

The measurability of $\widehat{\mathbf{U}}, \widehat{\mathbf{U}}^\pi$ can always be achieved in the following way:
First, by \cite{azoff1974}, we can always find a Borel measurable function $\mathbf{V}: \{\mathbf{A} \in \mathbb R^{r\times r} \mid \mathbf{A} = \mathbf{A}'\} \to \{\mathbf W\in \mathbb R^{r\times r} \mid \mathbf{W} \text{ is orthogonal} \}$ that maps a symmetric matrix $\mathbf{A}$ to an orthogonal matrix $\mathbf{V}(\mathbf{A})$ fulfilling that $\mathbf{V}'(\mathbf{A}) \mathbf{A} \mathbf{V}(\mathbf{A})$ is diagonal with decreasing diagonal elements.
Second, $\widehat{\boldsymbol{\Sigma}}$ and $\widehat{\boldsymbol{\Sigma}}^\pi$ converge in outer probability by Assumption~\ref{the_ass}.3 and, thus, are asymptotically measurable.
Hence, one can show that $\mathbf{V}(\widehat{\boldsymbol{\Sigma}}), \mathbf{V}(\widehat{\boldsymbol{\Sigma}}^\pi)$ are asymptotically measurable. 

In order to switch the covariance matrix $\boldsymbol{\Sigma}^\pi$ to $\boldsymbol{\Sigma}$, a very simple idea is to multiply $\widehat{\mathbf U} \widehat{\mathbf D}^{1/2} ((\widehat{\mathbf D}^{\pi})^{1/2})^+ \widehat{\mathbf U}^{\pi\prime}$ from the left side to $\sqrt{n}\widehat{\boldsymbol{\theta}^\pi}$ to correct for the wrong covariance matrix in the limit distribution. 
By Assumptions~\ref{the_ass}.3 and~\ref{distinctass}.2, we obtain that $\widehat{\mathbf{D}}$ and $ \widehat{\mathbf{D}}^\pi$ converge in outer probability to $\mathbf{D}$ and $\mathbf{D}^\pi$, respectively.
However, we have to be careful with the orthogonal matrices since the map to the orthogonal matrix of an eigenvalue decomposition is generally not continuous.
Although there exists a map with the same properties that is continuous in $\boldsymbol{\Sigma}^\pi$ due to Assumption~\ref{distinctass}.3, we can not specify it in general since $\boldsymbol{\Sigma}^\pi$ is usually unknown.
Furthermore, we only assumed the eigenvalues of $\boldsymbol{\Sigma}^\pi$ to be distinct.
The following example illustrates that the stated assumptions are not enough to guarantee 
$\sqrt{n}\widehat{\mathbf{U}} \widehat{\mathbf D}^{1/2} ((\widehat{\mathbf D}^\pi)^{1/2})^+  \widehat{\mathbf U}^{\pi\prime}\widehat{\boldsymbol{\theta}}^\pi \xrightarrow{d^*} \mathcal{N}_r(\mathbf{0}, \boldsymbol{\Sigma})$ as $n\to\infty$:
\begin{example}
    For $r=1$, let $\widehat{\boldsymbol{\theta}}^\pi := {Z}/\sqrt{n}, \boldsymbol{\Sigma} = \boldsymbol{\Sigma}^\pi = \mathbf{D} = \mathbf{D}^\pi = \mathbf{U} = \mathbf{U}^\pi = \widehat{\mathbf{D}} = \widehat{\mathbf{D}}^\pi = \widehat{\mathbf{U}}= 1$,
    $\widehat{\mathbf{U}}^\pi = 1$ if $Z \geq 0$ and $\widehat{\mathbf{U}}^\pi = -1$ if $Z < 0$
    for $Z \sim \mathcal{N}(0,1)$.
    Then, Assumptions~\ref{the_ass} and~\ref{distinctass} are satisfied but $\sqrt{n}\widehat{\mathbf{U}} \widehat{\mathbf D}^{1/2} ((\widehat{\mathbf D}^\pi)^{1/2})^+  \widehat{\mathbf U}^{\pi\prime}\widehat{\boldsymbol{\theta}}^\pi = |Z|$ is not normally distributed.
\end{example}

In the proof of the following theorem, we see that, technically, we do not need the consistency of $\widehat{\mathbf{U}}$ and $\widehat{\mathbf{U}}^\pi$ as long as we interpose a random matrix $\mathbf R$ that randomly switches the sign of the eigenvectors in $\widehat{\mathbf{U}}^\pi$ independently of everything else (i.e., independent of the data and the random permutation).

\begin{theorem}\label{distinct_thm}
   We define $\mathbf{R} := \mathrm{diag}(R_1,...,R_r)$, where $R_1,...,R_r$ are i.i.d. Rademacher random variables (i.e., $P(R_1 = 1) = P(R_1 = -1) = 0.5$) that are independent of the data and the random permutation, and
   $$\widetilde{\boldsymbol{\theta}}^\pi := \widehat{\mathbf{U}} \widehat{\mathbf D}^{1/2} ((\widehat{\mathbf D}^\pi)^{1/2})^+  \mathbf{R}\widehat{\mathbf U}^{\pi\prime}\widehat{\boldsymbol{\theta}}^\pi.$$ Under Assumptions~\ref{the_ass} and~\ref{distinctass}, we have $\sqrt{n}\, \widetilde{\boldsymbol{\theta}}^\pi \xrightarrow{d^*} \mathcal{N}_r(\mathbf{0}, \boldsymbol{\Sigma})$ as $n\to\infty$.
\end{theorem}

Following the previous theorem, we define
$$ \mathbf W^\pi := w\left( \sqrt{n}\, \widetilde{\boldsymbol{\theta}}^\pi,\widehat{\boldsymbol{\Gamma}}^\pi\right) .$$
Then, we have that 
$$\mathbf W^\pi\xrightarrow{d^*}  w(\mathbf{Z},\boldsymbol{\Gamma} )$$
as $n\to\infty$ by the continuous mapping theorem.
Hence, $\mathbf W^\pi$ approximates the limit distribution of $\mathbf W$ under the global null hypothesis $\bigcap_{\ell=1}^L\mathcal{H}_{0,\ell}$, i.e., the distribution of  $w(\mathbf{Z},\boldsymbol{\Gamma})$, which is what we need to use the permutation method for inference.

\begin{rem}[Avoiding numerical issues]
    Note that even if the inner probability of $\min\{\mathrm{rank}(\widehat{\boldsymbol{\Sigma}}), \mathrm{rank}(\widehat{\boldsymbol{\Sigma}}^\pi)\} \leq \widetilde r$ tends to 1, we have to be careful with the numerical implementation of this method.
    This is because eigenvalue decompositions are typically computed using numerical approximations. Thus, even if $\min\{\mathrm{rank}(\widehat{\boldsymbol{\Sigma}}), \mathrm{rank}(\widehat{\boldsymbol{\Sigma}}^\pi)\} \leq \widetilde r$ holds theoretically, practical implementations may yield non-zero and even negative eigenvalues $\widehat{\mathbf{D}}_{ii} , i > \mathrm{rank}(\widehat{\boldsymbol{\Sigma}}),$ and $\widehat{\mathbf{D}}^\pi_{ii} := 0,i > \mathrm{rank}(\widehat{\boldsymbol{\Sigma}}^\pi)$.
    To avoid resulting numerical issues, we recommend artificially setting $\widehat{\mathbf{D}}_{ii} := 0$ for $i > \mathrm{rank}(\widehat{\boldsymbol{\Sigma}})$ and $\widehat{\mathbf{D}}^\pi_{ii} := 0$ for $i > \mathrm{rank}(\widehat{\boldsymbol{\Sigma}}^\pi)$ after numerically calculating the eigenvalue decompositions.
\end{rem}

The following example shows that for the Wald-type test statistic in the global testing problem ($L=1$), the procedure from Theorem~\ref{distinct_thm} yields the naive permutation statistic of the Wald-type test statistic.

\begin{example}[Permuted WTS for $L=1$]\label{WTSexample}
Under the notation of Example~\ref{QFexample}, let us consider only one hypothesis, i.e., $L=1$, $\boldsymbol \theta = \mathbf H \boldsymbol\mu$ with a contrast matrix $\mathbf H \in \mathbb R^{r\times kd}$, and let
$w(\mathbf{x}, \mathbf{G}) = \mathbf{x}' (\mathbf H \mathbf{G} \mathbf H')^+ \mathbf{x}$ be the 'WTS function'.
Furthermore, let $\widehat{\boldsymbol{\Sigma}} = \mathbf H \widehat{\boldsymbol{\Gamma}} \mathbf H'$ and $\widehat{\boldsymbol{\Gamma}}^\pi = \widehat{\boldsymbol{\Gamma}}$.
Then, $\mathbf W = n (\mathbf H \widehat{\boldsymbol\mu} - \boldsymbol{\theta}_0)'\widehat{\boldsymbol{\Sigma}}^+(\mathbf H \widehat{\boldsymbol\mu} - \boldsymbol{\theta}_0)$
    is the classical WTS.
    For the permutation version, many terms cancel, and we obtain 
    $\mathbf W^\pi = n (\mathbf H \widehat{\boldsymbol\mu}^\pi)'(\widehat{\boldsymbol{\Sigma}}^\pi)^+(\mathbf H \widehat{\boldsymbol\mu}^\pi)$ whenever $\mathrm{rank}(\widehat{\boldsymbol{\Sigma}}) \geq \mathrm{rank}(\widehat{\boldsymbol{\Sigma}}^\pi)$.
    Hence, $\mathbf W^\pi$ is in fact just the naive permutation statistic of $\mathbf W$.
\end{example}

For the ANOVA-type test statistic in the global testing problem ($L=1$), we obtain the following result for the permutation statistic of the ANOVA-type test statistic.

\begin{example}[Permuted ATS for $L=1$]\label{ATSexample}
    In the same setup as in the previous example, we now consider the 'ATS function' $w(\mathbf{x}, \mathbf{G}) = \mathbf{x}' \mathbf{x}/\mathrm{tr}(\mathbf{H}\mathbf{G}\mathbf{H}')$.
    Then, $\mathbf W = n (\mathbf H \widehat{\boldsymbol\mu} - \boldsymbol{\theta}_0)'(\mathbf H \widehat{\boldsymbol\mu} - \boldsymbol{\theta}_0)/\mathrm{tr}(\mathbf{H}  \widehat{\boldsymbol{\Gamma}}\mathbf{H}')$
    is the classical standardized ATS.
    For the permutation version, many terms cancel, and we obtain
    $\mathbf W^\pi = n (\mathbf H \widehat{\boldsymbol\mu}^\pi)'\widehat{\mathbf U}^\pi \widehat{\mathbf D} (\widehat{\mathbf D}^\pi)^+ \widehat{\mathbf U}^{\pi\prime}(\mathbf H \widehat{\boldsymbol\mu}^\pi)/\mathrm{tr}(\mathbf{H}  \widehat{\boldsymbol{\Gamma}}^\pi\mathbf{H}')$.

    Obviously, $\mathbf W^\pi$ differs from the naive permutation statistic $w(\sqrt{n}\widehat{\boldsymbol{\theta}}^\pi, \widehat{\boldsymbol{\Gamma}}^\pi)$ of $\mathbf W$. This is because the naive permutation generally does not approximate the same limit distribution as the ANOVA-type test statistics.  In contrast, our corrected permutation statistic $\mathbf W^\pi$ can mimic the correct limit distribution.
\end{example}

\subsection{Case 3: Convergence Rate of $\widehat{\boldsymbol{\Sigma}}^\pi$}\label{ssec:convrate}

For satisfying Assumption~\ref{distinctass}.2, we need to know that at least one of the ranks of the covariance matrix estimators is asymptotically less than or equal to the rank of $\boldsymbol{\Sigma}^\pi$ and that all positive eigenvalues of $\boldsymbol{\Sigma}^\pi$ are distinct.
In some applications, however, these two assumptions are not satisfied.
Therefore, we propose one further set of assumptions excluding the mentioned assumptions and, instead, requiring a convergence rate for the convergence of $\widehat{\boldsymbol{\Sigma}}^\pi$. 
This ensures that we can detect which of the eigenvalues are 'close' (with respect to the convergence rate) and, thus, likely to be equal.
Then, we adopt the idea of case~2 and randomly rotate the respective eigenvector estimations that we detected to correspond to equal eigenvalues.

\begin{ass}\phantom{.}\label{conv_ass}
    \begin{enumerate}
    \item \textbf{Ranks of Covariance Matrices}: $\mathrm{rank}(\boldsymbol{\Sigma}) \leq \mathrm{rank}(\boldsymbol{\Sigma}^\pi) =: \widetilde r \leq r$.
    \item \textbf{Symmetry of the Covariance Matrix Estimators}: The covariance matrix estimators fulfil $\widehat{\boldsymbol{\Sigma}} = \widehat{\boldsymbol{\Sigma}}', \widehat{\boldsymbol{\Sigma}}^\pi = \widehat{\boldsymbol{\Sigma}}^{\pi\prime}$.
        \item \textbf{Convergence Rate of $\widehat{\boldsymbol{\Sigma}}^\pi$}: For some known sequence $(r_n)_{n\in\mathbb N}$ with $r_n \to\infty$, we have that $r_n || \widehat{\boldsymbol{\Sigma}}^\pi - \boldsymbol{\Sigma}^\pi || \xrightarrow{P} 0 $. 
        \item \textbf{Measurability of $\widehat{\mathbf{U}}, \widehat{\mathbf{U}}^\pi$}: $\widehat{\mathbf{U}}, \widehat{\mathbf{U}}^\pi$ 
        are asymptotically measurable.
    \end{enumerate}
\end{ass}

The idea is, as in case~2, that we add additional randomness to break possible dependencies of the estimation of the orthogonal matrix $\widehat{\mathbf{U}}^\pi$ and $\widehat{\boldsymbol{\theta}}^\pi$.
Therefore, we first check which eigenvalues seem to be plausible to be equal and to be zero regarding the convergence rate $r_n$.
In the next step, we randomly rotate the corresponding eigenvectors in $\widehat{\mathbf{U}}^\pi$ on their column space and set all small eigenvalues to zero.

To formulate this mathematically, fix $\varepsilon > 0$ and define the index set $I := \{i\in\{2,...,r\} \mid r_n (\widehat{\mathbf D}^\pi_{(i-1)(i-1)} - \widehat{\mathbf D}^\pi_{ii}) > \varepsilon\} $, denote the $J:=|I|$ elements of the set $I$ by $i_1 < ... < i_{J}$ and set $i_0 := 1, i_{J + 1} := r+1$.
Then, we aim to randomly rotate the eigenvalues corresponding to the 'close' eigenvalues with indices $i_j,...,i_{j+1}-1$ for all $j\in\{0,...,J\}$.
Moreover, we set all eigenvalue estimators to $0$ whenever they are 'too small'. If $r_n\widehat{\mathbf D}^\pi_{rr} \leq \varepsilon$, these are the eigenvalues with indices $i_J,...,i_{J+1}-1$. Otherwise, all eigenvalues are large enough.
We obtain the same result as in Theorem~\ref{distinct_thm}:

\begin{theorem}\label{conv_thm}
For all $i\in\{1,...,r\}$, let $\mathbf R_{i;0},..., \mathbf R_{i; r-1}$ be i.i.d.\ by the Haar measure on the set $O(i)$ of all orthogonal $\mathbb R^{i\times i}$ matrices that are mutually independent and independent of the data and the random permutation.
For $\varepsilon > 0$ and $J, i_0,...,i_{J+1}$ defined as above,
   we define the $r\times r$ block diagonal matrix
   \begin{align*}
       \mathbf R_{\varepsilon} := \begin{cases}
           \bigoplus_{j=0}^{J} \mathbf R_{i_{j+1}-i_j;j} & \text{if } r_n \widehat{\mathbf D}^\pi_{rr} > \varepsilon\\ \left(\bigoplus_{j=0}^{J-1} \mathbf R_{i_{j+1}-i_j;j}\right) \oplus \mathbf{0}_{\hat{i}\times \hat{i}} &\text{if } r_n \widehat{\mathbf D}^\pi_{rr} \leq \varepsilon,
       \end{cases}
   \end{align*}
    where 
    $\hat{i} := (r+1 - i_J)$.
   Moreover, let
   $$\widetilde{\boldsymbol{\theta}}^\pi := \widehat{\mathbf{U}} \widehat{\mathbf D}^{1/2} ((\widehat{\mathbf D}^\pi)^{1/2})^+  \mathbf{R}_{\varepsilon}\widehat{\mathbf U}^{\pi\prime}\widehat{\boldsymbol{\theta}}^\pi.$$ Under Assumptions~\ref{the_ass} and~\ref{conv_ass}, we have $\sqrt{n}\,\widetilde{\boldsymbol{\theta}}^\pi \xrightarrow{d^*} \mathcal{N}_r(\mathbf{0}, \boldsymbol{\Sigma})$ as $n\to\infty$.
\end{theorem}

Again, for $$ \mathbf W^\pi := w\left( \sqrt{n} \widetilde{\boldsymbol{\theta}}^\pi,\widehat{\boldsymbol{\Gamma}}^\pi\right) ,$$ it follows
$$\mathbf W^\pi\xrightarrow{d^*}  w(\mathbf{Z},\boldsymbol{\Gamma} )$$
as $n\to\infty$ by the continuous mapping theorem.

One main issue of this approach is that we need to choose $\varepsilon > 0$ and, even worse, a sequence $r_n$ fulfilling $r_n || \widehat{\boldsymbol{\Sigma}}^\pi - \boldsymbol{\Sigma}^\pi || \xrightarrow{P} 0 $. It should be noted that if a sequence $(r_n)_n$ fulfils this convergence, then all sequences $(r_n^p)_n$ with $p\in (0,1]$ do.
If $p$ is chosen small, we put higher weight on detecting all equal eigenvalues, whereas choosing $p$ larger spends more weight on detecting all distinct eigenvalues. 
Hence, if $s_n || \widehat{\boldsymbol{\Sigma}}^\pi - \boldsymbol{\Sigma}^\pi ||$ converges in distribution to some non-degenerate limit distribution for some sequence $(s_n)_n$, we propose to choose $r_n := \sqrt{s_n}$ to ensure a fair trade-off between detecting equal and distinct eigenvalues alike.

\section{Examples and Simulations}\label{sec:examples}

In order to illustrate the broad applicability of the method, we present three different problems from statistics:
For $k$ independent groups, we present multiple tests for the comparison of univariate means, also known as (univariate) ANOVA,
for the comparison of multivariate means across $k$ groups, also known as multivariate ANOVA,
and for the comparison of restricted mean survival times (RMSTs), a popular estimand from survival analysis.

Additionally, we will evaluate the small-sample performance of our methodology in simulations using \texttt{R version 4.5.0} \cite{R}.
For all simulation settings, $N_{sim} = 2000$ simulation repetitions with $B=1999$ permutation iterations were performed and the global level of significance was set to $\alpha = 5\%$. 
The simulation studies aim to compare the family-wise error rates (FWERs) as well as the power of the proposed permutation tests to asymptotic-based and Bonferroni-adjusted testing procedures. 
Therefore, under the global null hypothesis, we calculate the \textit{empirical FWERs} as the rate of at least one rejected hypothesis (although all hypotheses are true).
Regarding power, we compute two different types under the alternative hypothesis: the \emph{empirical family-wise power}, which gives the rate that all false hypotheses are rejected, and the \emph{empirical global power}, which describes the rate that at least one hypothesis is rejected.

\subsection{Univariate ANOVA}\label{ssec:ANOVA}
In this section, we apply the methodology of Section~\ref{sec:methods} to the multiple univariate mean contrast testing problem (MCTP) \cite{konietschke2013multiple}.

\subsubsection{Notation}\label{sssec:NotANOVA} As a special case of Example~\ref{QFexample} with $d=1$, we assume that we have i.i.d.\ random variables
$X_{i1},...,X_{in_i}$ with mean $\mu_i$ and strictly positive variance $\lambda_i^2 > 0$, for all $i\in\{1,...,k\}$.
Furthermore, we assume $n_i/n \to \kappa_i \in (0,1)$.
Then,
let $\boldsymbol{\theta} =  \mathbf{H} \boldsymbol{\mu}$, where $\boldsymbol{\mu} = (\mu_1,...,\mu_k)'\in\mathbb R^k$ denotes the vector of means and $\mathbf{H} = [\mathbf{h}_1',...,\mathbf{h}_r']' \in\mathbb R^{r\times k}$ denotes a contrast matrix.

For the estimation of $\boldsymbol{\mu}$, we use the empirical means $\widehat{\boldsymbol{\mu}} = (\widehat{\mu}_1,...,\widehat{\mu}_k)'$ and, then, set $\widehat{\boldsymbol{\theta} } := \mathbf{H}  \widehat{\boldsymbol{\mu}}$.
Moreover, we consider $$\widehat{\boldsymbol{\Gamma}} = \widehat{\boldsymbol{\Gamma}}^\pi = \mathrm{diag}( n/n_1 \widehat{\lambda}^2_1,..., n/n_k \widehat{\lambda}^2_k), \qquad \widehat{\boldsymbol{\Sigma}} := \mathbf{H} \widehat{\boldsymbol{\Gamma}} \mathbf{H}',$$ and $\widehat{\boldsymbol{\Sigma}}^\pi$ being the permutation counterpart of $\widehat{\boldsymbol{\Sigma}}$, where $\widehat{\lambda}^2_1,...,\widehat{\lambda}^2_k$ denote the empirical variance estimators of the $k$ groups.

For checking Assumption~\ref{the_ass}, firstly note that the central limit theorem ensures the asymptotic normality of $\sqrt{n}(\widehat{\boldsymbol{\theta}} - \boldsymbol{\theta})$. Moreover, Theorem~1 in \cite{munko2024conditionaldeltamethodresamplingempirical} guarantees the asymptotic normality of the permutation version with $\boldsymbol{\Sigma}^\pi = {\lambda}^{\pi2}\mathbf{H} \mathrm{diag}(  \kappa_1^{-1},..., \kappa_k^{-1}) \mathbf{H}'$ for some ${\lambda}^{\pi2} > 0$. The weak law of large numbers ensures the covariance matrix estimator consistencies.

In the following, we consider the $L=r$ multiple standardized test statistics\footnote{If $\widehat{\boldsymbol{\Sigma}}_{\ell\ell} = 0$, we would set $W_\ell =0$. However, this happens with probability $0$ in the considered settings in Section~\ref{sssec:ANOVAsett}.} constructed from the $r$ rows of the contrast matrix
$$ W_\ell = n(\mathbf{h}_\ell \widehat{\boldsymbol{\mu}})^2 / \widehat{\boldsymbol{\Sigma}}_{\ell\ell}, \quad \ell\in\{1,...,r\}, $$
to test the local hypotheses $\mathcal H_{0,\ell}: \mathbf{h}_\ell \boldsymbol{\mu} = 0$, $\ell\in\{1,...,r\}$.
They can be obtained from $(\sqrt{n}\mathbf{h}_\ell \widehat{\boldsymbol{\mu}}, \widehat{\boldsymbol{\Gamma}})$ through the functions $w_\ell : \mathbb R \times \mathbb G \to \mathbb R, w_\ell (x, \mathbf{G}) := x^2 / (\mathbf{h}_\ell \mathbf{G} \mathbf{h}_\ell'),$ $\ell \in\{1,...,r\}$, and, thus, are the univariate versions of the Wald- and ANOVA-type test statistic from Example~\ref{QFexample}.
For each $\ell\in\{1,...,L\}$, we reject the $\ell$th null hypothesis $\mathcal{H}_{0,\ell}$ whenever the corresponding test statistic $W_\ell$ exceeds the balanced critical value (cf. Section~2.3 in \cite{munko2025inference}).

\subsubsection{Simulation settings}\label{sssec:ANOVAsett}
Now, we aim to analyze the small-sample performance of the proposed methods in a simulation study.
Our simulation scenarios are motivated by the simulation study in Section~4.1 in \cite{baumeister2025early} and based on the model 
$$ X_{ij} = \sigma_i (\eta_{ij} - m) + \mu_i, \quad i\in\{1,...,k\}, j\in\{1,...,n_i\}$$ with $k=6$ groups.
For the sample sizes, $\mathbf{n}_1 = (n_1,...,n_6) = K\cdot (14,...,14)$ leads to balanced sample sizes and $\mathbf{n}_2 = K\cdot(4,8,12,16,20,24)$ leads to unbalanced sample sizes, where $K\in\{1,2,4\}$ generates small, medium, and large samples, respectively. 
The scaling $\boldsymbol{\sigma}_1 = (\sigma_1,...,\sigma_6) = (1,...,1)$ results in homoscedastic variances and $\boldsymbol{\sigma}_2 = (1,1.25,...,2.25)$, $\boldsymbol{\sigma}_3 = (2.25,2,...,1)$ result in heteroscedastic variance scenarios with positive and negative pairing for unbalanced sample sizes. 
Additionally, $\eta_{11},...,\eta_{kn_k}$ are i.i.d.\ random variables, where we consider the distributions of the simulation in \cite{baumeister2025early} with existing variances: the standard normal distribution $\mathcal{N}(0,1)$, the standard lognormal distribution $\mathcal{LN}(0,1)$, the chi-squared distribution with 3 degrees of freedom $\chi^2_3$, and the $t$-distribution with 3 degrees of freedom $t_3$. 
The constant $m$ represents the mean of the corresponding distribution.
As hypothesis matrix $\mathbf H$, we consider the 
{Dunnett-type} contrast matrix \cite{dunnett_1955}
\begin{align*}
    \mathbf{H} =  \begin{bmatrix}
-1 & 1 & 0 & \cdots & 0 \\
-1 & 0 & 1 & \cdots & 0 \\
\vdots  & \vdots  & \ddots & \vdots  \\
-1 & 0 & 0 & \cdots  & 1
\end{bmatrix} \in \mathbb R^{(k-1) \times k},
\end{align*} 
the centering matrix \cite{djira_hothorn_2009} 
\begin{align*}
    \mathbf{H} =  \begin{bmatrix}
1-1/k & -1/k & -1/k & \cdots & -1/k \\
-1/k & 1-1/k & -1/k & \cdots & -1/k \\
\vdots  & \vdots  & \ddots & \ddots & \vdots \\
-1/k & -1/k & -1/k & \cdots  & 1-1/k
\end{bmatrix} \in \mathbb R^{k \times k},
\end{align*} 
and the {Tukey-type} contrast matrix \cite{tukey}
\begin{align}\label{eq:Tukey}
    \mathbf{H} = \begin{bmatrix}
-1 & 1 & 0 & 0 & \cdots & \cdots & 0 \\
-1 & 0 & 1 & 0 &\cdots & \cdots & 0 \\
\vdots  & \vdots &\vdots & \vdots & \ddots & \vdots & \vdots  \\
-1 & 0 & 0 & 0& \cdots & \cdots & 1\\
0 & -1 & 1 & 0& \cdots & \cdots & 0 \\
0 & -1 & 0 & 1& \cdots & \cdots & 0 \\
\vdots  & \vdots & \vdots  & \vdots & \ddots & \vdots & \vdots \\
0 & 0 & 0 & 0 & \cdots & -1 & 1
\end{bmatrix} \in \mathbb R^{k(k-1)/2 \times k}
\end{align} 
for $k=6$.
Note that all resulting global null hypotheses $\bigcap_{\ell=1}^L \mathcal{H}_{0,\ell}$ equal the hypothesis that the means of all groups are equal ($\mu_1 = ... = \mu_6$). 
However, they provide different local hypotheses as outlined in Example~\ref{QFexample}.
For simulating under the global null hypothesis, we set $\mu_1 = ... = \mu_6 = 0$; for power simulations, we set $\mu_6 = 1.5$ and, if the distribution is non-symmetric (i.e., for $\mathcal{LN}(0,1)$ and $\chi^2_3$), we also consider $\mu_6 = -1.5$.

Now, we check which of the three cases from Sections~\ref{ssec:fullrank}--\ref{ssec:convrate} is applicable in which scenario:
\begin{itemize}
    \item \textbf{Case 1}: We have $\mathrm{rank}(\boldsymbol{\Sigma}^\pi)  = \mathrm{rank}({\lambda}^{\pi2}\mathbf{H} \mathrm{diag}(  \kappa_1^{-1},...,\kappa_6^{-1}) \mathbf{H}') = \mathrm{rank}(\mathbf{H})$. Hence, Assumption~\ref{fullrank_ass} is fulfilled if $\mathrm{rank}(\mathbf{H}) = r$. For the three contrast matrices considered in the simulation (Dunnett-, centring, and Tukey-type), this is only fulfilled for the \textit{Dunnett-type contrast matrix}.
    \item \textbf{Case 2}: We have $\mathrm{rank}(\boldsymbol{\Sigma}^\pi)  = \mathrm{rank}(\mathbf{H}) = \mathrm{rank}(\boldsymbol{\Sigma}) =: \widetilde{r}$. Moreover, $\mathrm{rank}(\widehat{\boldsymbol{\Sigma}}) = \mathrm{rank}(\mathbf{H} \widehat{\boldsymbol{\Gamma}} \mathbf{H}') \leq \mathrm{rank}(\mathbf{H}) = \widetilde{r}$ and, analogously, $\mathrm{rank}(\widehat{\boldsymbol{\Sigma}}^\pi) \leq \widetilde{r}$.
    Hence, Assumption~\ref{distinctass}  is fulfilled if the positive eigenvalues of $\mathbf{H} \mathrm{diag}(  \kappa_1^{-1},...,\kappa_6^{-1}) \mathbf{H}'$ are distinct.
    For the considered sample size scenarios in the simulation (balanced and unbalanced), this is only fulfilled for \textit{unbalanced sample sizes} ($\mathbf{n}_2$).
    \item \textbf{Case 3}: As in case~2, we have $\mathrm{rank}(\boldsymbol{\Sigma}^\pi)  = \mathrm{rank}(\mathbf{H}) = \mathrm{rank}(\boldsymbol{\Sigma}) =: \widetilde{r}$. 
    Hence, Assumption~\ref{conv_ass}  is fulfilled if $r_n \| \widehat{\boldsymbol{\Sigma}}^\pi - \boldsymbol{\Sigma}^\pi\| \xrightarrow{P} 0$.
        In our simulation, we have $\kappa_i = n_i/n$ and, thus, it holds $\| \widehat{\boldsymbol{\Sigma}}^\pi - \boldsymbol{\Sigma}^\pi\| = \mathcal{O}_P(n^{-1/2})$, which,  e.g., follows from Theorem 2.1 in \cite{sattler2024a}.
    Hence, we choose $r_n = \min\{n_1,...,n_6\}^{1/4}$ and $\varepsilon = 0.1$.
\end{itemize}

To evaluate the performance of our methods, we additionally simulated the following competitors:
\begin{itemize}
    \item \textbf{permutation\_bonf}: We were applying the naive permutation for each local hypothesis, resulting in the naive (uncorrected) permutation statistics
    $$  n(\mathbf{h}_\ell \widehat{\boldsymbol{\mu}}^\pi)^2 / (\widehat{\boldsymbol{\Sigma}}^\pi)_{\ell\ell}, \quad \ell\in\{1,...,L\}, $$
    and adjust the resulting $L$ permutation p-values with the Bonferroni correction.
    \item \textbf{asymptotic\_bonf}: Implements the multiple asymptotic tests with Bonferroni correction, where each test statistic is compared to the Bonferroni adjusted asymptotic critical value, that is the $(1-\alpha/L)$-quantile of the $\chi^2_1$ distribution.
    \item \textbf{asymptotic}: Implements the multiple asymptotic tests \cite{dickhaus2021multivariate}, where equicoordinate normal quantiles of $$\mathcal{N}_r(\mathbf 0, \mathrm{diag}(\widehat{\boldsymbol{\Sigma}}_{11},...,\widehat{\boldsymbol{\Sigma}}_{rr})^{-1/2}  \widehat{\boldsymbol{\Sigma}} \mathrm{diag}(\widehat{\boldsymbol{\Sigma}}_{11},...,\widehat{\boldsymbol{\Sigma}}_{rr})^{-1/2}) $$ are squared to determine the critical values.
\end{itemize}

\subsubsection{Simulation results}\label{sssec:ResANOVA}
The results of the simulation under the null and alternative hypothesis for the Dunnett-type contrast matrix can be found in Figures~\ref{fig:ANOVADunn} and~\ref{fig:ANOVADunnpower} as well as in Figure~\ref{fig:ANOVADunnpower2} in the appendix.
The results for the centering and Tukey-type contrast matrix are depicted in Figures~\ref{fig:ANOVAGM}--\ref{fig:ANOVATukpower2} in the appendix.
It should be noted that not always all methods are presented as some cases are not applicable in certain settings, see the explanations above.
\begin{figure}[tbh]
    \centering
    \includegraphics[width=\linewidth]{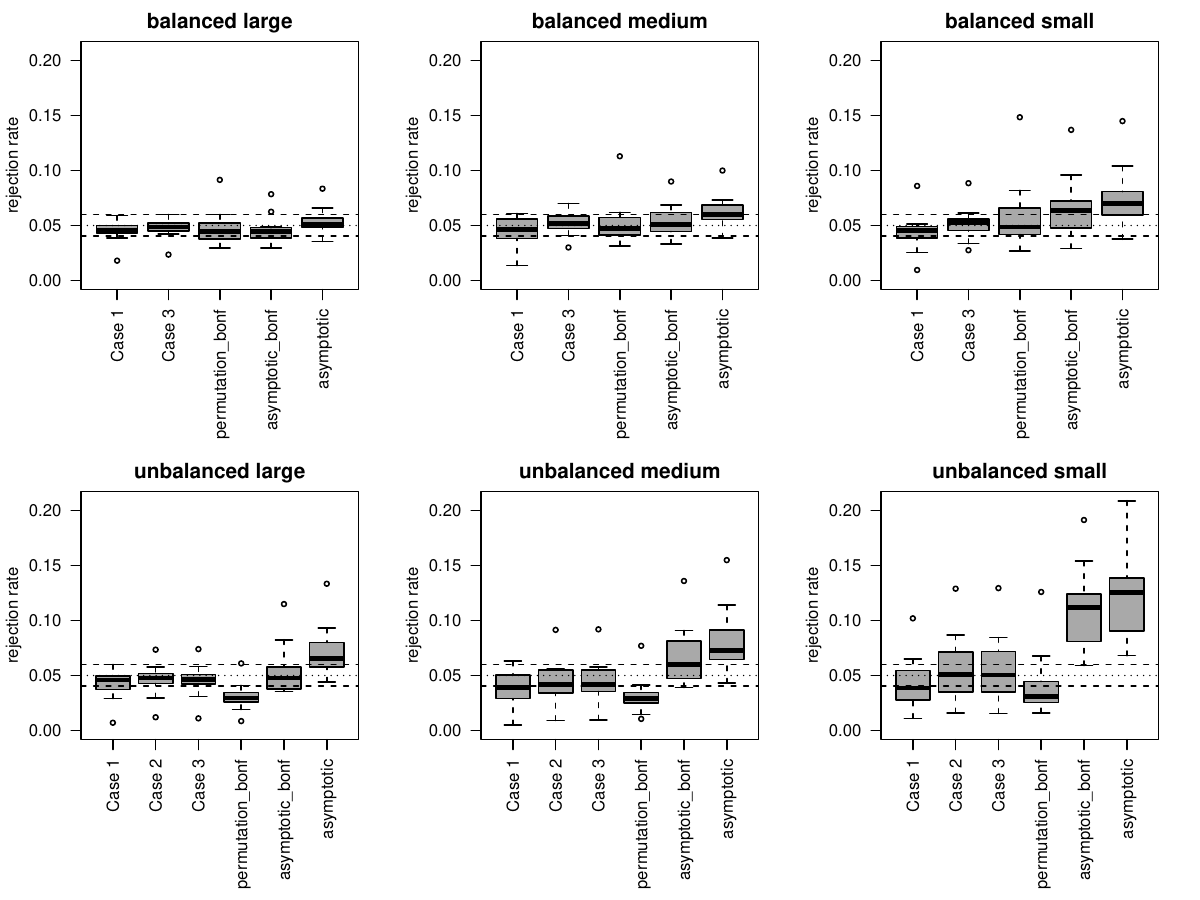}
    \caption{Empirical FWERs for the Dunnett-type contrast matrix. The dotted line represents the desired FWER of 5\% and the dashed lines represent the borders of the 95\% binomial confidence interval [0.0405, 0.06].}
    \label{fig:ANOVADunn}
\end{figure}
\begin{figure}[h]
    \centering
    \includegraphics[width=\linewidth]{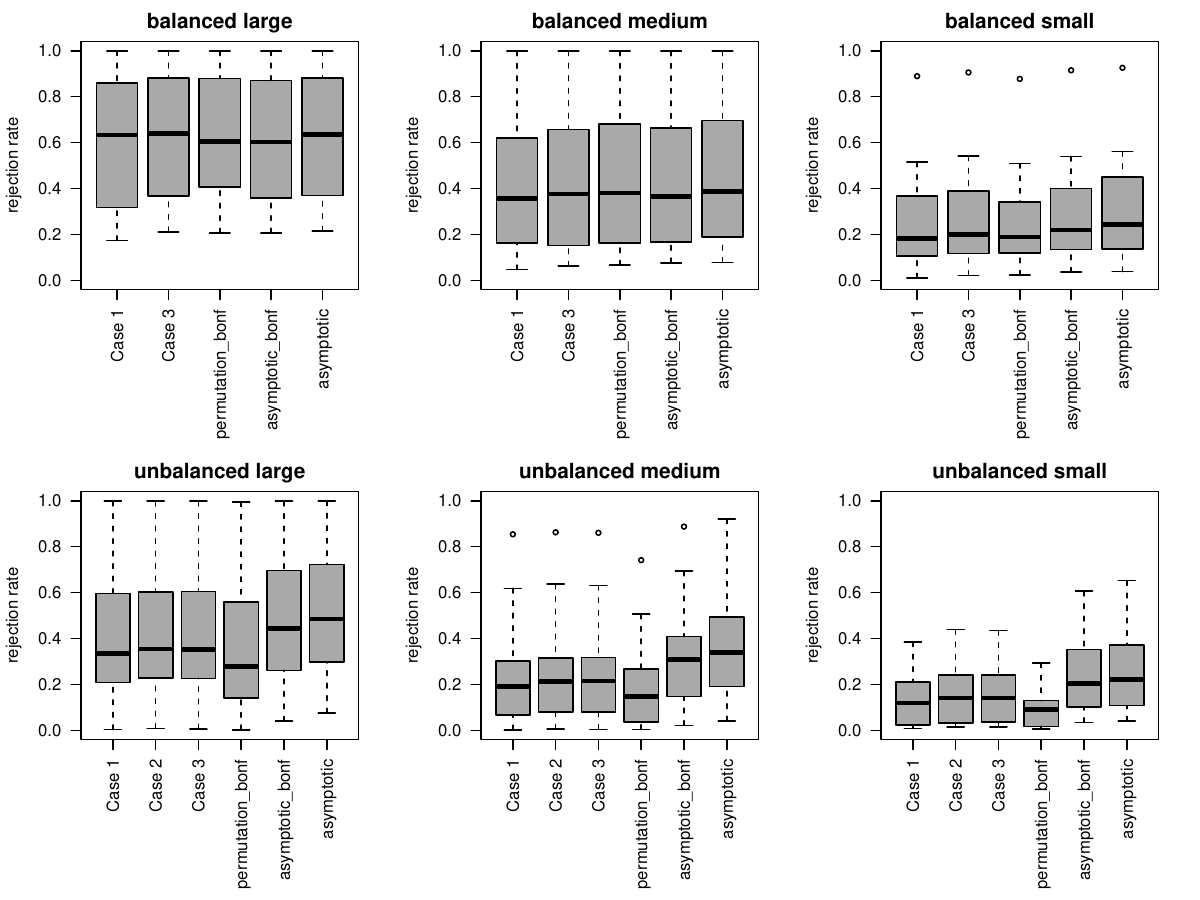}
    \caption{Empirical family-wise power for the Dunnett-type contrast matrix}
    \label{fig:ANOVADunnpower}
\end{figure}

\paragraph{FWER control for the Dunnett-type contrast matrix}
In Figure~\ref{fig:ANOVADunn}, we observe that all methods seem to keep the desired FWER of 5\% quite well in most settings with balanced sample sizes (first row) for the Dunnett-type contrast matrix. Even for smaller sample sizes with only 14 individuals per group, this can mainly be observed, although the multiple asymptotic tests have a slight tendency for liberal test decisions.
Interestingly, there are a few settings where all methods seem to struggle to control the FWER.
For the balanced scenarios, this is the setting with lognormally distributed random variables with the heteroscedastic variance scenario $\boldsymbol{\sigma}_2$, visible as upper outlier point in most of the boxplots in the first row of Figure~\ref{fig:ANOVADunn}.
For the sample sizes $n_1 = ... = n_6 = 14$, this scenario leads to empirical FWERs of 13.7--14.9\% for the Bonferroni adjusted naive permutation tests and both asymptotic approaches as well as an empirical FWER of approximately 9\% for case~1 and~3.
For $n_1 = ... = n_6 = 28$, the empirical FWERs decrease to 9.0--11.3\% for the Bonferroni adjusted naive permutation approach and the asymptotic approaches, while our methods seem to gain control over the FWER much quicker with only 5.9--7.0\% rejection rates. 
Even for $n_1 = ... = n_6 = 56$, the Bonferroni adjusted naive permutation tests and the asymptotic approaches still reach rejection rates over 8\%.
On the other side, the lowest outlier points for case~1 and~3 in the balanced scenarios (first row of Figure~\ref{fig:ANOVADunn}) correspond to lognormally distributed random variables as well but using the homoscedastic variances $\boldsymbol{\sigma}_1$. 

For the unbalanced scenarios (second row of Figure~\ref{fig:ANOVADunn}), the setting with lognormally distributed random variables with the negative pairing variance scenario $\boldsymbol{\sigma}_3$ leads to the upper outlier points and with the positive pairing variance scenario $\boldsymbol{\sigma}_2$ to the lower outlier points, reaching rejection rates even under 1\%.

Furthermore, it is apparent that the multiple asymptotic tests and the Bonferroni corrected asymptotic tests are performing much more liberal with unbalanced sample sizes, especially for small to medium sample sizes.
Additionally, the Bonferroni adjusted naive permutation tests are generally too conservative and the simulation suggests that the limit of the rejection rates seems to be notably below the desired global level of 5\%.

\paragraph{Power for the Dunnett-type contrast matrix}
For the balanced scenarios, all methods have relatively similar empirical family-wise and global power values (first row in Figures~\ref{fig:ANOVADunnpower} and~\ref{fig:ANOVADunnpower2}).

For the unbalanced scenarios, the multiple asymptotic tests and the Bonferroni corrected asymptotic tests reach the highest family-wise and global power (second row in Figures~\ref{fig:ANOVADunnpower} and~\ref{fig:ANOVADunnpower2}). However, this comes along with an increased FWER as noted previously and, thus, these methods should not be preferred for small to medium sample sizes.
On the other hand, the conservativeness of the Bonferroni adjusted naive permutation tests is reflected by a notable family-wise and global power loss of the Bonferroni adjusted naive permutation tests for the unbalanced sample sizes (second row in Figures~\ref{fig:ANOVADunnpower} and~\ref{fig:ANOVADunnpower2}) compared to our proposed methods (cases 1--3), which all yield relatively similar results regarding FWER control, family-wise power, and global power.

\paragraph{Centering and Tukey-type contrast matrix} The results with the centering and Tukey-type contrast matrix in Figures~\ref{fig:ANOVAGM}--\ref{fig:ANOVATukpower2} in the appendix look mainly similar to those of the Dunnett-type contrast matrix, which is why we will only highlight the most interesting observations.

One interesting point is that, while the asymptotic methods performed liberal primarily in unbalanced scenarios for the Dunnett-type contrast matrix,
for the centering matrix they exceed the desired FWER of 5\% in almost all scenarios (Figure~\ref{fig:ANOVAGM} in the appendix). Moreover, the FWERs seem to converge pretty slow as the rejection rates are still remarkably larger than 5\% for large sample sizes ($n=336$). 

Similarly to the Dunnett-type contrast matrix scenarios, the large outlying rejection rates for cases 2 and 3 in the unbalanced scenarios (second row of Figures~\ref{fig:ANOVAGM} and~\ref{fig:ANOVATuk} in the appendix) are reached for the skewed distributions (lognormal and chi-squared) with negative pairing.

Additionally, it is observable that the Bonferroni adjusted naive permutation tests control the FWER best in most scenarios for the centering and Tukey-type contrast matrix (Figures~\ref{fig:ANOVAGM} and~\ref{fig:ANOVATuk} in the appendix).
However, the family-wise power of the Bonferroni corrected tests is notably below those of the other tests for the Tukey-type contrast matrix (Figure~\ref{fig:ANOVATukpower} in the appendix).

For the centering matrix, all empirical family-wise powers equal zero, which is not really surprising:
The considered alternative yields the parameters $\boldsymbol{\theta} = 
(-1,-1,-1,-1,-1,5)'/4$; 
hence, most of the parameters have rather small deviations from $\theta_{0,\ell} = 0$.
Although the deviations are rather small, all local hypotheses are false, which means that all hypotheses need to be rejected to be counted for the family-wise power. This combination results in no empirical family-wise power.

\FloatBarrier
\subsection{Multivariate ANOVA}\FloatBarrier
In this section, we apply the methodology of Section~\ref{sec:methods} to the  quadratic form based multiple contrast testing problem~\cite{sattler2025quadraticformbasedmultiple}.
\subsubsection{Notation}
We reconsider Example~\ref{QFexample}, i.e., we have i.i.d.\ $d$-variate random variables
$\mathbf X_{i1},...,\mathbf X_{in_i}$ with mean $\boldsymbol{\mu}_i$ and strictly positive definite covariance matrix $\boldsymbol{\Gamma}_i$, for all $i\in\{1,...,k\}$.
Moreover, $\widehat{\boldsymbol{\Gamma}} = \widehat{\boldsymbol{\Gamma}}^\pi = \bigoplus_{i=1}^k n/n_i \widehat{\boldsymbol{\Gamma}}_i$, $\widehat{\boldsymbol{\Sigma}} := \mathbf{H} \widehat{\boldsymbol{\Gamma}} \mathbf{H}'$ and $\widehat{\boldsymbol{\Sigma}}^\pi$ being the permutation counterpart of $\widehat{\boldsymbol{\Sigma}}$, where $\widehat{\boldsymbol{\Gamma}}_1,...,\widehat{\boldsymbol{\Gamma}}_k$ denote the empirical covariance matrix estimators of the $k$ groups.
All further quantities are defined as in Example~\ref{QFexample}.

Similarly to the univariate case, Assumption~\ref{the_ass} follows by the multivariate central limit theorem, Theorem~1 in \cite{munko2024conditionaldeltamethodresamplingempirical}, and the weak law of large numbers.

In the following, we will focus on testing all pairwise equalities $\mathcal{H}_{0,\ell_1 \ell_2}: \boldsymbol{\mu}_{\ell_1} = \boldsymbol{\mu}_{\ell_2},$ $\ell_1 < \ell_2$.
Therefore, we can consider the $L= k(k-1)/2$ multiple Wald-type test statistics
$$ W_\ell = n(\mathbf{H}_\ell \widehat{\boldsymbol{\mu}})'(\mathbf{H}_\ell \widehat{\boldsymbol{\Gamma}} \mathbf{H}_\ell')^+(\mathbf{H}_\ell \widehat{\boldsymbol{\mu}}), \quad \ell\in\{1,...,L\} .$$
Here, we set $\mathbf{H}_\ell := \mathbf{h}_\ell \otimes \mathbf{I} \in \mathbb R^{d\times kd}, \ell\in\{1,...,L\},$ for $\mathbf{h}_\ell$ being the $\ell$th row of the Tukey-type contrast matrix \eqref{eq:Tukey}.
Alternatively, we can consider the multiple ANOVA-type test statistics
$$ W_\ell = n(\mathbf{H}_\ell \widehat{\boldsymbol{\mu}})'(\mathbf{H}_\ell \widehat{\boldsymbol{\mu}})/\mathrm{tr}(\mathbf{H}_\ell \widehat{\boldsymbol{\Gamma}} \mathbf{H}_\ell'), \quad \ell\in\{1,...,L\} .$$
The Wald- and ANOVA-type test statistics can be obtained from $(\sqrt{n}\mathbf{H}_\ell \widehat{\boldsymbol{\mu}}, \widehat{\boldsymbol{\Gamma}})$ through the functions presented in Example~\ref{QFexample}.
For each $\ell\in\{1,...,L\}$, we reject the $\ell$th null hypothesis iff $W_\ell$ exceeds the balanced critical value (cf. Section~2.3 in \cite{munko2025inference}).

\subsubsection{Simulation settings}
The simulation scenarios in this section are motivated by the simulation study in \cite{sattler2025quadraticformbasedmultiple}.
In detail, we consider $k=4$ groups with unbalanced sample sizes $(n_1,n_2,n_3,n_4) = K\cdot (10,10,5,5)$, $K \in\{1,2,4 \}$, and $d=4$ dimensions. The random vectors are generated independently from the model  $$\mathbf X_{ij} =  \boldsymbol{\Gamma}_i^{1/2} \mathbf Z_{ij}+ \boldsymbol{\mu}_i,\quad i\in\{1,...,k\}, j\in\{1,...,n_i\},$$ where $\mathbf Z_{ij} = (Z_{ij1},...,Z_{ij1})'$ consists of independently generated $Z_{ijl}$. For the distributions, we choose
\begin{itemize}
\item the centered exponential distribution with rate 1; $Exp(1) - 1,$
\item the standard normal distribution; $\mathcal {N}(0,1),$
\item the standardized student-t distribution with 9 degrees of freedom; $\sqrt{7/9}\cdot t_9$.
\end{itemize} 
For the covariance matrices, we consider heterogeneous equicorrelation matrices 
$$ \boldsymbol{\Gamma}_1=\boldsymbol{\Gamma}_2 = \begin{bmatrix}
3 & 1 & 1& 1\\
1 & 4 & 1 & 1\\
1 & 1 & 5 & 1\\
1 & 1 & 1 & 6
\end{bmatrix}  $$
and heterogenous autoregressive covariance matrices $ \boldsymbol{\Gamma}_3, \boldsymbol{\Gamma}_4$
with entries given by
$ (\boldsymbol{\Gamma}_3)_{ab}=(\boldsymbol{\Gamma}_4)_{ab} =  0.65^{|a-b|}\sqrt{(a+1)(b+1)}  $ for all $a,b\in\{1,...,4\}$.
Under the null hypothesis, we choose $\boldsymbol{\mu}_1 = ... = \boldsymbol{\mu}_4 = \mathbf 0$.
Under the alternative, we set $\boldsymbol{\mu}_4 = (\delta,\delta,\delta,\delta)'$ for $\delta = 1.5$. For the exponential distribution, we also consider $\delta = -1.5$.

Similarly as in Section~\ref{ssec:ANOVA}, case~1 is not applicable due to the singularity of $\boldsymbol{\Sigma}^\pi$, but the assumptions of cases~2 and 3 are satisfied.
Furthermore, note that all settings are non-exchangeable due to the different covariance matrices across the groups.
Hence, none of the methods guarantees FWER control for finite sample sizes.

In addition to our methods, we simulated the following competitors:
\begin{itemize}
    \item \textbf{permutation Bonferroni}: We were applying the naive permutation for each local hypothesis
    by using the (uncorrected) permutation counterparts of the test statistics
    and adjust the resulting $L$ permutation p-values with the Bonferroni correction.
    \item \textbf{asymptotic Bonferroni}: Only for the Wald-type test statistic, this method implements the multiple asymptotic tests with Bonferroni correction, where each test statistic is compared to the Bonferroni adjusted asymptotic critical value, that is the $(1-\alpha/L)$-quantile of the $\chi^2_d$ distribution. For the ANOVA-type test statistic, this method is not applied since the limit distributions of the ANOVA-type test statistics are generally not $\chi^2_d$ distributions.
    \item \textbf{asymptotic multiple}: Implements the multiple asymptotic tests, where the balanced critical values \cite{beran1988balanced,munko2024rmst} are obtained from 1,999 simulated values of $$\left((\mathbf{H}_\ell \mathbf Y)'(\mathbf{H}_\ell \widehat{\boldsymbol{\Gamma}} \mathbf{H}_\ell')^+(\mathbf{H}_\ell \mathbf Y)\right)_{\ell\in\{1,...,L\}} \quad\text{and}\quad \left((\mathbf{H}_\ell \mathbf Y)'(\mathbf{H}_\ell \mathbf Y)/\mathrm{tr}(\mathbf{H}_\ell \widehat{\boldsymbol{\Gamma}} \mathbf{H}_\ell')\right)_{\ell\in\{1,...,L\}},$$ respectively, with $\mathbf{Y} \sim\mathcal{N}_{16}(\mathbf{0},  \widehat{\boldsymbol{\Gamma}})$.
\end{itemize}

\subsubsection{Simulation results}
The results for the Wald-type test statistic are presented in Tables~\ref{tab:MANOVAFWER}--\ref{tab:MANOVApower}. The results for the ANOVA-type test statistic are moved to Tables~\ref{tab:MANOVAFWERATS}--\ref{tab:MANOVApowerATS} in the appendix for the sake of clarity.
\begin{table}[ht]
\centering
\begin{tabular}{|ll||rr|r|rr|}
  \hline
\multicolumn{2}{|c||}{setting} & \multicolumn{2}{c|}{multiple permutation} & \multicolumn{1}{c|}{permutation} & \multicolumn{2}{c|}{asymptotic}\\
 sample size & distribution& Case 2 & Case 3 & 
 Bonferroni & 
 Bonferroni & multiple \\ 
  \hline
  large & $Exp(1)$ & 3.45 & 3.45 & 
  3.65 & 
  10.15 & 11.45 \\ 
medium & $Exp(1)$ & 3.00 & 2.80 & 
  3.85 & 
  16.75 & 19.30 \\ 
small & $Exp(1)$ & 3.40 & 3.35 & 
  3.65 & 
  42.60 & 45.10 \\ 
large & $\mathcal N(0,1)$ & \textbf{4.55} & \textbf{4.45} & 
  \textbf{4.40} & 
  10.35 & 11.40 \\ 
medium & $\mathcal N(0,1)$ & \textbf{4.10} & \textbf{4.20} & 
  3.75 & 
  18.20 & 20.10 \\ 
small & $\mathcal N(0,1)$ & 2.70 & 2.95 & 
  3.25 & 
  46.15 & 48.60 \\ 
large & $t_9$ & 3.40 & 3.35 & 
  3.00 & 
  7.95 & 9.55 \\ 
medium & $t_9$ & 3.60 & 3.40 & 
  3.25 & 
  18.50 & 20.75 \\ 
small & $t_9$ & 2.65 & 2.75 & 
  2.95 & 
  44.95 & 47.80 \\ 
   \hline
\end{tabular}
\caption{Empirical FWERs in \% for the Wald-type test statistics. The values in the 95\% binomial confidence interval [4.05, 6] are printed in bold type.}
\label{tab:MANOVAFWER}
\end{table}
\paragraph{FWER control for the Wald-type test statistic}
In Table~\ref{tab:MANOVAFWER}, we can see that our multiple permutation tests as well as the Bonferroni corrected naive permutation tests based on the Wald-type test statistics seem to control the FWER for all distribution and sample size scenarios. However, most of the tests are rather conservative, especially for small sample sizes, as many empirical FWERs are clearly below the desired global significance level of 5\%.
In contrast, the multiple asymptotic tests and the Bonferroni corrected asymptotic tests perform too liberal with rejection rates up to 49\%.
For large sample sizes, all permutation based tests are still slightly conservative while the asymptotic based tests are still too liberal. Therefrom, we can deduce that the asymptotics did not really kick in for the considered sample sizes yet. However, it should be noted that the 'large' sample sizes considered here, $(n_1,n_2,n_3,n_4) = (40,40,20,20)$, are not particularly large, which may explain the results.

\begin{table}[ht]
\centering
\begin{tabular}{|rll||rr|r|rr|}
  \hline
\multicolumn{3}{|c||}{setting}  & \multicolumn{2}{c|}{multiple permutation} & \multicolumn{1}{c|}{permutation} & \multicolumn{2}{c|}{asymptotic}\\
 $\delta$ &sample size & distribution& Case 2 & Case 3 & Bonferroni &  Bonferroni & multiple \\ 
  \hline
   -1.5 & large & $Exp(1)$ & 43.60 & 43.50 & \textbf{51.95} & 65.95 & 67.85 \\ 
-1.5 & medium & $Exp(1)$ & 9.50 & 9.10 & \textbf{14.15} & 37.80 & 39.40 \\ 
-1.5 & small & $Exp(1)$ & 0.45 & 0.35 & \textbf{0.70} & 27.25 & 28.80 \\ 
1.5 & large & $Exp(1)$ & 35.55 & 36.05 & \textbf{36.50} & 54.55 & 57.40 \\ 
1.5 & medium & $Exp(1)$ & 5.75 & \textbf{5.85} & 5.10 & 21.55 & 22.95 \\ 
1.5 & small & $Exp(1)$ & 0.75 & \textbf{0.85} & 0.20 & 14.30 & 15.70 \\ 
1.5 & large & $\mathcal N(0,1)$ & 38.05 & \textbf{38.50} & 35.05 & 54.05 & 56.65 \\ 
1.5 & medium & $\mathcal N(0,1)$ & 4.70 & \textbf{4.90} & 4.10 & 21.55 & 22.55 \\ 
1.5 & small & $\mathcal N(0,1)$ & 0.40 & \textbf{0.45} & 0.25 & 15.30 & 16.80 \\ 
1.5 & large & $t_9$ & 39.20 & \textbf{39.75} & 36.90 & 55.25 & 57.35 \\ 
1.5 & medium & $t_9$ & \textbf{5.55} & 5.50 & 4.75 & 20.60 & 22.05 \\ 
1.5 & small & $t_9$ & \textbf{0.55} & \textbf{0.55} & 0.30 & 17.25 & 18.75 \\ 
   \hline
\end{tabular}
\caption{Empirical family-wise power in \% for the Wald-type test statistics. The largest values per setting from the methods that control the FWER regarding Table~\ref{tab:MANOVAFWER} are printed in bold type.}
\label{tab:MANOVAallpower}
\end{table}

\begin{table}[ht]
\centering
\begin{tabular}{|rll||rr|r|rr|}
  \hline
\multicolumn{3}{|c||}{setting} & \multicolumn{2}{c|}{multiple permutation} & \multicolumn{1}{c|}{permutation} & \multicolumn{2}{c|}{asymptotic}\\
 $\delta$ & sample size & distribution& Case 2 & Case 3 &  Bonferroni &  Bonferroni & multiple \\ 
  \hline
-1.5 & large & $Exp(1)$  & 84.30 & \textbf{84.45} & 
83.15 & 
91.95 & 92.65 \\ 
-1.5 & medium & $Exp(1)$  & \textbf{50.45} & 50.05 & 
49.90 & 
76.90 & 78.60 \\ 
-1.5 & small & $Exp(1)$  & \textbf{21.70} & 21.05 & 
19.05 & 
77.75 & 79.50 \\ 
1.5 & large & $Exp(1)$  & 89.85 & 89.65 & 
\textbf{91.75} & 
97.05 & 97.60 \\ 
1.5 & medium & $Exp(1)$  & 43.30 & 43.30 & 
\textbf{49.25} & 
83.15 & 84.90 \\ 
1.5 & small & $Exp(1)$  & 15.30 & 15.60 & 
\textbf{18.00} & 
79.30 & 81.50 \\ 
1.5 & large & $\mathcal N(0,1)$  & 84.50 & \textbf{84.80} & 
82.50 & 
92.40 & 93.30 \\ 
1.5 & medium & $\mathcal N(0,1)$  & 37.85 & \textbf{38.20} & 
33.80 & 
72.45 & 74.95 \\ 
1.5 & small & $\mathcal N(0,1)$  & \textbf{13.20} & 12.80 & 
11.05 & 
72.90 & 74.70 \\ 
1.5 & large & $t_9$  & 86.85 & \textbf{87.30} & 
85.30 & 
93.60 & 94.45 \\ 
1.5 & medium & $t_9$  & \textbf{40.60} & 40.55 & 
37.25 & 
72.45 & 74.80 \\ 
1.5 & small & $t_9$  & \textbf{14.20} & \textbf{14.20} & 
11.65 & 
73.35 & 75.05 \\ 
   \hline
\end{tabular}
\caption{Empirical global power in \% for the Wald-type test statistics. The largest values per setting from the methods that control the FWER regarding Table~\ref{tab:MANOVAFWER} are printed in bold type.}
\label{tab:MANOVApower}
\end{table}
\paragraph{Power for the Wald-type test statistic}
Regarding the power, Tables~\ref{tab:MANOVAallpower} and~\ref{tab:MANOVApower} show that the permutation based tests have a rather comparable family-wise and global power across all settings.
The most powerful tests from the permutation-based methods seem to depend on the direction of shift for the skewed distribution and on the type of power:
While the Bonferroni corrected naive permutation tests could most often detect all false hypotheses that resulted from a negative shift $(\delta = -1.5)$ (Table~\ref{tab:MANOVAallpower}), the multiple permutation approach yields slightly higher global powers, i.e., reject at least one hypothesis.
Interestingly, this effect is vice versa for a positive shift $(\delta = 1.5)$. 
For the symmetric distributions, the multiple permutation approach can outperform the Bonferroni corrected naive permutation tests for both power types.
By looking at the high power values of the asymptotic based tests, we should remind that these tests come along with a dramatically increased FWER and, thus, are not recommended to use for smaller sample sizes.

\paragraph{ANOVA-type test statistic}
The ANOVA-type test statistic provides much more power under the chosen alternative than the Wald-type test statistic, as shown in Tables~\ref{tab:MANOVAallpowerATS} and~\ref{tab:MANOVApowerATS} in the appendix.
However, under the null hypothesis, we observe that all methods in Table~\ref{tab:MANOVAFWERATS} in the appendix are performing too liberal by using the ANOVA-type test statistic with empirical FWERs of 5.25\% - 9.4\%.
Here, it seems that the Bonferroni corrected permutation test provides the best type-I error control, which can be explained by the conservativeness of the Bonferroni correction.

\FloatBarrier
\subsection{Comparison of RMSTs}
Finally, we apply the methodology of Section~\ref{sec:methods} to the multiple restricted mean survival time (RMST) based tests in \cite{munko2024rmst}.

\subsubsection{Notation}
In survival analysis, we are interested in (characteristics of) the distribution of the survival time $T$, which is a non-negative random variable.
However, the observations underwent right-censoring, modelled through an independent non-negative random variable $C$.
This means that we only observe the data pair $(X,\delta)$ instead of $T$, where $X = \min\{T,C\}$ denotes the censored event time and $\delta = \mathbbm{1}\{T \leq C\}$  denotes the censoring status. Here and throughout, $\mathbbm{1}$ denotes the indicator function.

For a factorial survival data setup, we assume that we have i.i.d.\ data pairs $(X_{ij},\delta_{ij}) = (\min\{T_{ij},C_{ij}\}, \mathbbm{1}\{T_{ij} \leq C_{ij}\}), j\in\{1,...,n_i\}$, for all $i\in\{1,...,k\}$, see e.g. \cite{ditzhaus2021inferring, munko2024rmst} for details.
The survival function is defined by $S_i(\cdot) := P(T_{ij} > \cdot )$ 
and the restricted mean survival times (RMSTs) by $\mu_i := \int_0^\tau S_i(t) \mathrm{d}t$ for some fixed $\tau > 0$
and for all $i\in\{1,...,k\}$.
Then, $\boldsymbol{\mu}=(\mu_1,...,\mu_k)'\in\mathbb R^k$ denotes the vector of RMSTs, $\widehat{\boldsymbol{\mu}}$ being its estimator as in \cite{munko2024rmst}, and
$\widehat{\lambda}^2_1,...,\widehat{\lambda}^2_k$ denote the variance estimators (Equation (3) in \cite{munko2024rmst}).
All other quantities are defined as in Section~\ref{sssec:NotANOVA}.

The proof of Theorems~1 and~2 in \cite{munko2024rmst} ensure the covariance matrix consistencies as well as the asymptotic normality of the estimator and the permutation version with $\boldsymbol{\Sigma}^\pi = {\sigma}^{\pi2}\mathbf{H} \mathrm{diag}(  \kappa_1^{-1},..., \kappa_k^{-1}) \mathbf{H}'$ for some ${\sigma}^{\pi2} > 0$ under mild assumptions. 

\subsubsection{Simulation settings}
We simulated the same 270 settings as described in Section~4.1 in \cite{munko2024rmst} together with the Dunnett-type, centering, and Tukey-type contrast matrix. 
Here, similarly to Section~\ref{ssec:ANOVA}, case~1 is only applicable for the Dunnett-type contrast matrix, whereas case~2 is only applicable for unbalanced sample sizes.

In order to compare our methods to competitors, we added the following methods:
\begin{itemize}
    \item \textbf{permutation\_bonf}: the global studentized permutation tests from Section~2.3 in \cite{munko2024rmst} adjusted with the Bonferroni correction,
    \item \textbf{groupwise}: the multiple groupwise bootstrap tests from Section~3 in \cite{munko2024rmst},
    \item \textbf{asymptotic}: asymptotic multiple Wald-type tests as described in Section~4.1 in \cite{munko2024rmst}.
\end{itemize}
We chose those competitors because the multiple groupwise bootstrap test and the global studentized permutation tests adjusted with the Bonferroni correction yielded the best results regarding FWER control in the simulation study of \cite{munko2024rmst} from all methods that provide multiple test decisions.
Moreover, the asymptotic multiple Wald-type tests serve as a benchmark method.

\subsubsection{Simulation results}
In Figures~\ref{fig:RMSTDunn}, \ref{fig:RMSTDunnpower}, and \ref{fig:RMSTDunnpower2}, the simulation results for the FWER, family-wise power, and global power, respectively, are depicted for the Dunnett-type contrast matrix.
The results for the centering and Tukey-type contrast matrix can be found in Figures~\ref{fig:RMSTGM}--\ref{fig:RMSTTukpower2} in the appendix.
\begin{figure}[tbh]
    \centering
    \includegraphics[width=\linewidth,trim= 1mm 2mm 7mm 2mm,clip]{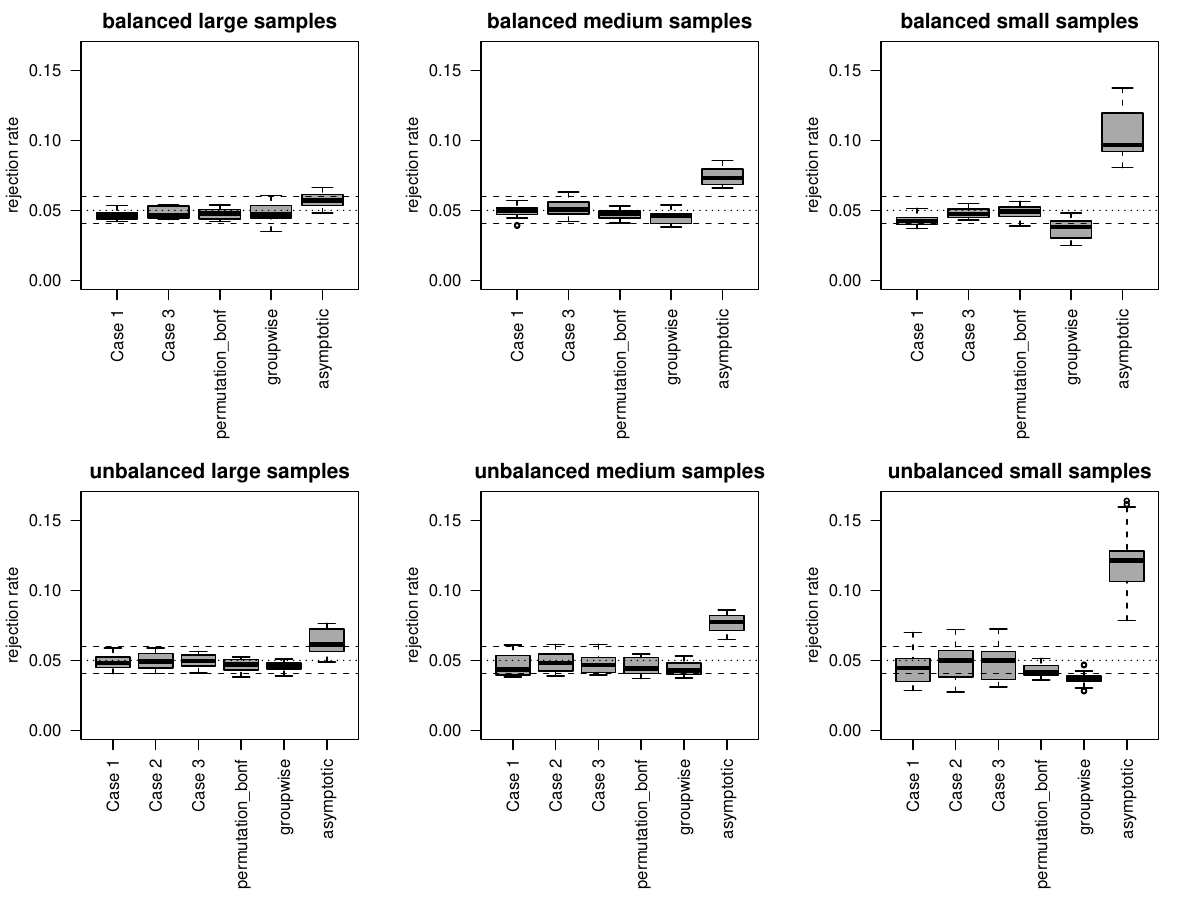}
    \caption{Empirical FWERs for the Dunnett-type contrast matrix. The dotted line represents the desired FWER of 5\% and the dashed lines represent the borders of the 95\% binomial confidence interval [0.0405, 0.06].}
    \label{fig:RMSTDunn}
\end{figure}
\begin{figure}[h]
    \centering
    \includegraphics[width=\linewidth,trim= 1mm 2mm 7mm 2mm,clip]{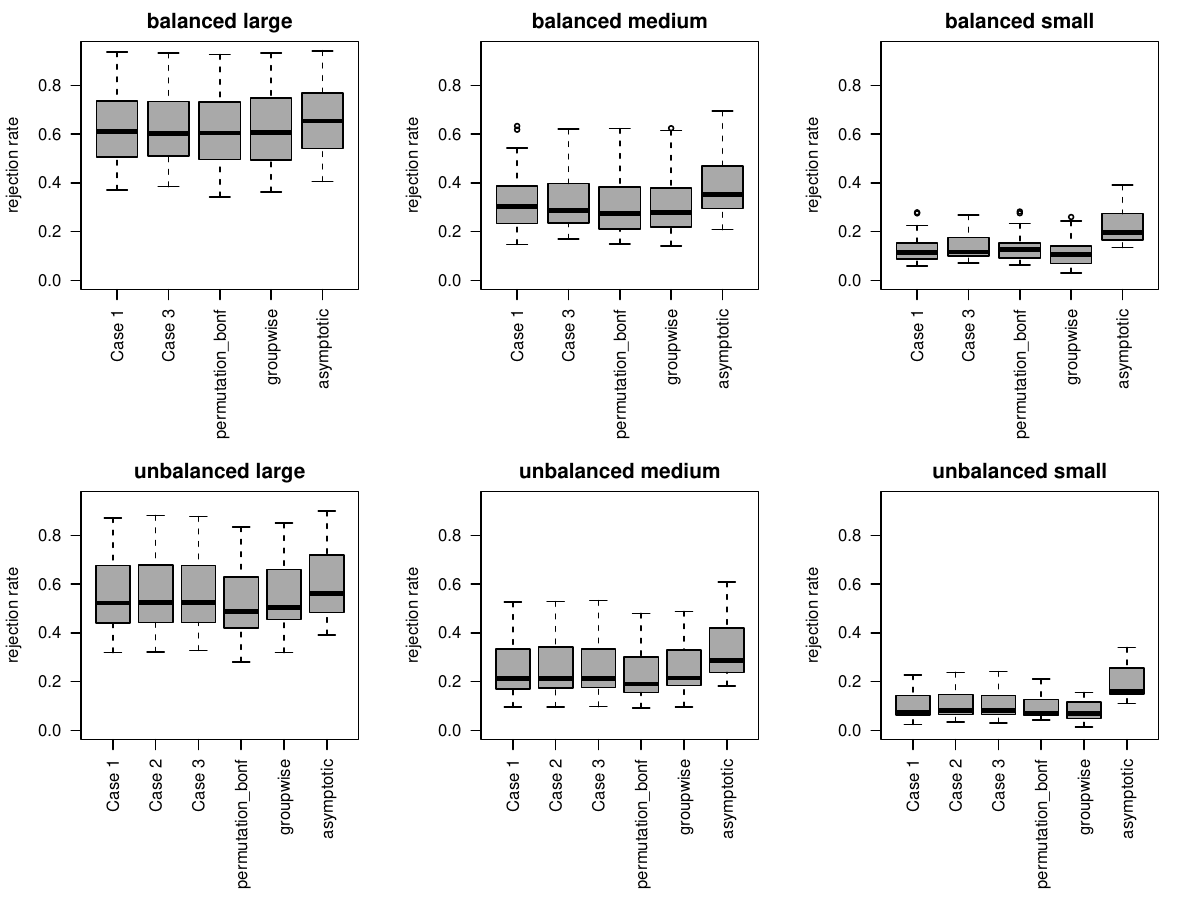}
    \caption{Empirical family-wise power for the Dunnett-type contrast matrix}
    \label{fig:RMSTDunnpower}
\end{figure}

\paragraph{FWER control}
The results of the simulation under the null hypothesis can be found in Figure~\ref{fig:RMSTDunn} for the Dunnett-type contrast matrix and in Figures~\ref{fig:RMSTGM} and~\ref{fig:RMSTTuk} in the appendix for the centering and Tukey-type contrast matrix.
We can observe that all cases 1--3, when applicable, lead to a quite accurate FWER control, similar to the multiple groupwise bootstrap tests and global permutation tests with Bonferroni correction.
The asymptotic-based multiple Wald-type tests without a resampling mechanism, however, lead to an inflated type-I error, especially for small and medium sample sizes, as already noted in \cite{munko2024rmst}.

\paragraph{Family-wise power}
Figures~\ref{fig:RMSTDunnpower}, \ref{fig:RMSTGMpower}, and \ref{fig:RMSTTukpower} show the corresponding empirical family-wise power.
The asymptotic-based multiple Wald-type tests yield the highest rejection rates; however, this procedure could not control the FWER well, especially for smaller sample sizes, under the null hypothesis. Hence, we do not recommend to use this approach for small to medium sample sizes.
The empirical family-wise power of the Bonferroni corrected permutation tests seem to be usually lower than the power of our methods, especially for unbalanced sample sizes (second row in Figures~\ref{fig:RMSTDunnpower} and \ref{fig:RMSTTukpower}), underlining the known conservativeness of the Bonferroni correction.
For small sample sizes, the multiple groupwise bootstrap tests has slightly less family-wise power than our methods (third column in Figures~\ref{fig:RMSTDunnpower} and \ref{fig:RMSTTukpower}).
All other methods perform mainly comparable across the different settings regarding family-wise power.
Again, as in Section~\ref{sssec:ResANOVA}, the centering matrix leads to almost no family-wise power due to rather small effect sizes under the chosen alternative.

\paragraph{Global power}
Our methods, as well as the permutation tests with Bonferroni correction, seem to have slightly more global power than the multiple groupwise bootstrap tests for small samples with the Dunnett-type contrast matrix (third column in Figure~\ref{fig:RMSTDunnpower2}).
In contrast, the multiple groupwise bootstrap tests and the permutation tests with Bonferroni correction achieve slightly more global power than the methods from case~2 and~3 for unbalanced sample sizes with the centering matrix (second row in Figure~\ref{fig:RMSTGMpower2}).
For the Tukey-type contrast matrix and balanced sample sizes (first row in Figure~\ref{fig:RMSTTukpower2}), the method from case~3 tends to be a bit more powerful than the multiple groupwise bootstrap tests and the permutation tests with Bonferroni correction.

\FloatBarrier

\section{Discussion}\label{sec:discussion}
Permutation tests are widely used in quantitative sciences due to their finite-sample exactness under exchangeability. As this assumption is often violated in practice, studentization has become a standard tool to obtain asymptotic validity and typically provides accurate approximations even for small sample sizes.
In multiple testing problems, however, available permutation approaches either remain computationally prohibitive, such as prepivoting procedures \cite{romano2005exact}, or rely on conservative corrections like Bonferroni adjustments that ignore the joint dependence structure of the test statistics. The present paper closes this gap by introducing a general framework that corrects the covariance structure of multiple permutation statistics, thereby recovering the correct joint limiting distribution and enabling asymptotically valid multiple permutation tests.
In several settings and for different parameters, extensive simulation studies indicate that the proposed framework shows favourable properties relative to existing approaches.

Our theoretical results establish asymptotic validity under three different sets of assumptions: The first case requires a full-rank permutation covariance matrix, the second relies primarily on distinct eigenvalues of the permutation covariance matrix, and the third assumes convergence of a covariance matrix estimator at a known rate. 
If two or all sets of assumptions can be validated, the question arises which of the three methods should be applied.
From a practical perspective, cases~1 and~2 are preferable whenever their assumptions are satisfied, as the third case depends on additional tuning parameters, such as the threshold $\varepsilon > 0$ and the rate sequence $(r_n)_n$, which may affect performance. Encouragingly, our simulation results indicate that all three approaches perform comparably well in terms of family-wise error rate control as well as family-wise and global power, whenever applicable.

A theoretical comparison of the three cases, particularly with respect to asymptotic efficiency and local power properties, remains an open problem. Future work could complement such analyses with more extensive simulation studies that systematically explore the constellations in which one case may outperform the others.

Furthermore, it would be desirable to develop data-driven procedures for selecting the threshold parameter $\varepsilon$ in case~3. Such procedures could reduce sensitivity to tuning choices and further enhance the practical applicability of the proposed approach.

Moreover, the approach can be expanded to other null hypotheses and parameters as outlined in Remark~\ref{Extension}.

Finally, it would be of considerable interest to extend the approach to non-Gaussian limit distributions in future works, such as those arising, for instance, from non-degenerate U-statistics or from increasing dimension as considered in \cite{sattler2021}. While such extensions would be substantially more demanding from a theoretical perspective, they would greatly broaden the applicability of the proposed methodology. A particularly relevant example would be limit distributions given by weighted sums of independent $\chi_1^2$-distributed random variables.

\section*{Acknowledgments}
Merle Munko gratefully acknowledges support from the Deutsche Forschungsgemeinschaft (Project No. 352692197 and 314838170).
Merle Munko would like to thank the late Marc Ditzhaus, the greatest fan of permutation tests she knew, for many early discussions on permutation tests and multiple testing.
The authors are grateful to Markus Pauly for insightful comments and valuable suggestions.
{OpenAI's  ChatGPT model was used to refine phrasing and improve readability. The authors reviewed and approved all resulting text.}

\bibliographystyle{abbrv}
\bibliography{literature}
\appendix
\section{Finite-Sample Result under Invariance}\label{sec:finite}

In this section, we study the finite-sample properties of our proposed methods.
It will turn out that, in some specific scenarios, the methods are finitely exact under exchangeability.
Therefore, we use the ideas of \cite{hemerik2018exact}.


Let $\mathbf X_n$ denote the data in a space $\chi_{1n}$ and $G$ be a finite set of transformations $g: \chi_{1n} \to \chi_{1n}$. Moreover, $G$ is a group with respect to the composition operation.
Assume that we can write $\mathbf W^\pi = w^\pi(h(\mathbf X_n),\mathbf{X}_n) = (w_\ell^\pi(h(\mathbf X_n),\mathbf{X}_n))_{\ell\in\{1,...,L\}}$, where $h$ is uniformly distributed on $G$ and $w^\pi = (w_\ell^\pi)_{\ell\in\{1,...,L\}}: \chi_{1n}^2 \to \mathbb R^L$ is a map,
and $\mathbf W = w^\pi(\mathbf X_n,\mathbf{X}_n)$. 
Intuitively speaking, the latter assumption means that the permutation statistics should depend in the same way on the permuted data as the original test statistics depend on the original data, such that the permutation test statistic equals the original test statistic if the data is not permuted.

Then, we obtain the following finite-sample result for the family-wise error rate:
\begin{theorem}
Let $\mathcal T \subset\{1,...,L\}$ and $\alpha\in [0,1)$.
    Furthermore, let $\mathbf W^\pi = w^\pi(h(\mathbf X_n),\mathbf{X}_n) = (w_\ell^\pi(h(\mathbf X_n),\mathbf{X}_n))_{\ell\in\{1,...,L\}}$, where $h$ is uniformly distributed on $G$ and $w^\pi = (w_\ell^\pi)_{\ell\in\{1,...,L\}}: \chi_{1n}^2 \to \mathbb R^L$ is a map,
 $\mathbf W = w^\pi(\mathbf X_n,\mathbf{X}_n)$, and the joint distribution of $\left(w_\ell^\pi(g(\mathbf X_n),\mathbf{X}_n)\right)_{\ell\in\mathcal T}, g\in G,$ is invariant under all transformations in $G$ of $\mathbf X_n$, i.e.,
 $$ \left( \left(w_\ell^\pi(a_1(\mathbf X_n),\mathbf{X}_n)\right)_{\ell\in\mathcal T},\dots,\left(w_\ell^\pi(a_{|G|}(\mathbf X_n),\mathbf{X}_n)\right)_{\ell\in\mathcal T} \right) \overset{d}{=} \left( \left(w_\ell^\pi(a_1(g(\mathbf X_n)),\mathbf{X}_n)\right)_{\ell\in\mathcal T},\dots,\left(w_\ell^\pi(a_{|G|}(g(\mathbf X_n)),\mathbf{X}_n)\right)_{\ell\in\mathcal T}  \right)$$ for all $g\in G = \{a_1,\dots, a_{|G|}\}$.
Then, we have $$P\left( \exists \ell\in \mathcal{T}:\: W_\ell > W_\ell^{\pi (k)}\right) \leq \alpha.$$
Here, $W_\ell^{\pi (1)} \leq ... \leq W_\ell^{\pi (|G|)}$ denote the sorted values of $w_\ell^{\pi }(g(\mathbf X_n),\mathbf{X}_n), g\in G,$ for all $\ell\in\mathcal{T}$ and $k$ chooses a suitable quantile such that
$|\{g\in G\mid \exists\ell\in\mathcal{T}:\:  w_\ell^{\pi }(g(\mathbf X_n),\mathbf{X}_n) >  W_\ell^{\pi (k)} \}|  \leq \alpha \cdot |G|.$
\end{theorem}
The proof follows analogously to the proof of Theorem~1 in \cite{hemerik2018exact} by writing $\mathbf{x} > \mathbf{y}$ for vectors $\mathbf{x}=(x_\ell)_{\ell\in\mathcal T}, \mathbf{y}=(y_\ell)_{\ell\in\mathcal T}\in\mathbb R^{\mathcal T}$ whenever there exists an $\ell\in\mathcal{T}$ such that $x_\ell > y_\ell$.

The theorem provides that the family-wise error rate is bounded by $\alpha$ even for finite sample sizes under $G$-invariance if the permutation statistics depend in the same way on the permuted data as the original test statistics depend on the original data.

Now, we aim to discuss the assumption $\mathbf W = w^\pi(\mathbf X_n,\mathbf{X}_n)$ in our applications further.
Therefore, we firstly introduce the following definition:
\begin{defi}
    For a probability space $\Omega$, a measurable space $\chi$ and two maps $\mathbf{Y}^\pi,\mathbf{Y}: \Omega \to \chi$, we say that $\mathbf{Y}^\pi$ is the permutation counterpart of $\mathbf{Y}$ if there is a map $Y : \chi_{1n} \to \chi  $ such that  $\mathbf{Y}^\pi = Y \circ h\circ \mathbf{X}_n$ for $h$ being uniformly distributed in $G$ and $\mathbf{Y} = Y \circ  \mathbf{X}_n$.
\end{defi}
Now, we reconsider Examples~\ref{WTSexample} and~\ref{ATSexample}:

\addtocounter{example}{-2}
\begin{example}[continued]
We already observed that the permutation statistic of the WTS is given by
    $\mathbf W^\pi = w^\pi(h(\mathbf X_n),\mathbf{X}_n) = n (\mathbf H \widehat{\boldsymbol\mu}^\pi)'(\widehat{\boldsymbol{\Sigma}}^\pi)^+(\mathbf H \widehat{\boldsymbol\mu}^\pi)$ whenever $\mathrm{rank}(\widehat{\boldsymbol{\Sigma}}) \geq \mathrm{rank}(\widehat{\boldsymbol{\Sigma}}^\pi)$.
    If 
    \begin{itemize}
        \item $\widehat{\boldsymbol{\Sigma}}^\pi$ is the permutation counterpart of $\widehat{\boldsymbol{\Sigma}}$, i.e.,
    we use the same covariance matrix estimator for the original test statistic and the permutation statistic,
    \item $\mathrm{rank}(\widehat{\boldsymbol{\Sigma}}) \geq \mathrm{rank}(\widehat{\boldsymbol{\Sigma}}^\pi)$,
    \item and $\boldsymbol{\theta}_0=\mathbf{0}$,
    \end{itemize}
    we have
    $\mathbf W = n (\mathbf H \widehat{\boldsymbol\mu})'(\widehat{\boldsymbol{\Sigma}})^+(\mathbf H \widehat{\boldsymbol\mu}) =w^\pi(\mathbf X_n,\mathbf{X}_n)$.
\end{example}

\begin{example}[continued]
    We already observed that the permutation statistic of the ATS is given by
    $$\mathbf W^\pi = w^\pi(h(\mathbf X_n),\mathbf{X}_n)= n (\mathbf H \widehat{\boldsymbol\mu}^\pi)'\widehat{\mathbf U}^\pi \widehat{\mathbf D} (\widehat{\mathbf D}^\pi)^+ \widehat{\mathbf U}^{\pi\prime}(\mathbf H \widehat{\boldsymbol\mu}^\pi)/\mathrm{tr}(\mathbf{H}  \widehat{\boldsymbol{\Gamma}}^\pi\mathbf{H}').$$
    If 
    \begin{itemize}
        \item $\widehat{\mathbf{U}}^\pi, \widehat{\mathbf{D}}^\pi$ are the permutation counterparts of $\widehat{\mathbf{U}}, \widehat{\mathbf{D}}$,
    \item either $\widehat{\boldsymbol{\Gamma}}^\pi = \widehat{\boldsymbol{\Gamma}}$ or $\widehat{\boldsymbol{\Gamma}}^\pi$ is the permutation counterpart of $\widehat{\boldsymbol{\Gamma}}$,
    \item $\mathrm{rank}(\widehat{\boldsymbol{\Sigma}}) = r$,
    \item and $\boldsymbol{\theta}_0=\mathbf{0}$,
    \end{itemize}
    we have
    $$\mathbf W = n (\mathbf H \widehat{\boldsymbol\mu})'(\mathbf H \widehat{\boldsymbol\mu})/\mathrm{tr}(\mathbf{H}  \widehat{\boldsymbol{\Gamma}}\mathbf{H}') = n (\mathbf H \widehat{\boldsymbol\mu})'\widehat{\mathbf{U}}\widehat{\mathbf{D}} \widehat{\mathbf{D}}^+ \widehat{\mathbf{U}}'(\mathbf H \widehat{\boldsymbol\mu})/\mathrm{tr}(\mathbf{H}  \widehat{\boldsymbol{\Gamma}}\mathbf{H}') =w^\pi(\mathbf X_n,\mathbf{X}_n).$$
\end{example}

For case~1, we actually get general assumptions such that we can write $\mathbf W = w^\pi(\mathbf X_n,\mathbf{X}_n)$:
We can write $\mathbf W = w^\pi(\mathbf X_n,\mathbf{X}_n)$ if $\widehat{\boldsymbol{\Sigma}}^\pi, \widehat{\boldsymbol{\theta}}^\pi$ are permutation counterparts of $ \widehat{\boldsymbol{\Sigma}}, \widehat{\boldsymbol{\theta}} - \boldsymbol{\theta}_0$, respectively, and $\mathrm{rank}(\widehat{\boldsymbol{\Sigma}}) = r$ with inner probability one.
Note that these assumptions are sufficient for $\mathbf W = w^\pi(\mathbf X_n,\mathbf{X}_n)$ but not necessary. For instance, the assumptions in Example~\ref{WTSexample} are weaker than those stated above by only requiring that $\mathrm{rank}(\widehat{\boldsymbol{\Sigma}}) \geq \mathrm{rank}(\widehat{\boldsymbol{\Sigma}}^\pi)$ instead of $\mathrm{rank}(\widehat{\boldsymbol{\Sigma}}) = r$.

For cases~2 and~3, the random matrix $\mathbf{R}$ makes things more complicated, as it only appears in the permutation statistic but not in the original test statistic. However, in some special cases, e.g., when the permutation statistics do not really depend on $\mathbf{R}$, the finite exactness can be guaranteed as, for instance, in Examples~\ref{WTSexample} and~\ref{ATSexample} above.

\section{Proofs}\label{sec:proofs}
\begin{proof}[Proof of Theorem~\ref{distinct_thm}]
    First, let us remind ourselves how weak convergence in outer probability is defined.
    The intuition behind the following definition is that $\mathbf M_n$ will represent additional randomness, e.g., induced by a random permutation and $\mathbf R$, whereas $\mathbf X_n$ will represent the original data.
    Let $\mathbf{X}_n:\Omega_1\to \chi_{1n},\mathbf M_n:\Omega_2\to \chi_{2n}$ be sequences of maps, where $(\Omega_1\times\Omega_2, \mathcal{A}_1\otimes\mathcal{A}_2, Q_1 \otimes Q_2)$ denotes a product probability space, $\chi_{1n}, (\chi_{2n}, \mathcal B_n)$ denote sequences of sets and measurable spaces, respectively, for $n\in\mathbb N$, and $\mathbf M_n$ is measurable.
Furthermore, assume that $\mathbf y_n:\chi_{1n}\times \chi_{2n} \to \mathbb R^r$ for all $n\in\mathbb{N}$ such that $\mathbf y_n(x,\cdot)$ is measurable for all $x\in \chi_{1n}$ and we can write $\widetilde{\boldsymbol{\theta}}^\pi = \mathbf y_n(\mathbf{X}_n,\mathbf M_n)$.
Similarly, we suppose that there are functions $\mathbf y_n^{(1)}:\chi_{1n}\times \chi_{2n} \to \mathbb R^r$ and $\mathbf y_n^{(2)}:\chi_{1n}\times \chi_{2n} \to \mathbb R^{r\times r}$ for all $n\in\mathbb{N}$ such that $\mathbf y_n^{(i)}(x,\cdot)$ is measurable for all $x\in \chi_{1n}, i\in\{1,2\}$ and we can write $\widehat{\boldsymbol{\theta}}^\pi = \mathbf y_n^{(1)}(\mathbf{X}_n,\mathbf M_n), \widehat{\mathbf U}^\pi = \mathbf y_n^{(2)}(\mathbf{X}_n,\mathbf M_n)$.
We say that {$\widetilde{\boldsymbol{\theta}}^\pi$ converges weakly (conditionally on $\mathbf{X}_n$) in outer probability to $\mathbf{Z} \sim \mathcal N_r(\mathbf{0}, \boldsymbol{\Sigma})$}, write $\widetilde{\boldsymbol{\theta}}^\pi \xrightarrow{d^*} \mathcal N_r(\mathbf{0}, \boldsymbol{\Sigma})$ (conditionally on $\mathbf{X}_n$) as $n\to\infty$, if
    \begin{align}\label{eq:precise}
    &\sup\limits_{h\in BL_1(\mathbb R^r)} \left| \E_{\mathbf{R}, \pi} \left[h\left(\sqrt{n}\,\widetilde{\boldsymbol{\theta}}^\pi \right)\right]  - \E \left[h(\mathbf Z)\right] \right| \xrightarrow{Q_1} 0\quad \text{and}
       \\& \label{eq:precise2}
        \E_{\mathbf{R}, \pi}\left[ h\left(\sqrt{n}\,\widetilde{\boldsymbol{\theta}}^\pi \right)^* \right] - \E_{\mathbf{R}, \pi} \left[h\left( \sqrt{n}\,\widetilde{\boldsymbol{\theta}}^\pi \right)_*\right]  \xrightarrow{Q_1} 0\quad\text{as $n\to\infty$ for all $h\in BL_1(\mathbb R^r)$.}
    \end{align}
    Here, $\E$ denotes the expectation with respect to $\Omega_1\times\Omega_2$, 
    $\E_{\mathbf{R}, \pi}$ denotes the conditional expectation with respect to $\Omega_2$,  $BL_1(\mathbb R^r)$ denotes the set of all real functions on $\mathbb R^r$ with a Lipschitz norm bounded by $1$, and the super- and subscript asterisks denote the minimal measurable majorants and maximal measurable minorants, respectively, with respect to $\Omega_1\times\Omega_2$ jointly, see \cite{vaartWellner2023} for details.

    We start by proving \eqref{eq:precise}.
    Therefore, we extend $\Omega_2$ by setting $\widetilde{\Omega}_2 := \Omega_2 \times \Omega_3$ equipped with the product $\sigma$-algebra and product probability measure 
    and
    define $\mathbf{Y}:\Omega_3 \to \mathbb R^r$ as a Borel measurable random variable with $\mathbf{Y} \sim \mathcal N_r(\mathbf 0, \boldsymbol{\Sigma}^\pi)$. In the following, we will denote the expectation with respect to $\widetilde{\Omega}_2$ by $\E_{\mathbf{R}, \pi, \mathbf{Y}}$. 
    Moreover, for $\widehat{\mathbf U}^\pi = (\widehat{\mathbf{u}}_1^\pi,...,\widehat{\mathbf{u}}_r^\pi), \mathbf{U}^\pi = (\mathbf{u}_1^\pi,...,\mathbf{u}_r^\pi)$, we define $\widehat{\mathbf{R}}^\pi := \mathrm{diag}\left(\mathrm{sgn}((\widehat{\mathbf{u}}_1^\pi)'{\mathbf{u}}_1^\pi),...,\mathrm{sgn}((\widehat{\mathbf{u}}_r^\pi)'{\mathbf{u}}_r^\pi)\right)$ with $\mathrm{sgn}: \mathbb R \to \{-1,1\}, \mathrm{sgn}(x):= -1$ if $x < 0$ and $\mathrm{sgn}(x) := 1$ if $x\geq 0$.
 Then, we divide \eqref{eq:precise} into the following terms:
    \begin{align}\notag
    &\sup\limits_{h\in BL_1(\mathbb R^r)} \left| \E_{\mathbf{R},\pi} \left[h(\sqrt{n}\,\widetilde{\boldsymbol{\theta}}^\pi )\right]  - \E \left[h(\mathbf Z)\right] \right| \\ &\leq  \sup\limits_{h\in BL_1(\mathbb R^r)} \left| \E_{\mathbf{R},\pi} \left[h(\sqrt{n}\,\widetilde{\boldsymbol{\theta}}^\pi)\right]  - 
    \E_{\mathbf{R},\pi} \left[h(\sqrt{n}\widehat{\mathbf{U}} {\mathbf D}^{1/2} (({\mathbf D}^\pi)^{1/2})^+  \mathbf{R} \widehat{\mathbf{R}}^\pi {\mathbf U}^{\pi\prime} \widehat{\boldsymbol{\theta}}^\pi)\right] \right| \label{eq:part1}
    \\&\quad + \sup\limits_{h\in BL_1(\mathbb R^r)} \left| 
    \E_{\mathbf{R},\pi} \left[h(\sqrt{n}\widehat{\mathbf{U}} {\mathbf D}^{1/2} (({\mathbf D}^\pi)^{1/2})^+  \mathbf{R} \widehat{\mathbf{R}}^\pi {\mathbf U}^{\pi\prime}  \widehat{\boldsymbol{\theta}}^\pi)\right] -
    \E_{\mathbf{R},\pi,\mathbf{Y}} \left[h(\widehat{\mathbf{U}} {\mathbf D}^{1/2} (({\mathbf D}^\pi)^{1/2})^+  \mathbf{R}  {\mathbf U}^{\pi\prime}  \mathbf{Y})\right]  \right| \label{eq:part2}
    \\&\quad +  \sup\limits_{h\in BL_1(\mathbb R^r)} \left| 
    \E_{\mathbf{R},\pi,\mathbf{Y}} \left[h(\widehat{\mathbf{U}} {\mathbf D}^{1/2} (({\mathbf D}^\pi)^{1/2})^+  \mathbf{R}  {\mathbf U}^{\pi\prime}  \mathbf{Y})\right] -
    \E \left[h( \mathbf{Z})\right]\right| .  \label{eq:part3} 
    \end{align}
    Now, we show that all terms converge to $0$ in outer $Q_1$-probability.    
    For \eqref{eq:part1}, let $\varepsilon > 0$ be arbitrary and write
    \begin{align*}
        &\sup\limits_{h\in BL_1(\mathbb R^r)} \left| \E_{\mathbf{R},\pi} \left[h(\sqrt{n}\,\widetilde{\boldsymbol{\theta}}^\pi)\right]  - 
    \E_{\mathbf{R},\pi} \left[h(\sqrt{n}\,\widehat{\mathbf{U}} {\mathbf D}^{1/2} (({\mathbf D}^\pi)^{1/2})^+  \mathbf{R} \widehat{\mathbf{R}}^\pi {\mathbf U}^{\pi\prime}  \widehat{\boldsymbol{\theta}}^\pi)\right] \right|\\
        &\leq \varepsilon + 2 {Q}_2\left( \sqrt{n}\| \widetilde{\boldsymbol{\theta}}^\pi - \widehat{\mathbf{U}} {\mathbf D}^{1/2} (({\mathbf D}^\pi)^{1/2})^+ \mathbf{R} \widehat{\mathbf{R}}^\pi{\mathbf U}^{\pi\prime}\widehat{\boldsymbol{\theta}}^\pi \|^* > \varepsilon  \right)
        \\&\leq \varepsilon + 2 Q_2\left( \| \widehat{\mathbf D}^{1/2} ((\widehat{\mathbf D}^\pi)^{1/2})^+ \mathbf{R} \widehat{\mathbf U}^{\pi\prime} - {\mathbf D}^{1/2} (({\mathbf D}^\pi)^{1/2})^+ \mathbf{R} \widehat{\mathbf{R}}^\pi {\mathbf U}^{\pi\prime} \|^* \|\sqrt{n}\, \widehat{\boldsymbol{\theta}}^\pi\|^* > \varepsilon  \right)
    \end{align*}
    since $\widehat{\mathbf{U}}$ is orthogonal, where $\| \mathbf x\|$ denotes the Euclidean norm of a vector $\mathbf x$ and $\| \mathbf A\|$ denotes the spectral norm of a matrix $\mathbf A$.
Then, we have for all $C>0$
\begin{align*}
    &Q_1\left(\left(\sup\limits_{h\in BL_1(\mathbb R^r)} \left| \E_{\mathbf{R},\pi} \left[h(\sqrt{n}\,\widetilde{\boldsymbol{\theta}}^\pi)\right]  - 
    \E_{\mathbf{R},\pi} \left[h(\sqrt{n}\widehat{\mathbf{U}} {\mathbf D}^{1/2} (({\mathbf D}^\pi)^{1/2})^+  \mathbf{R} \widehat{\mathbf{R}}^\pi {\mathbf U}^{\pi\prime}  \widehat{\boldsymbol{\theta}}^\pi)\right] \right|\right)^*  > 2\varepsilon \right)
    \\&\leq Q_1\left(2Q_2\left( \| \widehat{\mathbf D}^{1/2} ((\widehat{\mathbf D}^\pi)^{1/2})^+ \mathbf{R} \widehat{\mathbf U}^{\pi\prime} - {\mathbf D}^{1/2} (({\mathbf D}^\pi)^{1/2})^+ \mathbf{R} \widehat{\mathbf{R}}^\pi {\mathbf U}^{\pi\prime} \|^* \|\sqrt{n}\, \widehat{\boldsymbol{\theta}}^\pi\|^* > \varepsilon  \right)  > \varepsilon \right)
    \\&\leq 2 (Q_1\otimes Q_2)\left( \| \widehat{\mathbf D}^{1/2} ((\widehat{\mathbf D}^\pi)^{1/2})^+ \mathbf{R} \widehat{\mathbf U}^{\pi\prime} - {\mathbf D}^{1/2} (({\mathbf D}^\pi)^{1/2})^+ \mathbf{R} \widehat{\mathbf{R}}^\pi {\mathbf U}^{\pi\prime} \|^* \|\sqrt{n}\, \widehat{\boldsymbol{\theta}}^\pi\|^* > \varepsilon  \right)  / \varepsilon 
    \\&\leq \tfrac{2}{\varepsilon}\cdot\left({(Q_1\otimes Q_2)\left( \|\sqrt{n}\, \widehat{\boldsymbol{\theta}}^\pi\|^* > C \right) + (Q_1\otimes Q_2)\left( \| \widehat{\mathbf D}^{1/2} ((\widehat{\mathbf D}^\pi)^{1/2})^+ \mathbf{R} \widehat{\mathbf U}^{\pi\prime} - {\mathbf D}^{1/2} (({\mathbf D}^\pi)^{1/2})^+ \mathbf{R} \widehat{\mathbf{R}}^\pi {\mathbf U}^{\pi\prime} \|^*  > \tfrac{\varepsilon}{C} \right)}\right) .
\end{align*}
Since $\sqrt{n}\,\widehat{\boldsymbol{\theta}}^\pi$ conditionally converges weakly in probability by Assumption~\ref{the_ass}.2, it also converges weakly unconditionally by Lemma~1 in \cite{munko2024conditionaldeltamethodresamplingempirical}.
Hence, $\sqrt{n}\,\widehat{\boldsymbol{\theta}}^\pi$ is asymptotically tight, i.e., for $\delta > 0$ there exists a $C > 0$ such that $\limsup_{n\to\infty}(Q_1 \otimes Q_2)(\|\sqrt{n}\, \widehat{\boldsymbol{\theta}}^\pi\|^* > C)$ is smaller than $\delta$.
For the second probability, we first note that, by Assumption~\ref{distinctass}.2, $\| \widehat{\mathbf D}^{1/2} ((\widehat{\mathbf D}^\pi)^{1/2})^+ \mathbf{R} \widehat{\mathbf U}^{\pi\prime} - \widehat{\mathbf D}^{1/2} ((\widehat{\mathbf D}^\pi)^{1/2})^+ \widetilde{\mathbf{I}}\, \widetilde{\mathbf{I}} \mathbf{R} \widehat{\mathbf U}^{\pi\prime} \|$ converges to $0$ in outer probability, where $\widetilde{\mathbf{I}}$ is the diagonal matrix with entries $1$ for the first $\widetilde{r}$ diagonal elements and $0$ for the last diagonal elements.
Now, we use that, by Assumption~\ref{the_ass}.3 and \ref{distinctass}.1, 
$\widehat{\mathbf D}^{1/2} ((\widehat{\mathbf D}^\pi)^{1/2})^+ \widetilde{\mathbf I}\mathbf{R}  \xrightarrow{P} {\mathbf D}^{1/2} (({\mathbf D}^\pi)^{1/2})^+ \mathbf{R} $ as $n\to\infty$ follows from the continuous mapping theorem. Here, the first $\widetilde{r}$ diagonal elements of  $((\widehat{\mathbf D}^\pi)^{1/2})^+ $ converge in outer probability to the strictly positive first $\widetilde{r}$ diagonal elements of $(({\mathbf D}^\pi)^{1/2})^+$. The last diagonal elements are $0$, which results on the left side from multiplying by $\widetilde{\mathbf I}$ and on the right side from $\mathrm{rank}(\boldsymbol{\Sigma}^\pi) = \widetilde{r}$.
Hence, it suffices to show that $||\widetilde{\mathbf I}\widehat{\mathbf U}^{\pi\prime} - \widetilde{\mathbf I}\widehat{\mathbf{R}}^\pi {\mathbf U}^{\pi\prime}||^* = ||\widetilde{\mathbf I}\widehat{\mathbf{R}}^\pi \widehat{\mathbf U}^{\pi\prime} - \widetilde{\mathbf I}{\mathbf U}^{\pi\prime}||^* \xrightarrow{P} 0$ as $n\to\infty$.
Therefore, we consider the $i$th row of the matrix $\widehat{\mathbf{R}}^\pi \widehat{\mathbf U}^{\pi\prime} - {\mathbf U}^{\pi\prime}$ for all $i \in \{1,...,\widetilde{r}\}$.
We can apply Corollary~1 in \cite{yu2015} to obtain
\begin{align}
    \label{eq:ithvec}
    || \mathrm{sgn}(\widehat{\mathbf u}_i^{\pi\prime}\mathbf{u}_i^\pi)\widehat{\mathbf u}_i^{\pi\prime} -  {\mathbf u}_i^{\pi\prime} || \leq \frac{2^{3/2} || \widehat{\boldsymbol{\Sigma}}^\pi - {\boldsymbol{\Sigma}}^\pi ||}{\min\{\mathbf{D}_{(i-1)(i-1)}^\pi - \mathbf{D}_{ii}^\pi,\mathbf{D}_{ii}^\pi - \mathbf{D}_{(i+1)(i+1)}^\pi\}}
\end{align}
for all $i \in\{1,...,\widetilde{r}\}$, where the signum function ensures the condition $\mathrm{sgn}(\widehat{\mathbf u}_i^{\pi\prime}\mathbf{u}_i^\pi)\widehat{\mathbf u}_i^{\pi\prime} {\mathbf u}_i^{\pi} \geq 0$ of Corollary~1. The right side in \eqref{eq:ithvec} converges in outer probability to $0$ by Assumptions~\ref{the_ass}.3 and~\ref{distinctass}.3 for all $i \in\{1,...,\widetilde{r}\}$.
Thus, we can conclude that $\| \widehat{\mathbf D}^{1/2} ((\widehat{\mathbf D}^\pi)^{1/2})^+ \mathbf{R} \widehat{\mathbf U}^{\pi\prime} - {\mathbf D}^{1/2} (({\mathbf D}^\pi)^{1/2})^+ \mathbf{R} \widehat{\mathbf{R}}^\pi {\mathbf U}^{\pi\prime} \| $ converges to $0$ in outer probability and, further, that
$$\limsup_{n\to\infty} Q_1\left(\left(\sup\limits_{h\in BL_1(\mathbb R^r)} \left| \E_{\mathbf{R},\pi} \left[h(\sqrt{n}\,\widetilde{\boldsymbol{\theta}}^\pi)\right]  - 
    \E_{\mathbf{R},\pi} \left[h(\sqrt{n}\,\widehat{\mathbf{U}} {\mathbf D}^{1/2} (({\mathbf D}^\pi)^{1/2})^+  \mathbf{R} \widehat{\mathbf{R}}^\pi {\mathbf U}^{\pi\prime}  \widehat{\boldsymbol{\theta}}^\pi)\right] \right|\right)^*  > 2\varepsilon \right) $$ is bounded by $2{\delta}/{\varepsilon}.$
Since $\delta > 0$ was arbitrary, this proves that
\eqref{eq:part1} converges to $0$ in outer $Q_1$-probability.
Next, we consider \eqref{eq:part2}.
By Fubini's theorem, we get rid of $\widehat{\mathbf R}^\pi$ by
\begin{align*}
   & \E_{\mathbf R, \pi}\left[ h\left(
   \sqrt{n}\widehat{\mathbf{U}} {\mathbf D}^{1/2} (({\mathbf D}^\pi)^{1/2})^+  \mathbf{R} \widehat{\mathbf{R}}^\pi {\mathbf U}^{\pi\prime}  \widehat{\boldsymbol{\theta}}^\pi\right)\right] 
   \\&= \frac{1}{2^r} \sum_{s_1,...,s_r \in \{-1,1\}} \E_{\mathbf R, \pi}\left[ h\left(
   \sqrt{n}\widehat{\mathbf{U}} {\mathbf D}^{1/2} (({\mathbf D}^\pi)^{1/2})^+  \mathrm{diag}(s_1,...,s_r) \widehat{\mathbf{R}}^\pi {\mathbf U}^{\pi\prime}  \widehat{\boldsymbol{\theta}}^\pi\right)\right]
   \\&=   \E_{\mathbf R, \pi}\left[ \frac{1}{2^r} \sum_{s_1,...,s_r \in \{-1,1\}}h\left(
   \sqrt{n}\widehat{\mathbf{U}} {\mathbf D}^{1/2} (({\mathbf D}^\pi)^{1/2})^+  \mathrm{diag}(s_1,...,s_r) \widehat{\mathbf{R}}^\pi {\mathbf U}^{\pi\prime}  \widehat{\boldsymbol{\theta}}^\pi\right)\right]
   \\&=  \E_{\mathbf R, \pi}\left[ \frac{1}{2^r} \sum_{s_1,...,s_r \in \{-1,1\}}h\left(
   \sqrt{n}\widehat{\mathbf{U}} {\mathbf D}^{1/2} (({\mathbf D}^\pi)^{1/2})^+  \mathrm{diag}(s_1,...,s_r)  {\mathbf U}^{\pi\prime}  \widehat{\boldsymbol{\theta}}^\pi\right)\right]
\\&=\frac{1}{2^r} \sum_{s_1,...,s_r \in \{-1,1\}} \E_{\mathbf R, \pi}\left[ h\left(
   \sqrt{n}\widehat{\mathbf{U}} {\mathbf D}^{1/2} (({\mathbf D}^\pi)^{1/2})^+  \mathrm{diag}(s_1,...,s_r)  {\mathbf U}^{\pi\prime}  \widehat{\boldsymbol{\theta}}^\pi\right)\right]
   \\&=\E_{\mathbf R, \pi}\left[ h\left(
   \sqrt{n}\widehat{\mathbf{U}} {\mathbf D}^{1/2} (({\mathbf D}^\pi)^{1/2})^+  \mathbf{R} {\mathbf U}^{\pi\prime}  \widehat{\boldsymbol{\theta}}^\pi\right)\right] .
\end{align*}
Then, we can use that the function $\mathbf x \mapsto h(\widehat{\mathbf{U}} {\mathbf D}^{1/2} (({\mathbf D}^\pi)^{1/2})^+  \mathbf x)$ only depends on $\Omega_1$ and takes values in $BL:=BL_{\max\{1,\|\mathbf D^{1/2}((\mathbf D^\pi)^{1/2})^+\|\}}(\mathbb R^r)$ to obtain
\begin{align*}
   & \sup\limits_{h\in BL_1(\mathbb R^r)} \left| 
    \E_{\mathbf{R},\pi} \left[h(\sqrt{n}\widehat{\mathbf{U}} {\mathbf D}^{1/2} (({\mathbf D}^\pi)^{1/2})^+  \mathbf{R} {\mathbf U}^{\pi\prime}  \widehat{\boldsymbol{\theta}}^\pi)\right] -
    \E_{\mathbf{R},\pi,\mathbf{Y}} \left[h(\widehat{\mathbf{U}} {\mathbf D}^{1/2} (({\mathbf D}^\pi)^{1/2})^+  \mathbf{R}  {\mathbf U}^{\pi\prime}  \mathbf{Y})\right]  \right|
    \\&\leq \sup\limits_{h\in BL} \left| 
    \E_{\mathbf{R},\pi} \left[h(\sqrt{n}\mathbf{R}  {\mathbf U}^{\pi\prime}\widehat{\boldsymbol{\theta}}^\pi)\right] -
    \E_{\mathbf{R},\pi,\mathbf{Y}} \left[h(\mathbf{R}  {\mathbf U}^{\pi\prime}\mathbf{Y})\right]  \right|
    \\&=   \sup\limits_{h\in BL} \left| \frac{1}{2^r} \sum_{s_1,...,s_r \in \{-1,1\}}
    \E_{\mathbf{R},\pi} \left[h(\sqrt{n}\mathrm{diag}(s_1,...,s_r)  {\mathbf U}^{\pi\prime}\widehat{\boldsymbol{\theta}}^\pi)\right] -\frac{1}{2^r} \sum_{s_1,...,s_r \in \{-1,1\}}
    \E_{\mathbf{R},\pi,\mathbf{Y}} \left[h(\mathrm{diag}(s_1,...,s_r)  {\mathbf U}^{\pi\prime}\mathbf{Y})\right]  \right|
    \\&\leq \sup\limits_{h\in BL} \left| 
    \E_{\mathbf{R},\pi} \left[h(\sqrt{n}\,\widehat{\boldsymbol{\theta}}^\pi)\right] -
    \E_{\mathbf{R},\pi,\mathbf{Y}} \left[h(\mathbf{Y})\right]  \right|.
\end{align*}
By Assumption~\ref{the_ass}.2, the latter converges in outer $Q_1$-probability to $0$, which provides that \eqref{eq:part2} converges in outer $Q_1$-probability to $0$.
For \eqref{eq:part3}, we write
\begin{align*}
    &\sup\limits_{h\in BL_1(\mathbb R^r)} \left| 
    \E_{\mathbf{R},\pi,\mathbf{Y}} \left[h(\widehat{\mathbf{U}} {\mathbf D}^{1/2} (({\mathbf D}^\pi)^{1/2})^+  \mathbf{R}  {\mathbf U}^{\pi\prime}  \mathbf{Y})\right] -
    \E \left[h( \mathbf{Z})\right]\right|
    \\&= \sup\limits_{h\in BL_1(\mathbb R^r)} \left| 
     \sum_{s_1,...,s_r \in \{-1,1\}}
     \frac{ \E_{\mathbf{R},\pi,\mathbf{Y}} \left[h(\widehat{\mathbf{U}} {\mathbf D}^{1/2} (({\mathbf D}^\pi)^{1/2})^+  \mathrm{diag}(s_1,...,s_r)  {\mathbf U}^{\pi\prime}  \mathbf{Y})\right]}{2^r} -
    \E \left[h( \mathbf{Z})\right]\right|
     \\&\leq \sum_{s_1,...,s_r \in \{-1,1\}}\frac{1}{2^r}
    \sup\limits_{h\in BL_1(\mathbb R^r)} \left| 
    \E_{\mathbf{R},\pi,\mathbf{Y}} \left[h(\widehat{\mathbf{U}} {\mathbf D}^{1/2} (({\mathbf D}^\pi)^{1/2})^+  \mathrm{diag}(s_1,...,s_r)  {\mathbf U}^{\pi\prime}  \mathbf{Y})\right]-
    \E \left[h( \mathbf{Z})\right]\right|
        \\&= 
    \sup\limits_{h\in BL_1(\mathbb R^r)} \left| 
    \int_{\mathbb R^r} h(\mathbf{y})\;\mathrm{d}\mathcal{N}_r(\mathbf{0},\widehat{\mathbf{U}} \mathbf{D}\widehat{\mathbf{U}}')(\mathbf{y}) -
    \int h( \mathbf{z})\;\mathrm{d}\mathcal{N}_r(\mathbf{0}, \boldsymbol{\Sigma})(\mathbf{z})\right|.
\end{align*}
Note that the latter expression converges in outer $Q_1$-probability to $0$ whenever $\widehat{\mathbf{U}} \mathbf{D}\widehat{\mathbf{U}}'$ converges in outer $Q_1$-probability to $\boldsymbol{\Sigma}$.
We have
\begin{align*}
    \|\widehat{\mathbf{U}} \mathbf{D}\widehat{\mathbf{U}}' - \boldsymbol{\Sigma}\| &\leq \|\widehat{\mathbf{U}} \mathbf{D}\widehat{\mathbf{U}}' - \widehat{\mathbf{U}} \widehat{\mathbf{D}}\widehat{\mathbf{U}}'\| +  \|\widehat{\boldsymbol{\Sigma}} - \boldsymbol{\Sigma}\|
    \\&= \|\mathbf{D} - \widehat{\mathbf{D}}\| +  \|\widehat{\boldsymbol{\Sigma}} - \boldsymbol{\Sigma}\|.
\end{align*}
Here, both summands converge in outer $Q_1$-probability to $0$ by Assumption~\ref{the_ass}.3 and, thus, \eqref{eq:part3} is proved.
Hence, putting \eqref{eq:part1}--\eqref{eq:part3} together, we obtain \eqref{eq:precise}.
Proving \eqref{eq:precise2} is much simpler:
First, $\sqrt{n}\,\widehat{\boldsymbol{\theta}}^\pi$ is asymptotically measurable (unconditionally) by Lemma~1 in \cite{munko2024conditionaldeltamethodresamplingempirical}.
Furthermore, $\mathbf{R}$ is measurable and, thus, asymptotically measurable.
Next, $\widehat{\mathbf{U}}, \widehat{\mathbf{U}}^\pi$ are asymptotically measurable by Assumption~\ref{distinctass}.4.
Finally, $\widehat{\mathbf{D}}^{1/2} ((\widehat{\mathbf{D}}^\pi)^{1/2})^+$ converges in outer probability (unconditionally) by the continuous mapping theorem and Assumptions~\ref{the_ass}.3 and \ref{distinctass}.2 and, thus, is asymptotically measurable by Lemma~1.3.8 in \cite{vaartWellner2023}.
Now, we can conclude that $\sqrt{n}\,\widetilde{\boldsymbol{\theta}}^\pi$ is asymptotically measurable (unconditionally) by Lemma~1.4.4 and Problem~1.3.7 in \cite{vaartWellner2023}.
Hence, it follows
\begin{align*}
    Q_1\left(\E_{\mathbf{R}, \pi}\left[ h\left(\sqrt{n}\,\widetilde{\boldsymbol{\theta}}^\pi \right)^* \right] - \E_{\mathbf{R}, \pi} \left[h\left( \sqrt{n}\,\widetilde{\boldsymbol{\theta}}^\pi \right)_*\right]  > \varepsilon\right)
    &\leq \E\left[\E_{\mathbf{R}, \pi}\left[ h\left(\sqrt{n}\,\widetilde{\boldsymbol{\theta}}^\pi \right)^* \right] - \E_{\mathbf{R}, \pi} \left[h\left( \sqrt{n}\,\widetilde{\boldsymbol{\theta}}^\pi \right)_*\right]\right]/\varepsilon
    \\&= \left( \E\left[ h\left(\sqrt{n}\,\widetilde{\boldsymbol{\theta}}^\pi \right)^* \right] - \E \left[h\left( \sqrt{n}\,\widetilde{\boldsymbol{\theta}}^\pi \right)_*\right] \right)/\varepsilon \to 0
\end{align*} for all $\varepsilon>0, h \in BL_1(\mathbb R^r)$, which finishes the proof.
\end{proof}
 \begin{proof}[Proof of Theorem~\ref{conv_thm}]
 Let $\mathbf R := \left(\bigoplus_{j=0}^{J^*} \mathbf R_{i^*_{j+1}-i^*_j;j}\right)$ with $i^*_1, ..., i^*_{J^*}\in\{2,...,r\}, J^*\in\{0,...,r-1\}$ such that $i^*_0 := 1 < i^*_1 < ... < i^*_{J^*} < r+1 =: i^*_{J^*+1}$ and $\mathbf{D}^\pi_{11} = ...=\mathbf{D}^\pi_{(i^*_1-1)(i^*_1-1)} > \mathbf{D}^\pi_{i^*_1i^*_1} = ...=  \mathbf{D}^\pi_{(i^*_2-1)(i^*_2-1)} > ... > \mathbf{D}^\pi_{i^*_{J^*}i^*_{J^*}} = ...=  \mathbf{D}^\pi_{rr}$.
 First, we show that $ \mathbf{R}_{\varepsilon} \xrightarrow{P} \widetilde{\mathbf I}  \mathbf R$.
 For $\widetilde{r} = r$, it remains to show that the events $I \neq \{i^*_1, ..., i^*_{J^*}\}$ and $r_n \widehat{\mathbf D}^\pi_{rr} \leq \varepsilon$ have outer probability tending to $0$, respectively.
 For $\widetilde{r} < r$, we show that the events $I \neq \{i^*_1, ..., i^*_{J^*}\}$ and $r_n \widehat{\mathbf D}^\pi_{rr} > \varepsilon$ have outer probability tending to $0$, respectively. 
 For $P = (Q_1\otimes Q_2)^*$ being the outer probability w.r.t. $\Omega_1\times\Omega_2$ under the notation as in the proof of Theorem~\ref{distinct_thm}, it holds by Weyl's inequality
 \begin{align*}
     &P\left(I \neq \{i^*_1, ..., i^*_{J^*}\}\right) \\
     &\leq \sum_{i\in\{i_1^*,...,i_{J^*}^*\}}P\left(i \not\in I\right) + \sum_{i\in\{2,...,r\}\setminus\{i_1^*,...,i_{J^*}^*\}}P\left(i \in I\right)
     \\&\leq \sum_{i\in\{i_1^*,...,i_{J^*}^*\}}P\left(r_n (\widehat{\mathbf{D}}^\pi_{(i-1)(i-1)}-\widehat{\mathbf{D}}^\pi_{ii}) \leq \varepsilon\right) + \sum_{i\in\{2,...,r\}\setminus\{i_1^*,...,i_{J^*}^*\}}P\left(r_n (\widehat{\mathbf{D}}^\pi_{(i-1)(i-1)}-\widehat{\mathbf{D}}^\pi_{ii}) > \varepsilon\right)
     \\&\leq \sum_{i\in\{i_1^*,...,i_{J^*}^*\}}\left(\mathbbm{1}\{r_n ({\mathbf{D}}^\pi_{(i-1)(i-1)}-{\mathbf{D}}^\pi_{ii}) \leq 2\varepsilon\} + \mathbbm{1}\{r_n ({\mathbf{D}}^\pi_{(i-1)(i-1)}-{\mathbf{D}}^\pi_{ii}) > 2\varepsilon\} \cdot
     \right.\\&\left.\qquad\qquad\qquad\quad \left(P\left(r_n |\widehat{\mathbf{D}}^\pi_{(i-1)(i-1)}-{\mathbf{D}}^\pi_{(i-1)(i-1)}| > \varepsilon/2\right)+P\left(r_n |{\mathbf{D}}^\pi_{ii}-\widehat{\mathbf{D}}^\pi_{ii}| > \varepsilon/2\right)\right)
     \right)\\&\quad+ \sum_{i\in\{2,...,r\}\setminus\{i_1^*,...,i_{J^*}^*\}}\left(P\left(r_n |\widehat{\mathbf{D}}^\pi_{(i-1)(i-1)}-{\mathbf{D}}^\pi_{(i-1)(i-1)}| > \varepsilon/2\right) + P\left(r_n |{\mathbf{D}}^\pi_{ii}-\widehat{\mathbf{D}}^\pi_{ii}| > \varepsilon/2\right)\right)
     \\&\leq \sum_{i\in\{i_1^*,...,i_{J^*}^*\}}\mathbbm{1}\{r_n ({\mathbf{D}}^\pi_{(i-1)(i-1)}-{\mathbf{D}}^\pi_{ii}) \leq 2\varepsilon\} + 2(r-1) P\left(r_n \|\widehat{\boldsymbol{\Sigma}}^\pi - {\boldsymbol{\Sigma}}^\pi \| > \varepsilon/2\right) \to 0
 \end{align*} as $n\to\infty$.
 For the remaining events, we have 
 $$P(r_n \widehat{\mathbf D}^\pi_{rr} \leq \varepsilon) \leq P(r_n |\widehat{\mathbf D}^\pi_{rr} - {\mathbf D}^\pi_{rr}| > \varepsilon) + \mathbbm{1}\{r_n {\mathbf D}^\pi_{rr} \leq 2\varepsilon\} \to 0$$ if ${\mathbf D}^\pi_{rr} > 0$ (i.e., $\widetilde{r} = r$)
 and
 $$P(r_n \widehat{\mathbf D}^\pi_{rr} > \varepsilon) \leq P(r_n |\widehat{\mathbf D}^\pi_{rr} - {\mathbf D}^\pi_{rr}| > \varepsilon) \to 0$$ as $n\to\infty$ by Assumption~\ref{conv_ass}.3 if ${\mathbf D}^\pi_{rr} = 0$ (i.e., $\widetilde{r} < r$).
Thus, $ \mathbf{R}_{\varepsilon} \xrightarrow{P}\widetilde{\mathbf I}\mathbf R$ follows.
     Now, we follow the lines of the proof of Theorem~\ref{distinct_thm}. 
    We start to prove \eqref{eq:precise} by defining $\mathbf{Y} \sim \mathcal N(\mathbf 0, \boldsymbol{\Sigma}^\pi)$. 
    Moreover, for $\widehat{\mathbf U}^\pi = (\widehat{\mathbf{u}}_1^\pi,...,\widehat{\mathbf{u}}_r^\pi), \mathbf{U}^\pi = (\mathbf{u}_1^\pi,...,\mathbf{u}_r^\pi)$ and $\widehat{\mathbf U}_j^{\pi} := (\widehat{\mathbf u}_{i^*_j}^{\pi},...,\widehat{\mathbf u}_{i^*_{j+1}-1}^{\pi}), {\mathbf U}_j^{\pi} := ({\mathbf u}_{i^*_j}^{\pi},...,{\mathbf u}_{i^*_{j+1}-1}^{\pi})$, 
    we define $\widehat{\mathbf{R}}^\pi := \bigoplus_{j=1}^{J^*} \widehat{\mathbf{O}}_{j}' $, where $\widehat{\mathbf{O}}_{j}$ is the orthogonal matrix from Theorem~2 in \cite{yu2015} fulfilling
    $$||\widehat{\mathbf U}_j^{\pi}\widehat{\mathbf{O}}_{j} -  {\mathbf U}_j^{\pi} || \leq \frac{2^{3/2}r^{1/2} || \widehat{\boldsymbol{\Sigma}}^\pi - {\boldsymbol{\Sigma}}^\pi ||}{\min\{\mathbf{D}_{(i^*_j-1)(i^*_j-1)}^\pi - \mathbf{D}_{i^*_ji^*_j}^\pi,\mathbf{D}_{(i^*_{j+1}-1)(i^*_{j+1}-1)}^\pi - \mathbf{D}_{i^*_{j+1}i^*_{j+1}}^\pi\}}$$ for all $j\in\{1,...,J^*\}$.
    Then, we divide \eqref{eq:precise} again into the terms \eqref{eq:part1}--\eqref{eq:part3} and prove that all terms converge to $0$ in outer probability.
        For \eqref{eq:part1}, we have, by $P(\hat{i} \neq r - \widetilde{r}) = P(i_J \neq i_{J^*}^*) \to 0$, that 
    $\| \widehat{\mathbf D}^{1/2} ((\widehat{\mathbf D}^\pi)^{1/2})^+ \mathbf{R}_\varepsilon \widehat{\mathbf U}^{\pi\prime} - \widehat{\mathbf D}^{1/2} ((\widehat{\mathbf D}^\pi)^{1/2})^+ \widetilde{\mathbf I} \mathbf{R}_\varepsilon \widehat{\mathbf U}^{\pi\prime} \|$ converges to $0$ in outer probability. 
    By $ \mathbf{R}_{\varepsilon} \xrightarrow{P} \widetilde{\mathbf I}  \mathbf R$, we get
    $\| \widehat{\mathbf D}^{1/2} ((\widehat{\mathbf D}^\pi)^{1/2})^+ \widetilde{\mathbf I} \mathbf{R}_\varepsilon \widehat{\mathbf U}^{\pi\prime} - \widehat{\mathbf D}^{1/2} ((\widehat{\mathbf D}^\pi)^{1/2})^+ \widetilde{\mathbf I}\mathbf R   \widehat{\mathbf U}^{\pi\prime} \| \xrightarrow{P} 0$.
Moreover, we obtain 
$\widehat{\mathbf D}^{1/2} ((\widehat{\mathbf D}^\pi)^{1/2})^+ \widetilde{\mathbf I} \xrightarrow{P} {\mathbf D}^{1/2} (({\mathbf D}^\pi)^{1/2})^+  \widetilde{\mathbf I} ={\mathbf D}^{1/2} (({\mathbf D}^\pi)^{1/2})^+$ as $n\to\infty$. 
Hence, it suffices to show that $||\widetilde{\mathbf I}\widehat{\mathbf U}^{\pi\prime} - \widetilde{\mathbf I}\widehat{\mathbf{R}}^\pi {\mathbf U}^{\pi\prime}||^* = ||\widetilde{\mathbf I}\widehat{\mathbf{R}}^\pi \widehat{\mathbf U}^{\pi\prime} - \widetilde{\mathbf I}{\mathbf U}^{\pi\prime}||^* \xrightarrow{P} 0$ as $n\to\infty$.
Therefore, we consider the $j$th row-wise block of the matrix $\widehat{\mathbf{R}}^\pi \widehat{\mathbf U}^{\pi\prime} - {\mathbf U}^{\pi\prime}$ for all $j\in\{1,...,J^*\}$ with $i^*_{j+1} \leq \widetilde{r}+1$.
For all such $j$, we can apply Theorem~2 in \cite{yu2015} to obtain
\begin{align}
    \label{eq:ithvec2}
    ||\widehat{\mathbf{O}}_{j}'\widehat{\mathbf U}_j^{\pi\prime} -  {\mathbf U}_j^{\pi\prime} ||= ||\widehat{\mathbf U}_j^{\pi}\widehat{\mathbf{O}}_{j} -  {\mathbf U}_j^{\pi} || \leq \frac{2^{3/2}r^{1/2} || \widehat{\boldsymbol{\Sigma}}^\pi - {\boldsymbol{\Sigma}}^\pi ||}{\min\{\mathbf{D}_{(i^*_j-1)(i^*_j-1)}^\pi - \mathbf{D}_{i^*_ji^*_j}^\pi,\mathbf{D}_{(i^*_{j+1}-1)(i^*_{j+1}-1)}^\pi - \mathbf{D}_{i^*_{j+1}i^*_{j+1}}^\pi\}}
\end{align} by definition of $\widehat{\mathbf{O}}_{j}$.
The right side in \eqref{eq:ithvec2} converges in outer probability to $0$ by Assumptions~\ref{the_ass}.3 and the definition of $i^*_1,...,i^*_{J^*+1}$.
Thus, we can conclude that $\| \widehat{\mathbf D}^{1/2} ((\widehat{\mathbf D}^\pi)^{1/2})^+ \mathbf{R}_\varepsilon \widehat{\mathbf U}^{\pi\prime} - {\mathbf D}^{1/2} (({\mathbf D}^\pi)^{1/2})^+ \mathbf{R} \widehat{\mathbf{R}}^\pi {\mathbf U}^{\pi\prime} \| $ converges to $0$ in outer probability and, further, that
\eqref{eq:part1} converges to $0$ in outer $Q_1$-probability by doing the same steps as in the proof of Theorem~\ref{distinct_thm}.
Next, we consider \eqref{eq:part2}.
By Fubini's theorem, we get rid of $\widehat{\mathbf R}^\pi$ by
\begin{align*}
   & \E_{\mathbf R, \pi}\left[ h\left(
   \sqrt{n}\widehat{\mathbf{U}} {\mathbf D}^{1/2} (({\mathbf D}^\pi)^{1/2})^+  \mathbf{R} \widehat{\mathbf{R}}^\pi {\mathbf U}^{\pi\prime}  \widehat{\boldsymbol{\theta}}^\pi\right)\right] 
   \\&=   \E_{\mathbf R, \pi}\left[ \int_{O(i^*_1-i^*_0)\times ... \times O(i^*_{J^*+1}-i^*_{J^*})} h\left(
   \sqrt{n}\widehat{\mathbf{U}} {\mathbf D}^{1/2} (({\mathbf D}^\pi)^{1/2})^+  \mathbf S \widehat{\mathbf{R}}^\pi {\mathbf U}^{\pi\prime}  \widehat{\boldsymbol{\theta}}^\pi\right)\;\mathrm{d}\mu(\mathbf S)\right]
   \\&=  \E_{\mathbf R, \pi}\left[\int_{O(i^*_1-i^*_0)\times ... \times O(i^*_{J^*+1}-i^*_{J^*})} h\left(
   \sqrt{n}\widehat{\mathbf{U}} {\mathbf D}^{1/2} (({\mathbf D}^\pi)^{1/2})^+  \mathbf S  {\mathbf U}^{\pi\prime}  \widehat{\boldsymbol{\theta}}^\pi\right)\;\mathrm{d}\mu(\mathbf S)\right]
   \\&=\E_{\mathbf R, \pi}\left[ h\left(
   \sqrt{n}\widehat{\mathbf{U}} {\mathbf D}^{1/2} (({\mathbf D}^\pi)^{1/2})^+  \mathbf{R} {\mathbf U}^{\pi\prime}  \widehat{\boldsymbol{\theta}}^\pi\right)\right] ,
\end{align*}
where $\mu$ denotes the product measure of the Haar measures on $O(i^*_j-i^*_{j-1}), j\in\{0,...,J^*\}$.
Then, we can use that $\mathbf x \mapsto h(\widehat{\mathbf{U}} {\mathbf D}^{1/2} (({\mathbf D}^\pi)^{1/2})^+  \mathbf x)$ only depends on $\Omega_1$ and takes values in $BL:=BL_{\max\{1,\|\mathbf D^{1/2}((\mathbf D^\pi)^{1/2})^+\|\}}(\mathbb R^r)$ to obtain
\begin{align*}
   & \sup\limits_{h\in BL_1(\mathbb R^r)} \left| 
    \E_{\mathbf{R},\pi} \left[h(\sqrt{n}\widehat{\mathbf{U}} {\mathbf D}^{1/2} (({\mathbf D}^\pi)^{1/2})^+  \mathbf{R} {\mathbf U}^{\pi\prime}  \widehat{\boldsymbol{\theta}}^\pi)\right] -
    \E_{\mathbf{R},\pi,\mathbf{Y}} \left[h(\widehat{\mathbf{U}} {\mathbf D}^{1/2} (({\mathbf D}^\pi)^{1/2})^+  \mathbf{R}  {\mathbf U}^{\pi\prime}  \mathbf{Y})\right]  \right|
    \\&\leq \sup\limits_{h\in BL} \left| 
    \E_{\mathbf{R},\pi} \left[h(\sqrt{n}\mathbf{R}  {\mathbf U}^{\pi\prime}\widehat{\boldsymbol{\theta}}^\pi)\right] -
    \E_{\mathbf{R},\pi,\mathbf{Y}} \left[h(\mathbf{R}  {\mathbf U}^{\pi\prime}\mathbf{Y})\right]  \right|
    \\&=   \sup\limits_{h\in BL} \left| \int_{O(i^*_1-i^*_0)\times ... \times O(i^*_{J^*+1}-i^*_{J^*})}
    \E_{\mathbf{R},\pi} \left[h(\sqrt{n}\mathbf S  {\mathbf U}^{\pi\prime}\widehat{\boldsymbol{\theta}}^\pi)\right]\;\mathrm{d}\mu(\mathbf S) -\int_{O(i^*_1-i^*_0)\times ... \times O(i^*_{J^*+1}-i^*_{J^*})}
    \E_{\mathbf{R},\pi,\mathbf{Y}} \left[h(\mathbf S  {\mathbf U}^{\pi\prime}\mathbf{Y})\right]\;\mathrm{d}\mu(\mathbf S)  \right|
    \\&\leq \sup\limits_{h\in BL} \left| 
    \E_{\mathbf{R},\pi} \left[h(\sqrt{n}\widehat{\boldsymbol{\theta}}^\pi)\right] -
    \E_{\mathbf{R},\pi,\mathbf{Y}} \left[h(\mathbf{Y})\right]  \right|.
\end{align*}
By Assumption~\ref{the_ass}.2, the latter converges in outer $Q_1$-probability to $0$, which provides that \eqref{eq:part2} converges in outer $Q_1$-probability to $0$.
For \eqref{eq:part3}, we write
\begin{align*}
    &\sup\limits_{h\in BL_1(\mathbb R^r)} \left| 
    \E_{\mathbf{R},\pi,\mathbf{Y}} \left[h(\widehat{\mathbf{U}} {\mathbf D}^{1/2} (({\mathbf D}^\pi)^{1/2})^+  \mathbf{R}  {\mathbf U}^{\pi\prime}  \mathbf{Y})\right] -
    \E \left[h( \mathbf{Z})\right]\right|
    \\&= \sup\limits_{h\in BL_1(\mathbb R^r)} \left| 
    \int_{O(i^*_1-i^*_0)\times ... \times O(i^*_{J^*+1}-i^*_{J^*})}
    \E_{\mathbf{R},\pi,\mathbf{Y}} \left[h(\widehat{\mathbf{U}} {\mathbf D}^{1/2} (({\mathbf D}^\pi)^{1/2})^+  \mathbf{S}  {\mathbf U}^{\pi\prime}  \mathbf{Y})\right]\;\mathrm{d}\mu(\mathbf S) -
      \E \left[h( \mathbf{Z})\right]\right|
    \\&\leq \int_{O(i^*_1-i^*_0)\times ... \times O(i^*_{J^*+1}-i^*_{J^*})}
    \sup\limits_{h\in BL_1(\mathbb R^r)} \left| 
    \E_{\mathbf{R},\pi,\mathbf{Y}} \left[h(\widehat{\mathbf{U}} {\mathbf D}^{1/2} (({\mathbf D}^\pi)^{1/2})^+  \mathbf{S}  {\mathbf U}^{\pi\prime}  \mathbf{Y})\right] -
     \E \left[h( \mathbf{Z})\right]\right|\;\mathrm{d}\mu(\mathbf S)
    \\&= 
    \sup\limits_{h\in BL_1(\mathbb R^r)} \left| 
    \int_{\mathbb R^r} h(\mathbf{y})\;\mathrm{d}\mathcal{N}_r(\mathbf{0},\widehat{\mathbf{U}} \mathbf{D}\widehat{\mathbf{U}}')(\mathbf{y}) -
    \int h( \mathbf{z})\;\mathrm{d}\mathcal{N}_r(\mathbf{0}, \boldsymbol{\Sigma})(\mathbf{z})\right|.
\end{align*}
The latter expression converges in outer $Q_1$-probability to $0$ as in the proof of Theorem~\ref{distinct_thm}.
Hence, putting \eqref{eq:part1}--\eqref{eq:part3} together, we obtain \eqref{eq:precise}.
The second equation \eqref{eq:precise2} follows in the same way as in the proof of Theorem~\ref{distinct_thm}.
\end{proof}
\newpage
\section{Further Simulation Results}\label{sec:simu}
\subsection{Univariate ANOVA}
\begin{figure}[h]
    \centering
    \includegraphics[width=0.8\linewidth,trim= 1mm 1mm 7mm 2mm,clip]{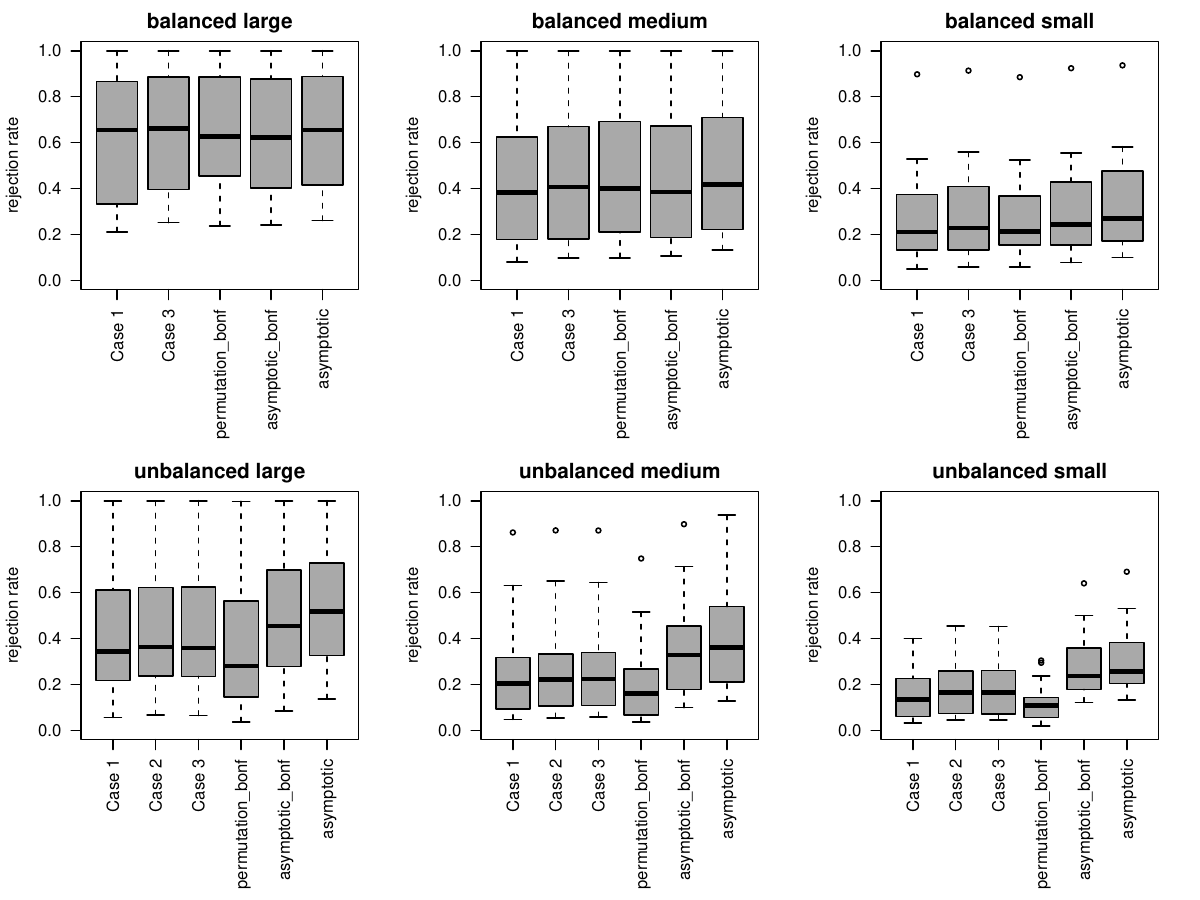}
    \caption{Empirical global power for the Dunnett-type contrast matrix}
    \label{fig:ANOVADunnpower2}
\end{figure}
\begin{figure}[h]
    \centering
    \includegraphics[width=0.8\linewidth,trim= 1mm 1mm 7mm 2mm,clip]{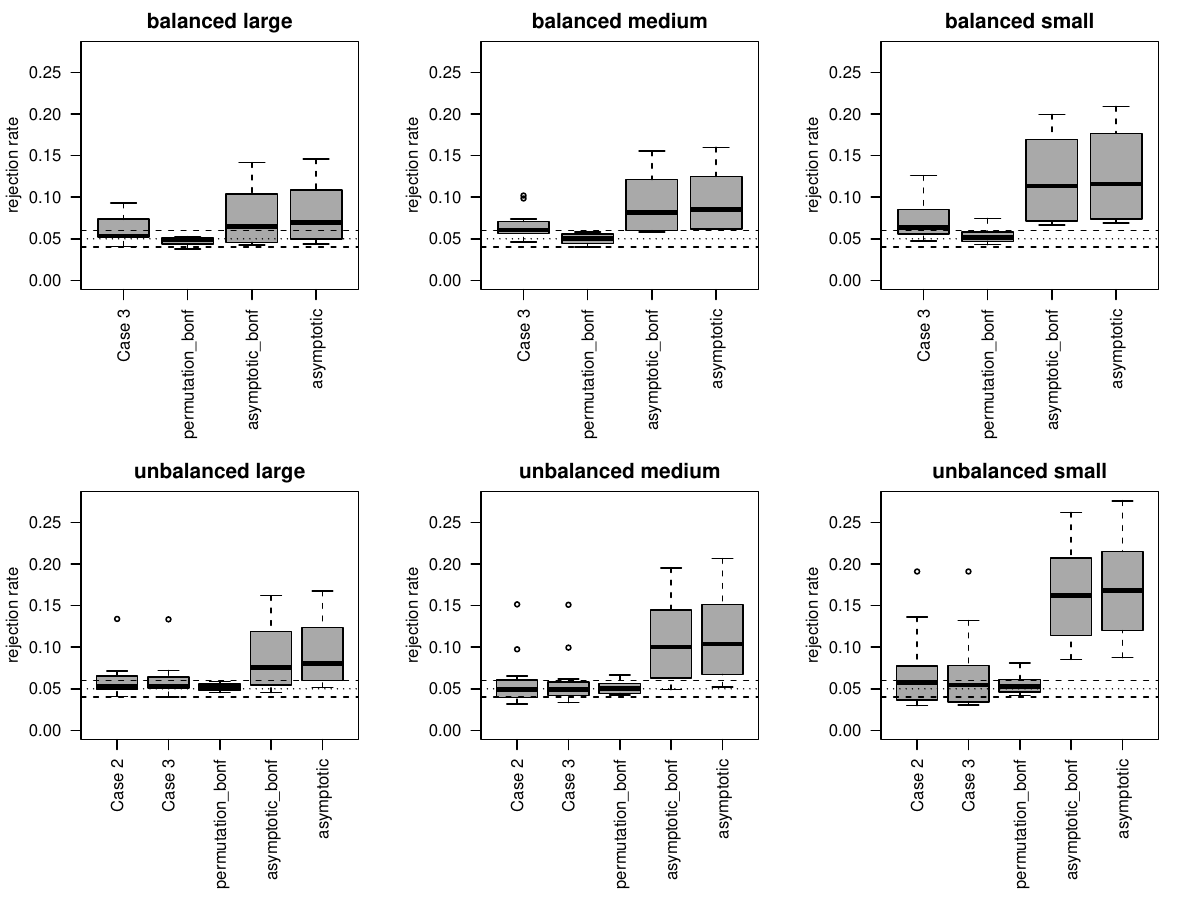}
    \caption{Empirical FWERs for the centering matrix. The dotted line represents the desired FWER of 5\% and the dashed lines represent the borders of the 95\% binomial confidence interval [0.0405, 0.06].}
    \label{fig:ANOVAGM}
\end{figure}
\begin{figure}[h]
    \centering
    \includegraphics[width=0.8\linewidth,trim= 1mm 1mm 7mm 2mm,clip]{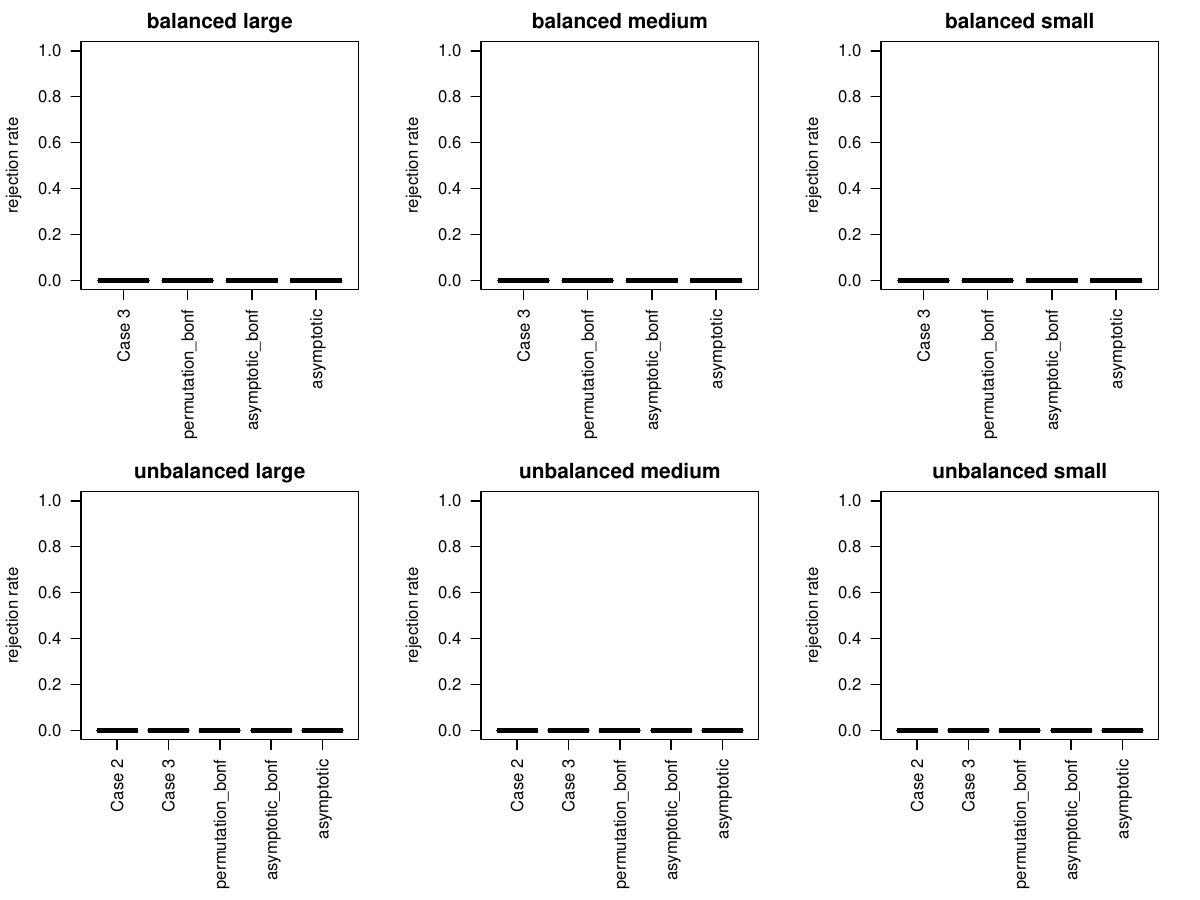}
    \caption{Empirical family-wise power for the centering matrix}
    \label{fig:ANOVAGMpower}
\end{figure}
\begin{figure}[h]
    \centering
    \includegraphics[width=0.8\linewidth,trim= 1mm 1mm 7mm 2mm,clip]{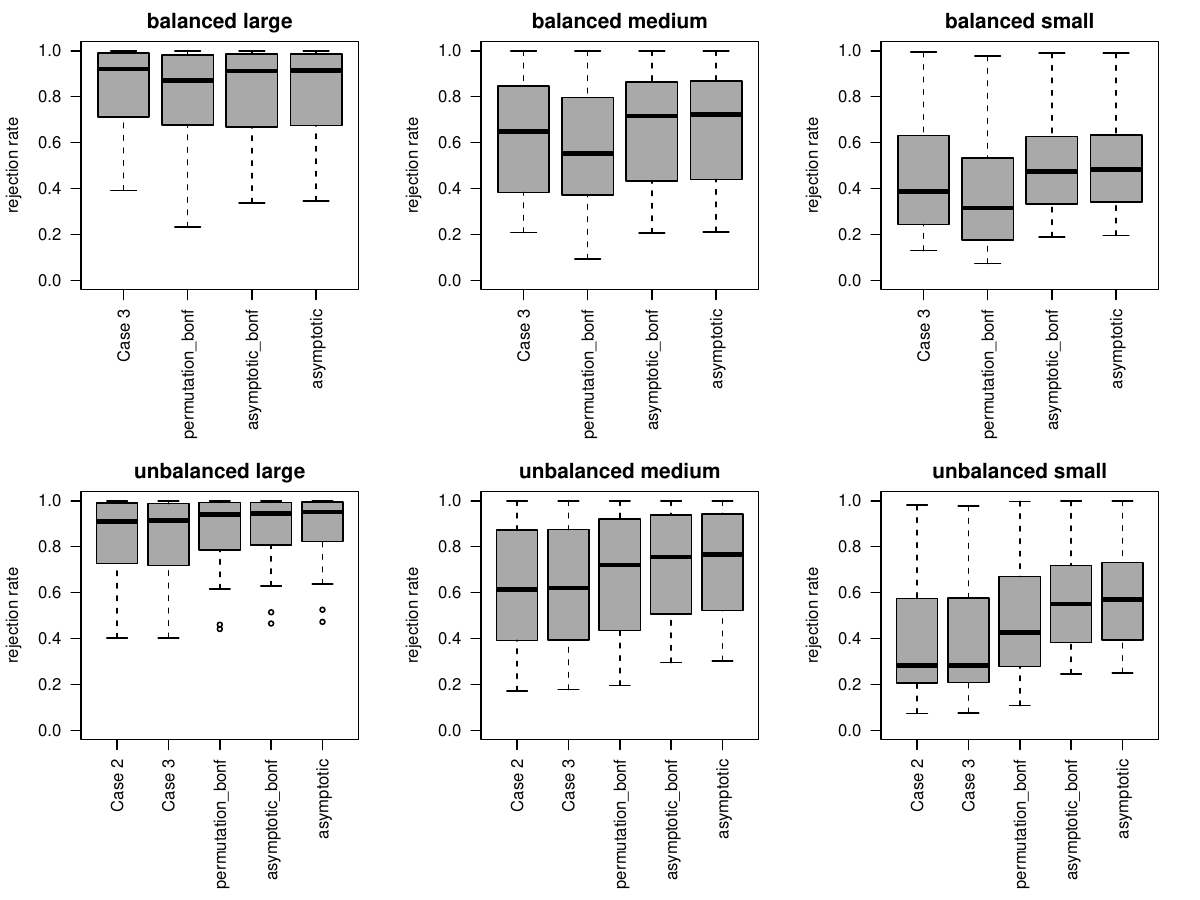}
    \caption{Empirical global power for the centering matrix}
    \label{fig:ANOVAGMpower2}
\end{figure}

\begin{figure}[h]
    \centering
    \includegraphics[width=0.8\linewidth,trim= 1mm 1mm 7mm 2mm,clip]{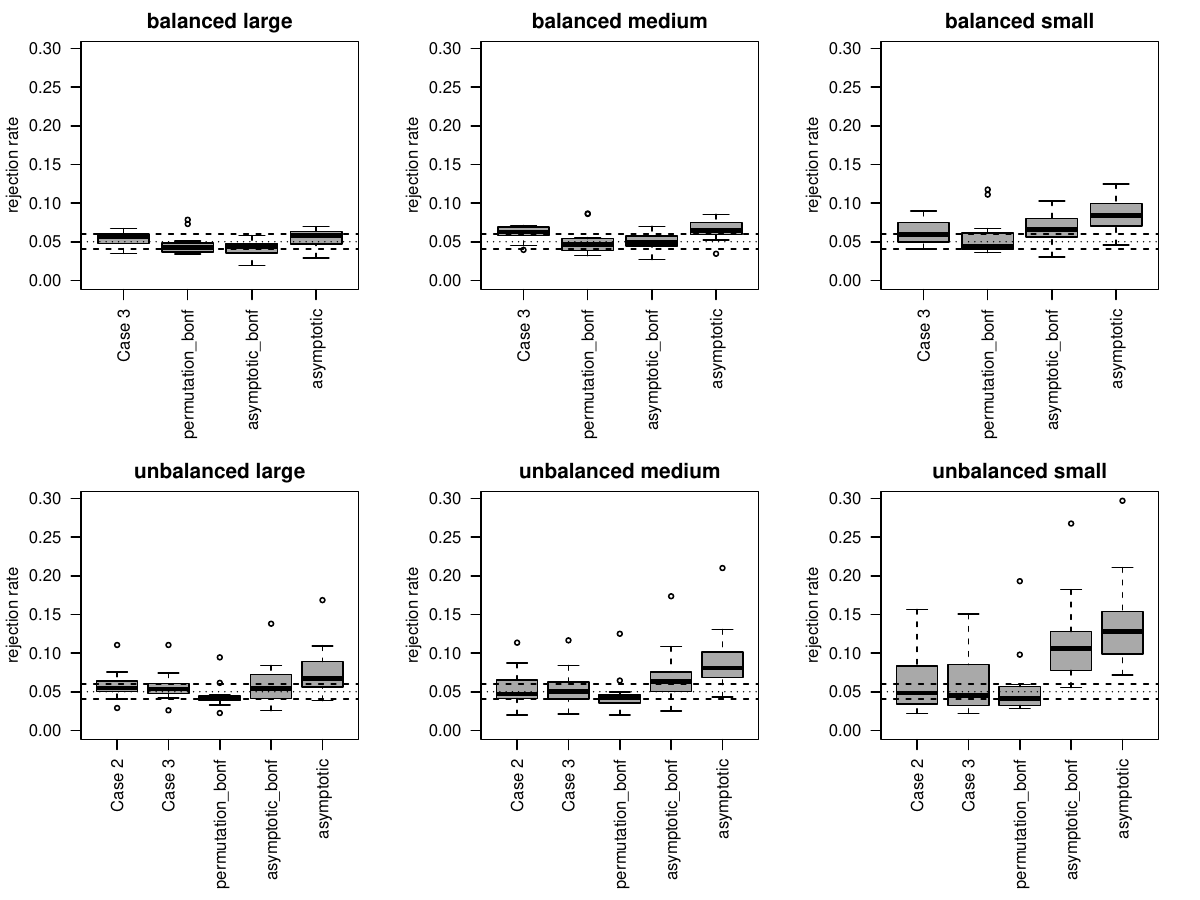}
    \caption{Empirical FWERs for the Tukey-type contrast matrix. The dotted line represents the desired FWER of 5\% and the dashed lines represent the borders of the 95\% binomial confidence interval [0.0405, 0.06].}
    \label{fig:ANOVATuk}
\end{figure}

\begin{figure}[h]
    \centering
    \includegraphics[width=0.8\linewidth,trim= 1mm 1mm 7mm 2mm,clip]{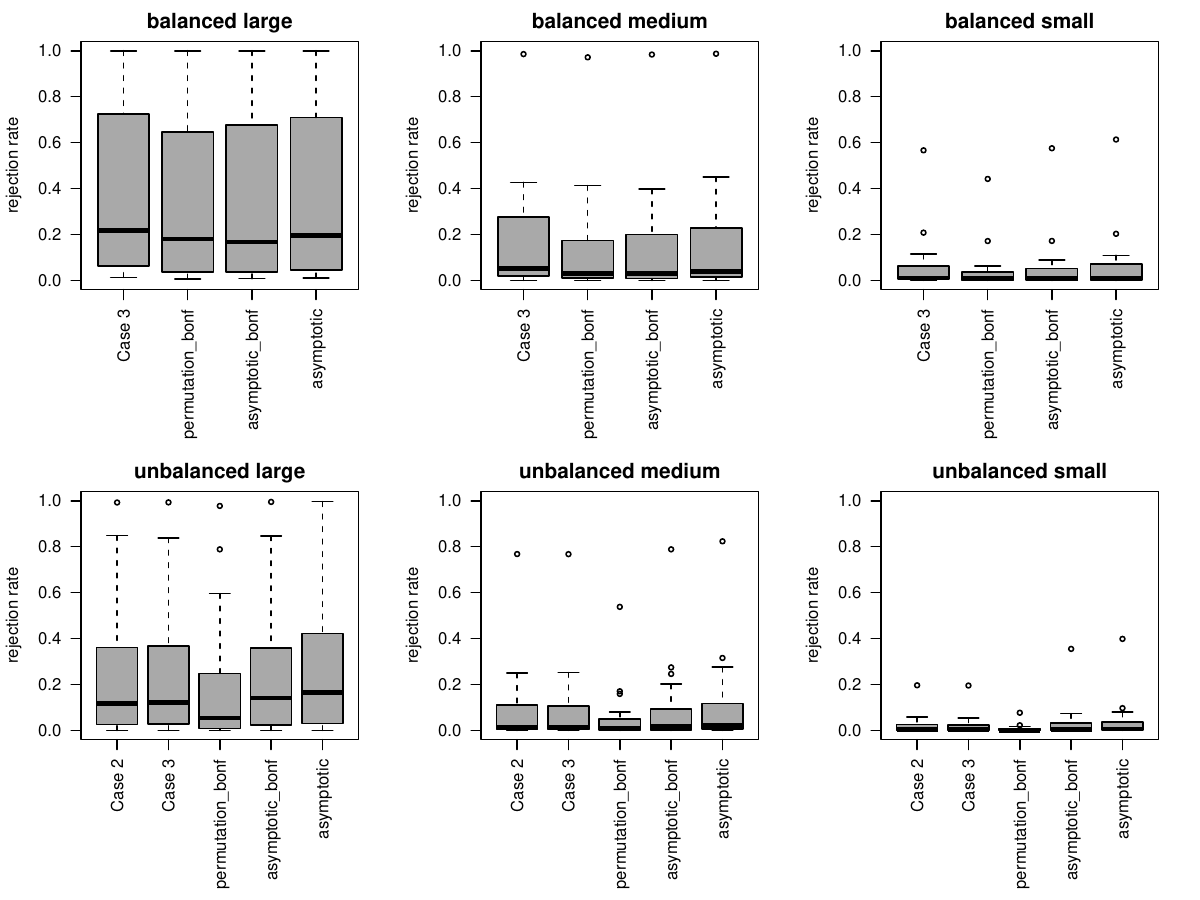}
    \caption{Empirical family-wise power for the Tukey-type contrast matrix}
    \label{fig:ANOVATukpower}
\end{figure}

\begin{figure}[h]
    \centering
    \includegraphics[width=0.8\linewidth,trim= 1mm 1mm 7mm 2mm,clip]{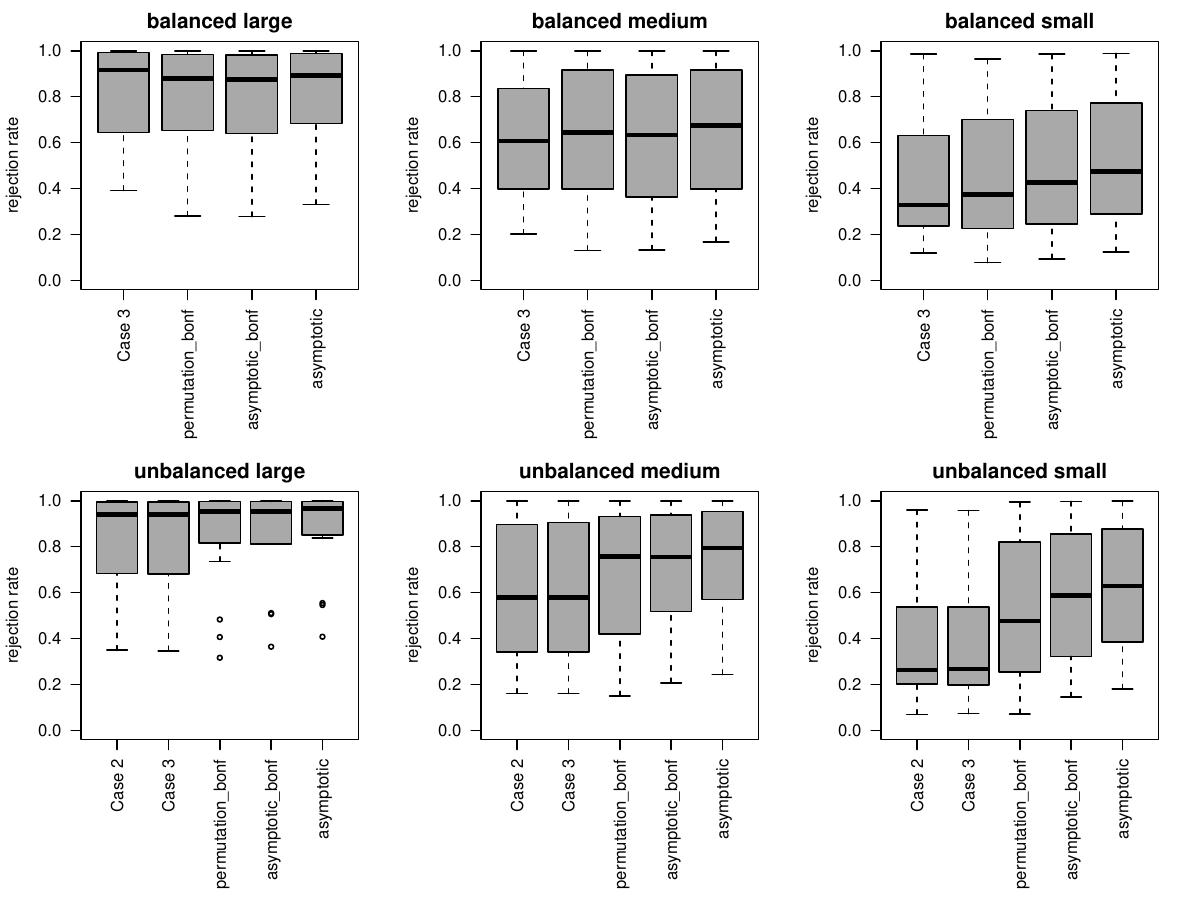}
    \caption{Empirical global power for the Tukey-type contrast matrix}
    \label{fig:ANOVATukpower2}
\end{figure}
\FloatBarrier
\subsection{Multivariate ANOVA}\FloatBarrier
\begin{table}[ht]
\centering
\begin{tabular}{|ll||rr|r|r|}
  \hline
\multicolumn{2}{|c||}{setting} & \multicolumn{2}{c|}{multiple permutation} & \multicolumn{1}{c|}{permutation} & \multicolumn{1}{c|}{asymptotic}\\
 sample size & distribution& Case 2 & Case 3 &  Bonferroni &  multiple \\ 
  \hline
  large & $Exp(1)$ & 7.25 & 6.65 & 
  \textbf{5.25} & 6.45 \\ 
medium & $Exp(1)$ & 7.10 & 7.50 & 
\textbf{5.40} & 6.65 \\ 
small & $Exp(1)$ & 6.85 & 7.35 & 
\textbf{5.95} & 8.50 \\ 
large &$\mathcal N(0,1)$& 7.55 & 7.15 & 
6.60 & 6.60 \\ 
medium &$\mathcal N(0,1)$& 8.10 & 7.95 & 
6.10 & 7.65 \\ 
small &$\mathcal N(0,1)$& 7.20 & 7.25 & 
6.10 & 9.40 \\ 
large & $t_9$ & \textbf{5.70} & \textbf{5.95} & 
\textbf{5.25} & \textbf{5.50} \\ 
medium & $t_9$ & 7.55 & 7.50 & 
\textbf{5.35} & 7.00 \\ 
small & $t_9$ & 8.15 & 7.80 & 
6.05 & 8.85 \\ 
   \hline
\end{tabular}
\caption{Empirical FWERs in \% for the ANOVA-type test statistics. The values in the 95\% binomial confidence interval [4.05, 6] are printed in bold type.}
\label{tab:MANOVAFWERATS}
\end{table}

\begin{table}[ht]
\centering
\begin{tabular}{|rll||rr|r|r|}
  \hline
\multicolumn{3}{|c||}{setting}  & \multicolumn{2}{c|}{multiple permutation} & \multicolumn{1}{c|}{permutation} & \multicolumn{1}{c|}{asymptotic}\\
 $\delta$ &sample size & distribution& Case 2 & Case 3 & Bonferroni &  multiple \\ 
  \hline
   -1.5 & large & $Exp(1)$ & 68.35 & \textbf{69.20} & 67.85 & 64.25 \\ 
  -1.5 &  medium & $Exp(1)$ & \textbf{35.15} & 35.00 & 31.55 &  31.75 \\ 
  -1.5 & small & $Exp(1)$ & 12.50 & 12.60 & 11.25 &  \textbf{13.70} \\ 
  1.5 & large & $Exp(1)$ & \textbf{66.30} & 66.25 & 62.90 &  59.95 \\ 
  1.5 & medium & $Exp(1)$  & 26.65 & \textbf{27.20} & 17.85 &  19.90 \\ 
  1.5 & small & $Exp(1)$  & 6.90 & \textbf{7.10} & 2.70 &  4.95 \\ 
  1.5 & large & $\mathcal N(0,1)$ & \textbf{67.60} & 67.00 & 65.75 &  63.00 \\ 
   1.5 & medium & $\mathcal N(0,1)$ & 26.90 & \textbf{27.80} & 22.00 & 22.90 \\ 
   1.5 & small & $\mathcal N(0,1)$ & 8.30 & 8.70 & 5.95 & \textbf{8.85} \\ 
   1.5 & large & $t_9$ & 65.15 & \textbf{65.55} & 64.10 & 61.95 \\ 
   1.5 & medium & $t_9$ & 24.55 & \textbf{24.95} & 19.85 & 21.05 \\
   1.5 & small & $t_9$ & 9.10 & \textbf{9.25} & 6.65 & 9.00 \\ 
   \hline
\end{tabular}
\caption{Empirical family-wise power in \% for the ANOVA-type test statistics. The largest values per setting are printed in bold type.}
\label{tab:MANOVAallpowerATS}
\end{table}

\begin{table}[ht]
\centering
\begin{tabular}{|rll||rr|r|r|}
  \hline
\multicolumn{3}{|c||}{setting} & \multicolumn{2}{c|}{multiple permutation} & \multicolumn{1}{c|}{permutation} & \multicolumn{1}{c|}{asymptotic}\\
 $\delta$ & sample size & distribution& Case 2 & Case 3 & Bonferroni & multiple \\ 
  \hline
-1.5 & large & $Exp(1)$  & 93.20 & \textbf{93.35} & 
91.90 & 90.10 \\ 
  -1.5 & medium & $Exp(1)$ & 70.55 & \textbf{70.75} & 
64.45 & 65.05 \\ 
  -1.5 & small & $Exp(1)$ & 45.65 & \textbf{45.85} & 
39.10 & 44.50 \\ 
  1.5 & large & $Exp(1)$ & \textbf{98.45} & 98.30 &
98.15 & 97.70 \\ 
  1.5 & medium & $Exp(1)$ & 78.20 & \textbf{78.55} & 
74.65 & 73.85 \\ 
  1.5 & small & $Exp(1)$ & \textbf{43.90} & 43.70 & 
39.15 & 43.25 \\ 
  1.5 & large &$\mathcal N(0,1)$& 95.50 & \textbf{95.70} & 
95.15 & 94.90 \\ 
  1.5 & medium &$\mathcal N(0,1)$& 72.00 & \textbf{72.30} & 
66.55 & 68.85 \\ 
  1.5 & small &$\mathcal N(0,1)$& \textbf{43.90} & 43.65 & 
36.40 & 43.60 \\ 
  1.5 & large & $t_9$ & 94.95 & \textbf{95.05} & 
94.75 & 93.80 \\ 
  1.5 & medium & $t_9$ & 70.90 & \textbf{71.45} & 
65.55 & 66.35 \\ 
  1.5 & small & $t_9$ & 41.50 & 41.05 & 
35.75 & \textbf{41.55} \\ 
   \hline
\end{tabular}
\caption{Empirical global power in \% for the ANOVA-type test statistics. The largest values per setting are printed in bold type.}
\label{tab:MANOVApowerATS}
\end{table}
\clearpage
\subsection{Comparison of RMSTs}
\begin{figure}[h]
    \centering
    \includegraphics[width=0.8\linewidth,trim= 1mm 1mm 7mm 2mm,clip]{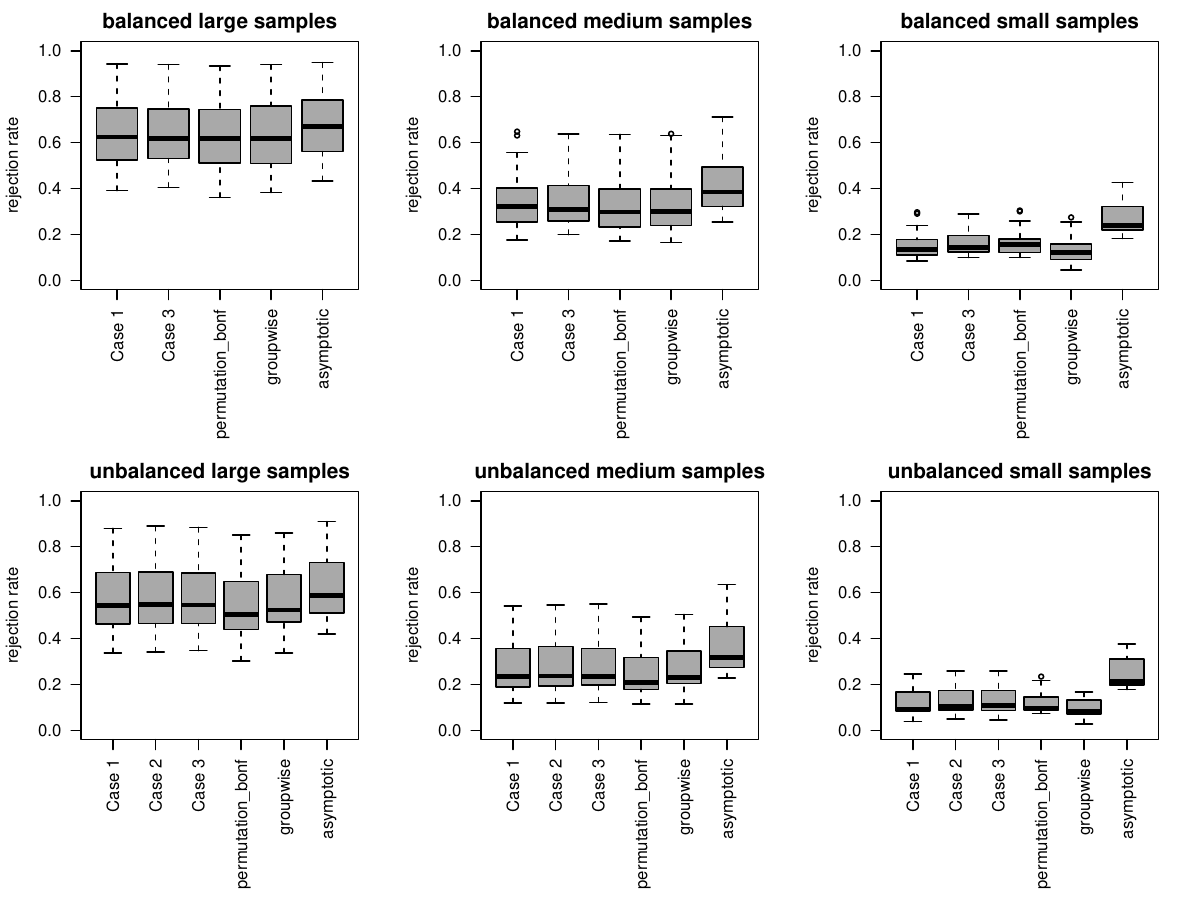}
    \caption{Empirical global power for the Dunnett-type contrast matrix}
    \label{fig:RMSTDunnpower2}
\end{figure}

\begin{figure}[h]
    \centering
    \includegraphics[width=0.8\linewidth,trim= 1mm 1mm 7mm 2mm,clip]{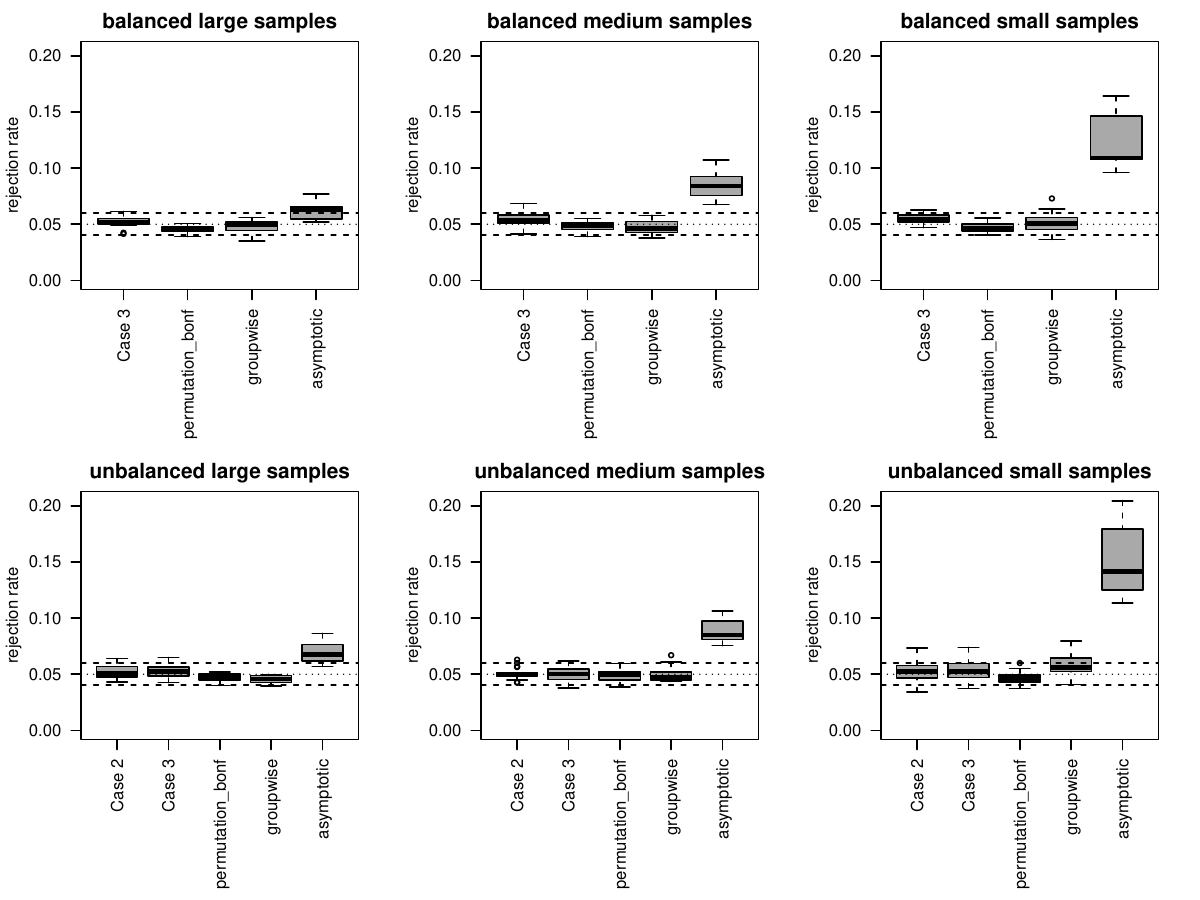}
    \caption{Empirical FWERs for the centering matrix. The dotted line represents the desired FWER of 5\% and the dashed lines represent the borders of the 95\% binomial confidence interval [0.0405, 0.06].}
    \label{fig:RMSTGM}
\end{figure}

\begin{figure}[h]
    \centering
    \includegraphics[width=0.8\linewidth,trim= 1mm 1mm 7mm 2mm,clip]{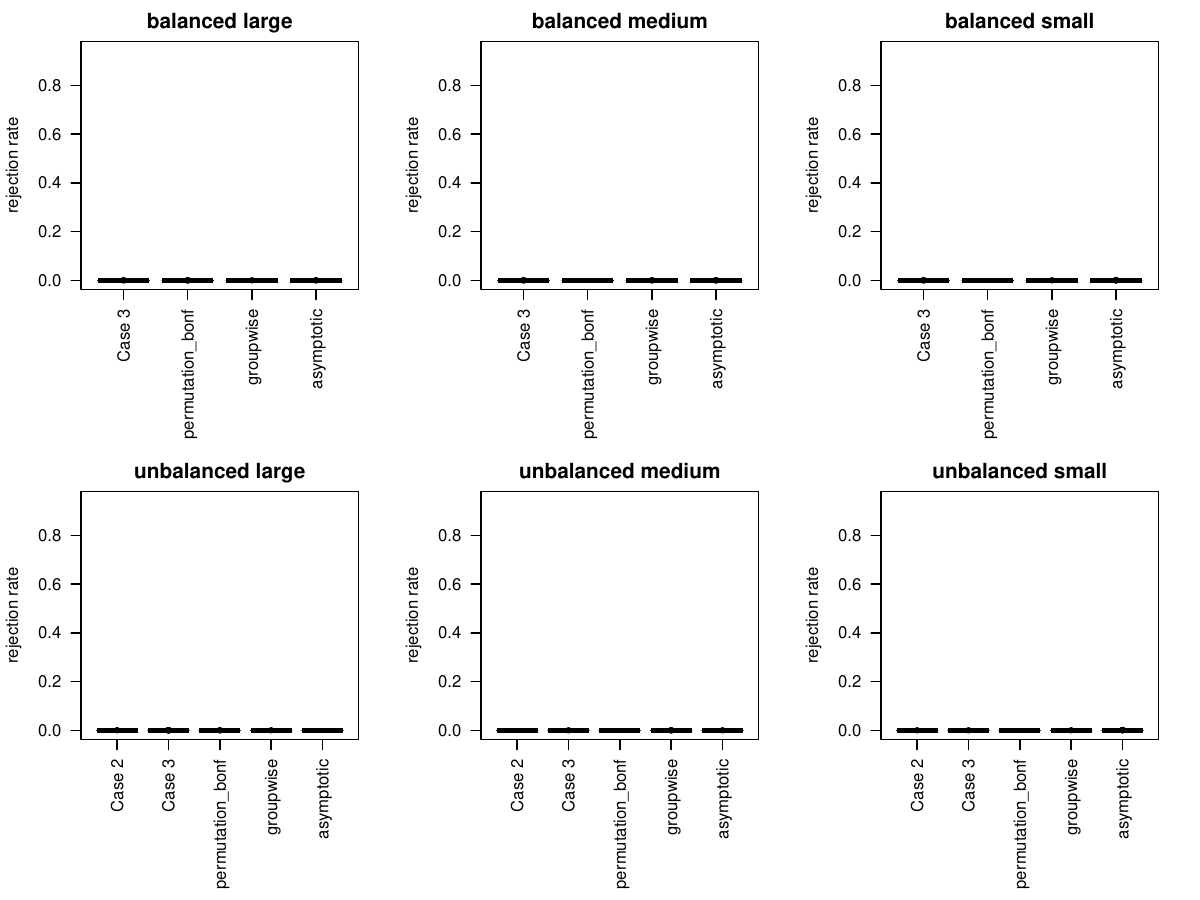}
    \caption{Empirical family-wise power for the centering matrix}
    \label{fig:RMSTGMpower}
\end{figure}

\begin{figure}[h]
    \centering
    \includegraphics[width=0.8\linewidth,trim= 1mm 1mm 7mm 2mm,clip]{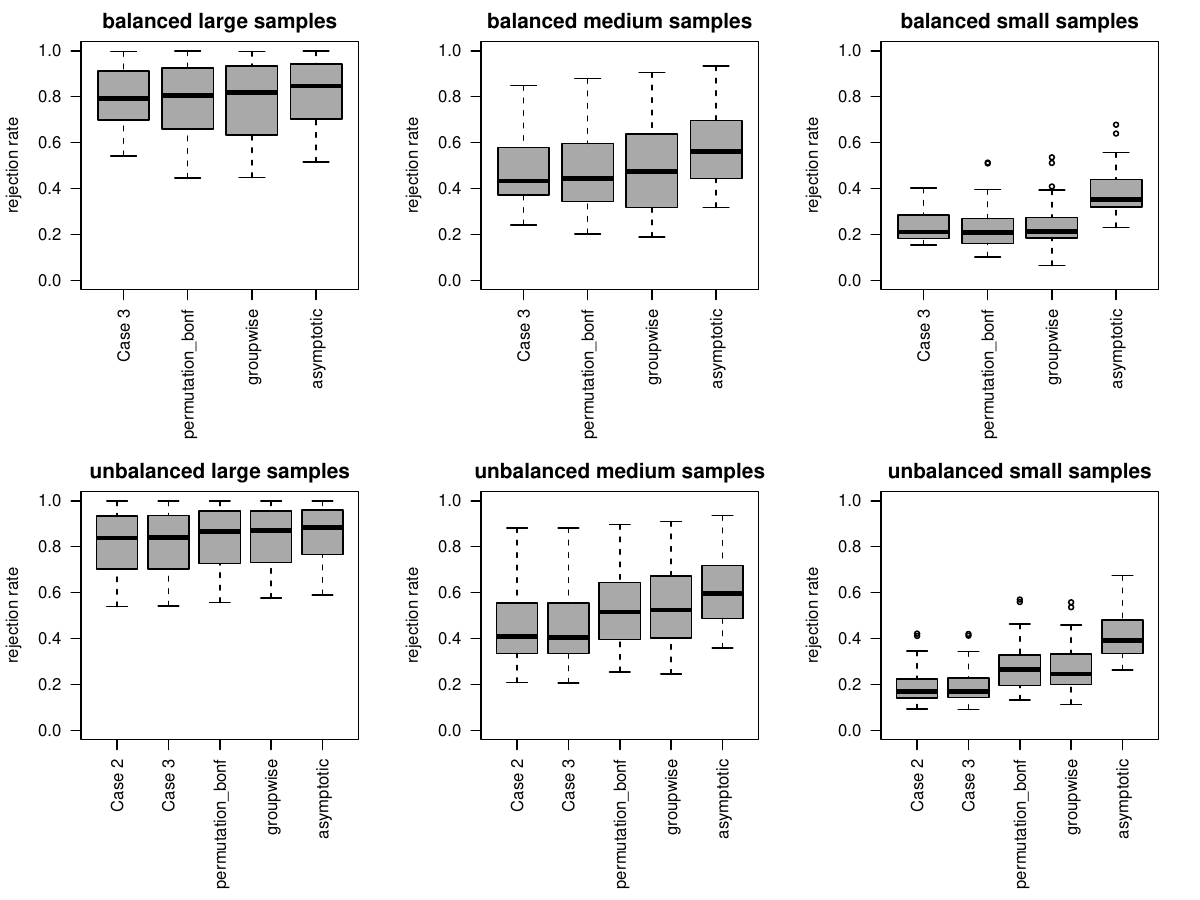}
    \caption{Empirical global power for the centering matrix}
    \label{fig:RMSTGMpower2}
\end{figure}
\begin{figure}[h]
    \centering
    \includegraphics[width=0.8\linewidth,trim= 1mm 1mm 7mm 2mm,clip]{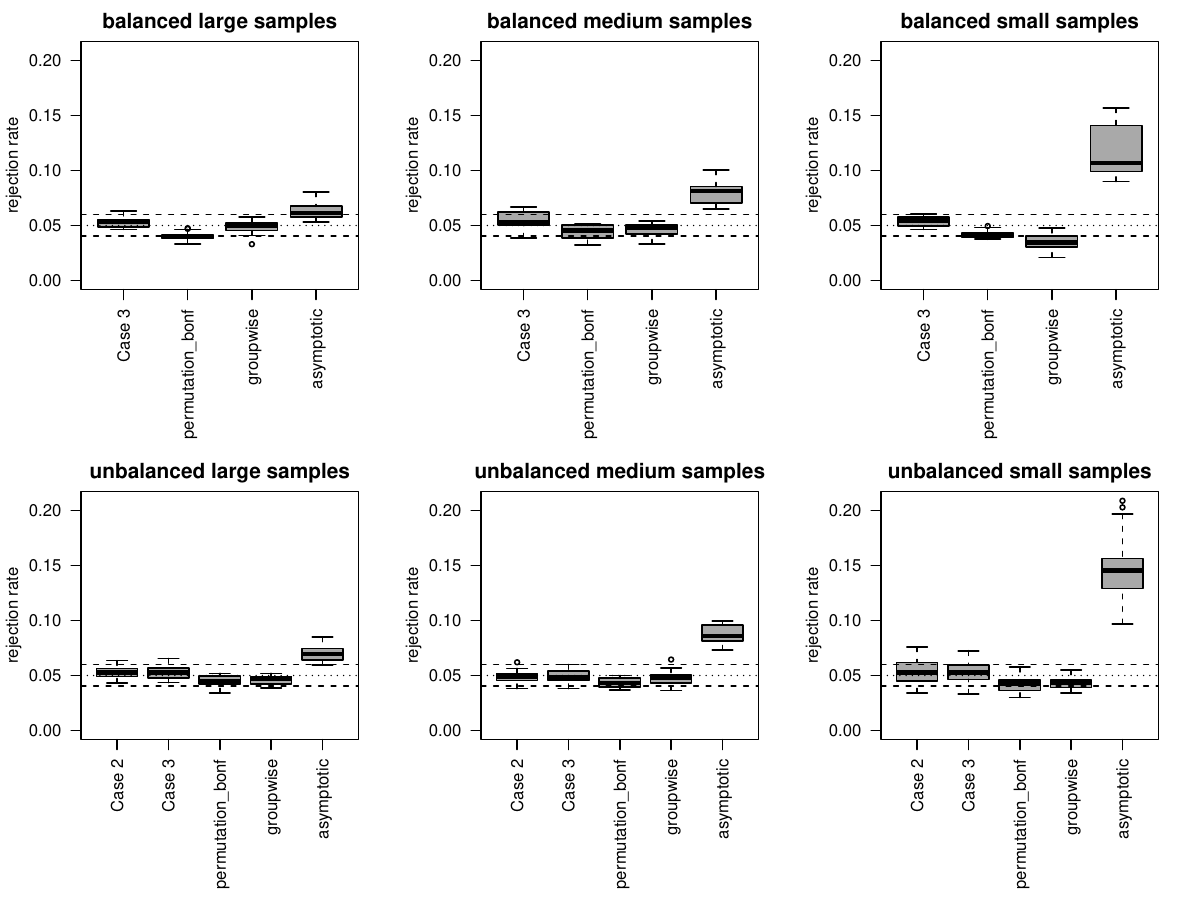}
    \caption{Empirical FWERs for the Tukey-type contrast matrix. The dotted line represents the desired FWER of 5\% and the dashed lines represent the borders of the 95\% binomial confidence interval [0.0405, 0.06].}
    \label{fig:RMSTTuk}
\end{figure}
\begin{figure}[h]
    \centering
    \includegraphics[width=0.8\linewidth,trim= 1mm 1mm 7mm 2mm,clip]{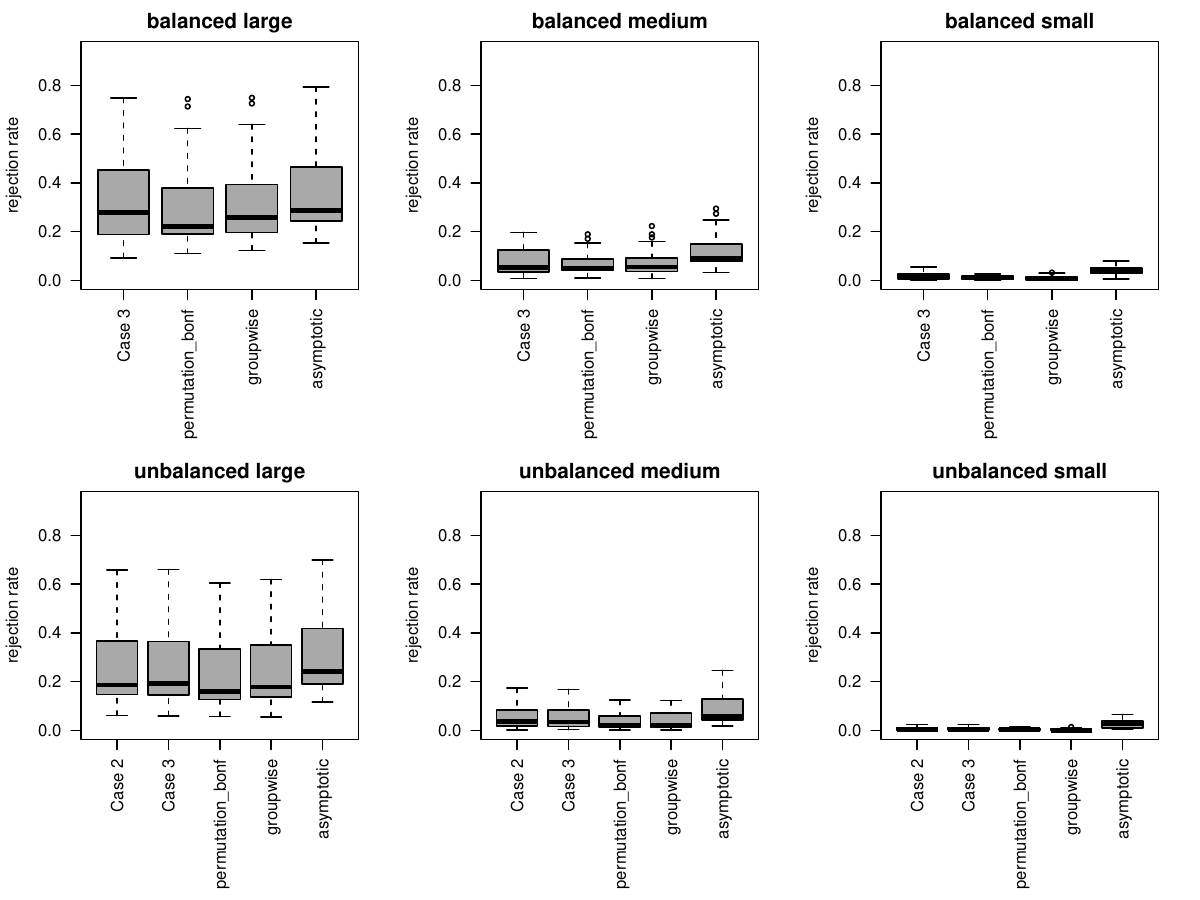}
    \caption{Empirical family-wise power for the Tukey-type contrast matrix}
    \label{fig:RMSTTukpower}
\end{figure}
\begin{figure}[h]
    \centering
    \includegraphics[width=0.8\linewidth,trim= 1mm 1mm 7mm 2mm,clip]{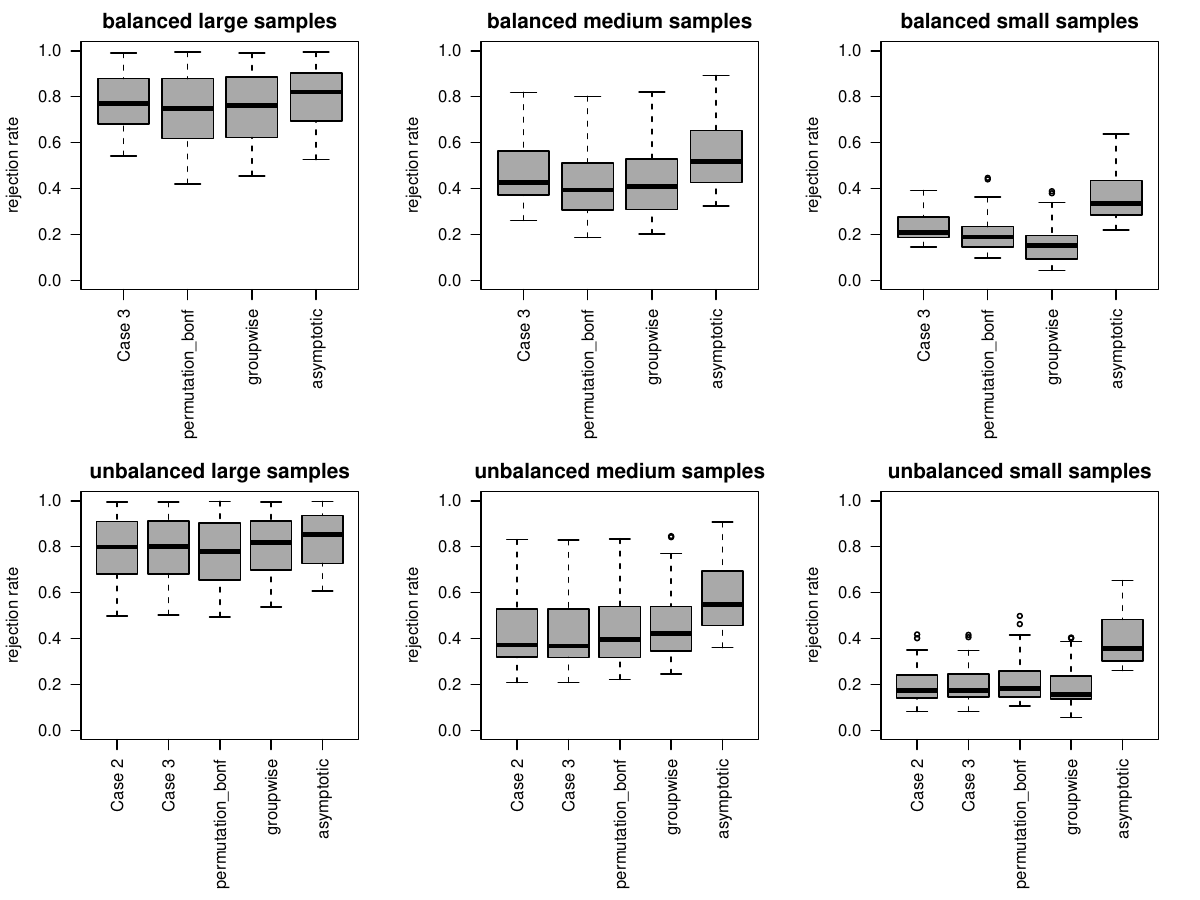}
    \caption{Empirical global power for the Tukey-type contrast matrix}
    \label{fig:RMSTTukpower2}
\end{figure}

\end{document}